\documentclass[11pt,fleqn]{article}

\usepackage{latexsym}
\usepackage{amssymb}
\usepackage{stmaryrd}
\usepackage{graphicx}
\usepackage{xcolor}
\usepackage{amsmath}
\usepackage{thmtools}
\usepackage{diagbox}
\usepackage{pgf}
\usepackage{tikz}
\usetikzlibrary{automata,positioning,arrows,fit}
\tikzstyle{morphism}=[-stealth, thick]
\tikzstyle{inclusion}=[{right hook}-stealth, thick]
\usepackage{xurl}
\usepackage{hyperref}
\hypersetup{
  colorlinks = true,
  linkcolor={blue!70!black},
  citecolor={blue!70!black},
  urlcolor={blue!70!black}
}
\usepackage{enumerate}

\newcommand{\be}{\begin{enumerate}}
\newcommand{\ee}{\end{enumerate}}
\newcommand{\bi}{\begin{itemize}}
\newcommand{\ei}{\end{itemize}}
\newcommand{\bc}{\begin{center}}
\newcommand{\ec}{\end{center}}
\newcommand{\bsp}{\begin{sloppypar}}
\newcommand{\esp}{\end{sloppypar}}

\newcommand{\sB}{{\cal B}}
\newcommand{\sC}{{\cal C}}
\newcommand{\sD}{{\cal D}}
\newcommand{\sE}{{\cal E}}

\newcommand{\sP}{{\cal P}}
\newcommand{\sQ}{{\cal Q}}

\newcommand{\sS}{{\cal S}}
\newcommand{\sT}{{\cal T}}

\newcommand{\bB}{\mathbb{B}}

\newcommand{\bN}{\mathbb{N}}

\newcommand{\sglsp}{\ }
\newcommand{\dblsp}{\ \ }
\newcommand{\textrel}[1]{\dblsp \text{#1} \dblsp}

\newcommand{\TRUE}{\mbox{{\sc t}}}
\newcommand{\FALSE}{\mbox{{\sc f}}}
 
\ifdefined \And 
\renewcommand{\And}{\wedge}
\else
\newcommand{\And}{\wedge}
\fi
\newcommand{\Or}{\vee}

\newcommand{\Forsome}{\exists}

\newcommand{\LambdaApp}{\lambda\,}

\newcommand{\SynEqual}{\equiv}

\renewcommand{\phi}{\varphi}
\newcommand{\compl}[1]{\overline{#1}}

\newcommand{\tarrow}{\rightarrow}

\newcommand{\wrestricted}{\mbox{$|$}\nolinebreak}
\newcommand{\srestricted}{\mbox{$\upharpoonright$}\nolinebreak}

\newcommand{\qzero}{${\cal Q}_0$}
\newcommand{\qzerou}{${\cal Q}^{\rm u}_{0}$}

\newcommand{\fstar}{${\rm F}^\ast$}

\ifdefined \proof
  \renewenvironment{proof}[1][\unskip]{\par\noindent{\bf Proof{#1}\dblsp}}{\hfill$\Box$}
\else
  \newenvironment{proof}[1][\unskip]{\par\noindent{\bf Proof{#1}\dblsp}}{\hfill$\Box$}
\fi

\usepackage{amssymb}
\usepackage{amsmath}
\usepackage{phonetic}
\usepackage{xcolor}

\newcommand{\mName}[1]{\mathsf{#1}}
\newcommand{\mSet}[1]{\{ #1 \}}
\newcommand{\mTuple}[1]{( #1 )}
\newcommand{\mList}[1]{[ #1 ]}
\newcommand{\mSeq}[1]{\langle #1 \rangle}
\newcommand{\mSeqlike}[1]{\mSeq{\!\mSeq{#1}\!}}
\newcommand{\mAbs}[1]{\lvert #1 \rvert}

\newcommand{\mDot}{\mathrel{.}}

\newcommand{\cBoolTy}{\omicron}
\newcommand{\cB}{\cBoolTy}
\newcommand{\cBaseTy}[1]{#1}
\newcommand{\cFunTy}[2]{({#1} \rightarrow {#2})}
\newcommand{\cFunTyX}[2]{{#1} \rightarrow {#2}}

\newcommand{\cFunTyBX}[3]{\cFunTyX {#1} {\cFunTyX {#2} {#3}}}

\newcommand{\cFunTyCX}[4]{\cFunTyX {#1} {\cFunTyX {#2} {\cFunTyX {#3} {#4}}}}
\newcommand{\cProdTy}[2]{({#1} \times {#2})}
\newcommand{\cProdTyX}[2]{{#1} \times {#2}}

\newcommand{\cVar}[2]{({#1} : {#2})}

\newcommand{\cCon}[2]{{#1}_{#2}}

\newcommand{\cEq}[2]{({#1} = {#2})}
\newcommand{\cEqX}[2]{{#1} = {#2}}
\newcommand{\cFunApp}[2]{(#1\,#2)}
\newcommand{\cFunAppX}[2]{#1\,#2}
\newcommand{\cFunAppB}[3]{(\cFunAppX {\cFunAppX {#1} {#2}} {#3})}
\newcommand{\cFunAppBX}[3]{\cFunAppX {\cFunAppX {#1} {#2}} {#3}}
\newcommand{\cFunAppC}[4]{(\cFunAppX {\cFunAppX {\cFunAppX {#1}{#2}}{#3}}{#4})}
\newcommand{\cFunAppCX}[4]{\cFunAppX {\cFunAppX {\cFunAppX {#1}{#2}}{#3}}{#4}}
\newcommand{\cFunAbs}[3]{(\lambda\, #1 : #2 \mDot #3)}
\newcommand{\cFunAbsX}[3]{\lambda\, #1 : #2 \mDot #3}
\newcommand{\cDefDes}[3]{(\mathrm{I}\, #1 : #2 \mDot #3)}
\newcommand{\cDefDesX}[3]{\mathrm{I}\, #1 : #2 \mDot #3}
\newcommand{\cOrdPair}[2]{(#1,#2)}

\newcommand{\cTPC}{T_\cB}
\newcommand{\cT}{\cTPC}
\newcommand{\cFPC}{F_\cB}
\newcommand{\cF}{\cFPC}
\newcommand{\cAndPC}{\wedge_{\cFunTyBX{\cB}{\cB}{\cB}}}
\newcommand{\cAnd}[2]{({#1} \wedge {#2})}
\newcommand{\cAndX}[2]{{#1} \wedge {#2}}

\newcommand{\cAndBX}[3]{\cAndX {#1} {\cAndX {#2} {#3}}}
\newcommand{\cImpliesPC}{\Rightarrow_{\cFunTyBX {\cB} {\cB} {\cB}}}
\newcommand{\cImplies}[2]{({#1} \Rightarrow {#2})}
\newcommand{\cImpliesX}[2]{{#1} \Rightarrow {#2}}
\newcommand{\cNegPC}{\neg_{\cFunTyX{\cB}{\cB}}}
\newcommand{\cNeg}[1]{(\neg{#1})}
\newcommand{\cNegX}[1]{\neg{#1}}
\newcommand{\cOrPC}{\vee_{\cFunTyBX{\cB}{\cB}{\cB}}}
\newcommand{\cOr}[2]{({#1} \vee {#2})}
\newcommand{\cOrX}[2]{{#1} \vee {#2}}

\newcommand{\cBin}[3]{({#1} \mathrel{#2} {#3})}
\newcommand{\cBinX}[3]{{#1} \mathrel{#2} {#3}}
\newcommand{\cBinB}[5]{({#1} \mathrel{#2} {#3} \mathrel{#4} {#5})}
\newcommand{\cBinBX}[5]{{#1} \mathrel{#2} {#3} \mathrel{#4} {#5}}
\newcommand{\cIff}[2]{({#1} \Leftrightarrow {#2})}
\newcommand{\cIffX}[2]{{#1} \Leftrightarrow {#2}}
\newcommand{\cNotEq}[2]{({#1} \not= {#2})}
\newcommand{\cNotEqX}[2]{{#1} \not= {#2}}

\newcommand{\cForall}[3]{(\forall\, #1 : #2 \mDot #3)}
\newcommand{\cForallX}[3]{\forall\, #1 : #2 \mDot #3}

\newcommand{\cForallBX}[5]{\forall\, #1 : #2,\, #3 : #4 \mDot #5}

\newcommand{\cForallCX}[7]{\forall\, #1 : #2,\, #3 : #4,\, #5 : #6 \mDot #7}
\newcommand{\cForsome}[3]{(\exists\, #1 : #2 \mDot #3)}
\newcommand{\cForsomeX}[3]{\exists\, #1 : #2 \mDot #3}

\newcommand{\cForsomeBX}[5]{\exists\, #1 : #2,\, #3 : #4 \mDot #5}

\newcommand{\cForsomeCX}[7]{\exists\, #1 : #2,\, #3 : #4,\, #5 : #6 \mDot #7}
\newcommand{\cForsomeUnique}[3]{(\exists!\, #1 : #2 \mDot #3)}

\newcommand{\cBotPC}[1]{\bot_{#1}}
\newcommand{\cEmpFunPC}[2]{\Delta_{\cFunTyX {#1} {#2}}}
\newcommand{\cIsDef}[1]{(#1{\downarrow})}
\newcommand{\cIsDefX}[1]{#1{\downarrow}}
\newcommand{\cIsUndef}[1]{(#1{\uparrow})}
\newcommand{\cIsUndefX}[1]{#1{\uparrow}}
\newcommand{\cQuasiEq}[2]{({#1} \simeq {#2})}
\newcommand{\cQuasiEqX}[2]{{#1} \simeq {#2}}
\newcommand{\cNotQuasiEq}[2]{({#1} \not\simeq {#2})}

\newcommand{\cIfThenElse}[3]{\mName{IF}(#1,#2,#3)}
\newcommand{\cIf}[3]{(#1 \mapsto #2 \mid #3)}
\newcommand{\cIfX}[3]{#1 \mapsto #2 \mid #3}

\newcommand{\cSetTy}[1]{\mSet{#1}}
\newcommand{\cIn}[2]{({#1} \in {#2})}
\newcommand{\cInX}[2]{{#1} \in {#2}}
\newcommand{\cNotIn}[2]{({#1} \not\in {#2})}
\newcommand{\cNotInX}[2]{{#1} \not\in {#2}}
\newcommand{\cSet}[3]{\mSet{{{#1} : {#2}} \mid {#3}}}
\newcommand{\cEmpSetPC}[1]{\emptyset_{\cSetTy {#1}}}
\newcommand{\cEmpSetAltPC}[1]{\mSet{\,}_{\cSetTy {#1}}}
\newcommand{\cUnivSetPC}[1]{U_{\cSetTy {#1}}}
\newcommand{\cFinSet}[2]{\textsf{{$#1$}-{$#2$}-SET}}
\newcommand{\cFinSetL}[1]{\mSet{#1}}  % separator is ","
\newcommand{\cSubseteqPC}[1]{\subseteq_{\cFunTyBX {\cSetTy #1} {\cSetTy #1} {\cB}}}

\newcommand{\cSubseteqX}[2]{\cBinX {#1} {\subseteq} {#2}}
\newcommand{\cUnionPC}[1]{\cup_{\cFunTyBX {\cSetTy #1} {\cSetTy #1} {\cSetTy #1}}}

\newcommand{\cIntersPC}[1]{\cap_{\cFunTyBX {\cSetTy #1} {\cSetTy #1} {\cSetTy #1}}}

\newcommand{\cIntersX}[2]{\cBinX {#1} {\cap} {#2}}
\newcommand{\cComplPC}[1]{\overline{\,\cdot\,}_{\cFunTyX {\cSetTy #1} {\cSetTy #1}}}
\newcommand{\cCompl}[1]{\big(\,\overline{#1}\,\big)}
\newcommand{\cComplX}[1]{\overline{#1}}
\newcommand{\cSetDiffPC}[1]{\setminus_{\cFunTyBX {\cSetTy #1} {\cSetTy #1} {\cSetTy #1}}}

\newcommand{\cTupleTyL}[1]{(#1)}  % separator is $\times$
\newcommand{\cTupleL}[1]{(#1)}    % separator is ","
\newcommand{\cFstPC}[2]{\mName{fst}_{\cFunTyX {\cProdTy {#1} {#2}} {#1}}}
\newcommand{\cSndPC}[2]{\mName{snd}_{\cFunTyX {\cProdTy {#1} {#2}} {#2}}}

\newcommand{\cIdFunPC}[1]{\mName{id}_{\cFunTyX {#1} {#1}}}
\newcommand{\cDomPC}[2]{\mName{dom}_{\cFunTyX {\cFunTy {#1} {#2}} {\cSetTy {#1}}}}
\newcommand{\cRanPC}[2]{\mName{ran}_{\cFunTyX {\cFunTy {#1} {#2}} {\cSetTy {#2}}}}
\newcommand{\cSubfuneqPC}[2]{\sqsubseteq_{\cFunTyBX {\cFunTy {#1} {#2}} {\cFunTy {#1} {#2}} {\cB}}}
\newcommand{\cFunCompPC}[3]{\circ_{\cFunTyBX {\cFunTy {#1} {#2}} {\cFunTy {#2} {#3}} {\cFunTy {#1} {#3}}}}
\newcommand{\cFunComp}[2]{({#1} \circ {#2})}
\newcommand{\cFunCompX}[2]{#1 \circ {#2}}
\newcommand{\cRestrictPC}[2]{|_{\cFunTyBX {\cFunTy {#1} {#2}} {\cSetTy {#1}} {\cFunTy {#1} {#2}}}}
\newcommand{\cRestrict}[2]{(#1 |_{#2})}
\newcommand{\cRestrictX}[2]{#1 |_{#2}}

\newcommand{\cTotal}[1]{\mName{TOTAL}(#1)}
\newcommand{\cTotalB}[1]{\mName{TOTAL2}(#1)}
\newcommand{\cSurj}[1]{\mName{SURJ}(#1)}
\newcommand{\cSurjB}[1]{\mName{SURJ2}(#1)}
\newcommand{\cInj}[1]{\mName{INJ}(#1)}
\newcommand{\cInjB}[1]{\mName{INJ2}(#1)}
\newcommand{\cBij}[1]{\mName{BIJ}(#1)}
\newcommand{\cDistinctL}[1]{\mName{DISTINCT}(#1)}  % separator is ","

\newcommand{\cFunAbsQTy}[3]{\cFunAbs {#1} {#2} {#3}}
\newcommand{\cFunAbsQTyX}[3]{\cFunAbsX {#1} {#2} {#3}}
\newcommand{\cForallQTy}[3]{\cForall {#1} {#2} {#3}}
\newcommand{\cForallQTyX}[3]{\cForallX {#1} {#2} {#3}}

\newcommand{\cForallQTyBX}[5]{\cForallBX {#1} {#2} {#3} {#4} {#5}}
\newcommand{\cForsomeQTy}[3]{\cForsome {#1} {#2} {#3}}
\newcommand{\cForsomeQTyX}[3]{\cForsomeX {#1} {#2} {#3}}

\newcommand{\cDefDesQTy}[3]{\cDefDes {#1} {#2} {#3}}
\newcommand{\cDefDesQTyX}[3]{\cDefDesX {#1} {#2} {#3}}
\newcommand{\cIsDefInQTy}[2]{({#1} \downarrow {#2})}
\newcommand{\cIsDefInQTyX}[2]{{#1} \downarrow {#2}}
\newcommand{\cIsUndefInQTy}[2]{({#1} \uparrow {#2})}

\newcommand{\cFunQTyPC}[2]{\rightarrow_{\cFunTyBX {\cSetTy {#1}} {\cSetTy {#2}} {\cSetTy {\cFunTyX {#1} {#2}}}}}
\newcommand{\cFunQTy}[2]{\cFunTy {#1} {#2}}
\newcommand{\cFunQTyX}[2]{\cFunTyX {#1} {#2}}

\newcommand{\cProdQTyPC}[2]{\times_{\cFunTyBX {\cSetTy {#1}} {\cSetTy {#2}} {\cSetTy {\cProdTyX {#1} {#2}}}}}
\newcommand{\cProdQTy}[2]{\cProdTy {#1} {#2}}
\newcommand{\cProdQTyX}[2]{\cProdTyX {#1} {#2}}

\newcommand{\cSetQTy}[1]{{\cal P}(#1)}
\newcommand{\cTotalOn}[3]{\textsf{TOTAL-ON}(#1,#2,#3)}

\newcommand{\cSurjOn}[3]{\textsf{SURJ-ON}(#1,#2,#3)}

\newcommand{\cInjOn}[2]{\textsf{INJ-ON}(#1,#2)}
\newcommand{\cInjOnB}[3]{\textsf{INJ-ON2}(#1,#2,#3)}
\newcommand{\cBijOn}[3]{\textsf{BIJ-ON}(#1,#2,#3)}
\newcommand{\cInf}[1]{\mName{INF}(#1)}
\newcommand{\cFin}[1]{\mName{FIN}(#1)}
\newcommand{\cCount}[1]{\mName{COUNT}(#1)}

\newcommand{\cSequencesPC}[2]{\mName{sequences}_{\cSetTy {\cFunTyX {#1} {#2}}}}
\newcommand{\cSeqQTy}[1]{\mSeqlike{#1}}
\newcommand{\cStreamsPC}[2]{\mName{streams}_{\cSetTy {\cFunTyX {#1} {#2}}}}
\newcommand{\cSeqInfQTy}[1]{\mSeq{#1}}
\newcommand{\cListsPC}[2]{\mName{lists}_{\cSetTy {\cFunTyX {#1} {#2}}}}
\newcommand{\cSeqFinQTy}[1]{\mList{#1}}
\newcommand{\cConsPC}[2]{\mName{cons}_{\cFunTyBX {#2} {\cFunTy {#1} {#2}} {\cFunTy {#1} {#2}}}}
\newcommand{\cCons}[2]{({#1} :: {#2})}
\newcommand{\cConsX}[2]{{#1} :: {#2}}
\newcommand{\cHdPC}[2]{\mName{hd}_{\cFunTyX {\cFunTy {#1} {#2}} {#2}}}
\newcommand{\cTlPC}[2]{\mName{tl}_{\cFunTyX {\cFunTy {#1} {#2}} {\cFunTy {#1} {#2}}}}
\newcommand{\cNilPC}[2]{\mName{nil}_{\cFunTyX {#1} {#2}}}
\newcommand{\cEmpListPC}[2]{{\mList{\;}}_{\cFunTyX {#1} {#2}}}
\newcommand{\cListL}[1]{\mList{#1}}  % separator is ","
\newcommand{\cLenPC}[2]{\mName{len}_{\cFunTyX {\cFunTy {#1} {#2}} {#1}}}
\newcommand{\cLen}[1]{\mAbs {#1}}
\newcommand{\cAppendPC}[2]{\mName{++}_{\cFunTyBX {\cFunTy {#1} {#2}} {\cFunTy {#1} {#2}} {\cFunTy {#1} {#2}}}}

\newcommand{\cNlistsPC}[2]{\mName{nlists}_{\cFunTyX {#1} {\cSetTy {\cFunTyX {#1} {#2}}}}}
\newcommand{\cSeqNFinQTy}[2]{\mList{#1}^{#2}}

\newcommand{\cProd}[4]{\Big( \prod\limits_{{#1} = {#2}}^{#3} {#4} \Big)}
\newcommand{\cProdX}[4]{\prod\limits_{{#1} = {#2}}^{#3} {#4}}

\newtheorem{thm}{Theorem}[section]

\newtheorem{lem}[thm]{Lemma}

\newtheorem{thydef}[thm]{Theory Definition}
\newtheorem{thyext}[thm]{Theory Extension}
\newtheorem{indtypethyext}[thm]{Inductive Type Theory Extension}
\newtheorem{devdef}[thm]{Development Definition}
\newtheorem{devext}[thm]{Development Extension}
\newtheorem{thytransdef}[thm]{Theory Translation Definition}
\newtheorem{thytransext}[thm]{Theory Translation Extension}
\newtheorem{devtransdef}[thm]{Development Translation Definition}
\newtheorem{devtransext}[thm]{Development Translation Extension}
\newtheorem{deftransport}[thm]{Definition Transportation}
\newtheorem{thmtransport}[thm]{Theorem Transportation}
\newtheorem{grouptransport}[thm]{Group Transportation}

\newenvironment{theory-def}[5]
{
\color{brown!90!black}
\begin{thydef}[#1]\em
\noindent
\begin{itemize} \setlength{\itemsep}{0pt}
\item[]\hspace{-3ex}\textbf{Name:} #2
\item[]\hspace{-3ex}\textbf{Base types:} #3
\item[]\hspace{-3ex}\textbf{Constants:} #4
\item[]\hspace{-3ex}\textbf{Axioms:}
\end{itemize}
 #5
\end{thydef} 
}

\newenvironment{theory-ext}[6]
{
\color{brown!90!black}
\begin{thyext}[#1]\em
\noindent
\begin{itemize} \setlength{\itemsep}{0pt}
\item[]\hspace{-3ex}\textbf{Name:} #2
\item[]\hspace{-3ex}\textbf{Extends\ } #3
\item[]\hspace{-3ex}\textbf{New base types:} #4
\item[]\hspace{-3ex}\textbf{New constants:} #5
\item[]\hspace{-3ex}\textbf{New axioms:}
\end{itemize}
#6
\end{thyext}
}

\newenvironment{ind-type-theory-ext}[5]
{
\color{brown!90!black}
\begin{indtypethyext}[#1]\em
\noindent
\begin{itemize} \setlength{\itemsep}{0pt}
\item[]\hspace{-3ex}\textbf{Name:} #2
\item[]\hspace{-3ex}\textbf{Extends\ } #3
\item[]\hspace{-3ex}\textbf{New base type:} #4
\item[]\hspace{-3ex}\textbf{Constructors:}
\end{itemize}
#5
\end{indtypethyext}
}

\newenvironment{dev-def}[4]
{
\color{brown!90!black}
\begin{devdef}[#1]\em
\noindent
\begin{itemize} \setlength{\itemsep}{0pt}
\item[]\hspace{-3ex}\textbf{Name:} #2
\item[]\hspace{-3ex}\textbf{Bottom theory:} #3
\item[]\hspace{-3ex}\textbf{Definitions and theorems:}
\end{itemize}
#4
\end{devdef}
}

\newenvironment{dev-ext}[4]
{
\color{brown!90!black}
\begin{devext}[#1]\em
\noindent
\begin{itemize} \setlength{\itemsep}{0pt}
\item[]\hspace{-3ex}\textbf{Name:} #2
\item[]\hspace{-3ex}\textbf{Extends\ } #3
\item[]\hspace{-3ex}\textbf{New definitions and theorems:}
\end{itemize}
#4
\end{devext}
}

\newenvironment{theory-trans-def}[6]
{
\color{brown!90!black}
\begin{thytransdef}[#1]\em
\noindent
\begin{itemize} \setlength{\itemsep}{0pt}
\item[]\hspace{-3ex}\textbf{Name:} #2
\item[]\hspace{-3ex}\textbf{Source theory:} #3
\item[]\hspace{-3ex}\textbf{Target theory:} #4
\item[]\hspace{-3ex}\textbf{Base type mapping:}
\end{itemize}
#5
\begin{itemize}
\item[]\hspace{-3ex}\textbf{Constant mapping:}
\end{itemize}
#6
\end{thytransdef}
}

\newenvironment{theory-trans-ext}[7]
{
\color{brown!90!black}
\begin{thytransext}[#1]\em
\noindent
\begin{itemize} \setlength{\itemsep}{0pt}
\item[]\hspace{-3ex}\textbf{Name:} #2
\item[]\hspace{-3ex}\textbf{Extends\ } #3
\item[]\hspace{-3ex}\textbf{New source theory:} #4
\item[]\hspace{-3ex}\textbf{New target theory:} #5
\item[]\hspace{-3ex}\textbf{New base type mapping:}
\end{itemize}
#6
\begin{itemize}
\item[]\hspace{-3ex}\textbf{New constant mapping:}
\end{itemize}
#7
\end{thytransext}
}

\newenvironment{dev-trans-def}[6]
{
\color{brown!90!black}
\begin{devtransdef}[#1]\em
\noindent
\begin{itemize} \setlength{\itemsep}{0pt}
\item[]\hspace{-3ex}\textbf{Name:} #2
\item[]\hspace{-3ex}\textbf{Source development:} #3
\item[]\hspace{-3ex}\textbf{Target development:} #4
\item[]\hspace{-3ex}\textbf{Base type mapping:}
\end{itemize}
#5
\begin{itemize}
\item[]\hspace{-3ex}\textbf{Constant mapping:}
\end{itemize}
#6
\end{devtransdef}
}

\newenvironment{dev-trans-ext}[6]
{
\color{brown!90!black}
\begin{devtransext}[#1]\em
\noindent
\begin{itemize} \setlength{\itemsep}{0pt}
\item[]\hspace{-3ex}\textbf{Name:} #2
\item[]\hspace{-3ex}\textbf{Extends\ } #3
\item[]\hspace{-3ex}\textbf{New source development:} #4
\item[]\hspace{-3ex}\textbf{New target development:} #5
\item[]\hspace{-3ex}\textbf{New defined constant mapping:}
\end{itemize}
#6
\end{devtransext}
}

\newenvironment{def-transport}[9]
{
\color{brown!90!black}
\begin{deftransport}[#1]\em
\noindent
\begin{itemize} \setlength{\itemsep}{0pt}
\item[]\hspace{-3ex}\textbf{Name:} #2
\item[]\hspace{-3ex}\textbf{Source development:} #3
\item[]\hspace{-3ex}\textbf{Target development:} #4
\item[]\hspace{-3ex}\textbf{Development morphism:} #5
\item[]\hspace{-3ex}\textbf{Definition:}
\end{itemize}
#6
\begin{itemize}
\item[]\hspace{-3ex}\textbf{Transported definition:}
\end{itemize}
#7
\begin{itemize} \setlength{\itemsep}{0pt}
\item[]\hspace{-3ex}\textbf{New target development:} #8
\item[]\hspace{-3ex}\textbf{New development morphism:} #9
\end{itemize}
\end{deftransport}
}

\newenvironment{thm-transport}[8]
{
\color{brown!90!black}
\begin{thmtransport}[#1]\em
\noindent
\begin{itemize} \setlength{\itemsep}{0pt}
\item[]\hspace{-3ex}\textbf{Name:} #2
\item[]\hspace{-3ex}\textbf{Source development:} #3
\item[]\hspace{-3ex}\textbf{Target development:} #4
\item[]\hspace{-3ex}\textbf{Development morphism:} #5
\item[]\hspace{-3ex}\textbf{Theorem:}
\end{itemize}
#6
\begin{itemize}
\item[]\hspace{-3ex}\textbf{Transported theorem:}
\end{itemize}
#7
\begin{itemize}
\item[]\hspace{-3ex}\textbf{New target development:} #8
\end{itemize}
\end{thmtransport}
}

\newenvironment{group-transport}[9]
{
\color{brown!90!black}
\begin{grouptransport}[#1]\em
\noindent
\begin{itemize} \setlength{\itemsep}{0pt}
\item[]\hspace{-3ex}\textbf{Name:} #2
\item[]\hspace{-3ex}\textbf{Source development:} #3
\item[]\hspace{-3ex}\textbf{Target development:} #4
\item[]\hspace{-3ex}\textbf{Development morphism:} #5
\item[]\hspace{-3ex}\textbf{Definitions and theorems:}
\end{itemize}
#6
\begin{itemize}
\item[]\hspace{-3ex}\textbf{Transported definitions and theorems:}
\end{itemize}
#7
\begin{itemize} \setlength{\itemsep}{0pt}
\item[]\hspace{-3ex}\textbf{New target development:} #8
\item[]\hspace{-3ex}\textbf{New development morphism:} #9
\end{itemize}
\end{grouptransport}
}

\newcommand{\cCatApp}[2]{({#1}{#2})}
\newcommand{\cCatAppX}[2]{{#1}{#2}}
\newcommand{\cSetCatApp}[2]{({#1}{#2})}

\newcommand{\cIterCat}[4]{\Big( \underset{{#1} = {#2}}{\overset{#3}{\mbox{\large \textsf{cat}}}} {#4} \Big)}

\newcommand{\cSetProdPC}[4]{{#1}_{\cFunTyX {\cFunTy {\cProdTy {#2} {#3}} {#4}} {\cFunTy {\cProdTy {\cSetTy {#2}} {\cSetTy {#3}}} {\cSetTy {#4}}}}}
\newcommand{\cFunCompPairPC}[3]{\circ_{\cFunTyX {\cProdTy {\cFunTy {#1} {#2}} {\cFunTy {#2} {#3}}} {\cFunTy {#1} {#3}}}}
\newcommand{\cFunAppPairPC}[2]{\bullet_{\cFunTyX {\cProdTy {\cFunTy {#1} {#2}} {#1}} {#2}}}

\newcommand{\cMonoid}[3]{\mName{MONOID}(#1,#2,#3)}
\newcommand{\cComMonoid}[3]{\textsf{COM-MONOID}(#1,#2,#3)}
\newcommand{\cMonAction}[5]{\textsf{MON-ACTION}(#1,#2,#3,#4,#5)}
\newcommand{\cMonHomom}[7]{\textsf{MON-HOMOM}(#1,#2,#3,#4,#5,#6,#7)}

\title{
{\bf {\LARGE Monoid Theory in Alonzo}}\\[0.5ex]
{\bf {\large A Little Theories Formalization in Simple Type Theory}}
}

\author{\large 
William M. Farmer%
\thanks{Address: Department of Computing and Software, McMaster
  University, 1280 Main Street West, Hamilton, Ontario L8S 4L7,
  Canada.  Email: {\texttt{wmfarmer@mcmaster.ca},
    \texttt{yankovsd@mcmaster.ca}.}}
\ and Dennis Y. Zvigelsky\footnotemark[1]
}

\date{October 31, 2025}

\begin{document}

\maketitle

\vspace{-3ex}

\begin{abstract}
\emph{Alonzo} is a practice-oriented classical higher-order version of
predicate logic that extends first-order logic and that admits
undefined expressions.  Named in honor of Alonzo Church, Alonzo is
based on Church's type theory, Church's formulation of simple type
theory.  The \emph{little theories method} is a method for formalizing
mathematical knowledge as a \emph{theory graph} consisting of
\emph{theories} as nodes and \emph{theory morphisms} as directed
edges.  The development of a mathematical topic is done in the
``little theory'' in the theory graph that has the most convenient
level of abstraction and the most convenient vocabulary, and then the
definitions and theorems produced in the development are transported,
as needed, to other theories via the theory morphisms in the theory
graph.

The purpose of this paper is to illustrate how a body of mathematical
knowledge can be formalized in Alonzo using the little theories
method.  This is done by formalizing \emph{monoid theory} --- the body
of mathematical knowledge about monoids --- in Alonzo.  Instead of
using the \emph{standard approach to formal mathematics} in which
mathematics is done with the help of a proof assistant and all details
are formally proved and mechanically checked, we employ an
\emph{alternative approach} in which everything is done within a
formal logic but proofs are not required to be fully formal.  The
standard approach focuses on \emph{certification}, while this
alternative approach focuses on \emph{communication} and
\emph{accessibility}.

\medskip


\medskip

\noindent
\textbf{MSC classification codes:} 03B16, 03B38, 68V20, 68V30.

\end {abstract}

\newpage

\setcounter{tocdepth}{1} 
\tableofcontents

\newpage

\section{Introduction}\label{sec:intro}

\emph{Formal mathematics} is mathematics done within a formal logic.
\emph{Formalization} is the act of expressing mathematical knowledge
in a formal logic.  One of the chief benefits of formal mathematics is
that a body of mathematical knowledge can be formalized as a precise,
rigorous, and highly organized~structure.  This structure records the
logical relationships between mathematical concepts and facts, how
these concepts and facts are expressed in a given theory, and how one
theory is related to another.  Since it is based on a formal logic, it
can be developed and analyzed using software.

An attractive and powerful method for organizing mathematical
knowledge is the \emph{little theories method}~\cite{FarmerEtAl92b}.
A body of mathematical knowledge is represented in the form of a
\emph{theory graph}~\cite{KohlhaseEtAl10} consisting of
\emph{theories} as nodes and \emph{theory morphisms} as directed
edges.  Each mathematical topic is developed in the ``little theory''
in the theory graph that has the most convenient level of abstraction
and the most convenient vocabulary.  Then the definitions and theorems
produced in the development are transported, as needed, from this
abstract theory to other, usually more concrete, theories in the graph
via the theory morphisms in the graph.

The \emph{standard approach to formal mathematics} focuses on
\emph{certification}: Mathematics is done with the help of a proof
assistant and all details are formally proved and mechanically
checked.  We present in Section~\ref{sec:alternative-approach} an
\emph{alternative approach to formal mathematics}, first introduced
in~\cite{Farmer23b}, that focuses on two other goals:
\emph{communication} and \emph{accessibility}.  The idea is that
everything is done within a formal logic but proofs are not required
to be fully formal and the entire development is optimized for
communication and accessibility.  We believe that formal mathematics
can be made more useful, accessible, and natural to a wider range of
mathematics practitioners --- mathematicians, computing professionals,
engineers, and scientists who use mathematics in their work --- by
pursuing this alternative approach.

The purpose of this paper is to illustrate how a body of mathematical
knowledge can be formalized in Alonzo~\cite{Farmer25}, a
practice-oriented classical higher-order logic that extends
first-order logic, using the little theories method and the
alternative approach to formal mathematics.  Named in honor of Alonzo
Church, Alonzo is based on Church's type theory~\cite{Church40},
Church's formulation of simple type theory~\cite{Farmer08}, and is
closely related to Peter Andrews' {\qzero}~\cite{Andrews02};
{\qzerou}~\cite{Farmer08a}, a version of {\qzero} with undefined
expressions; and LUTINS~\cite{Farmer90,Farmer93b,Farmer94}, the logic
of the IMPS proof assistant~\cite{FarmerEtAl93,FarmerEtAl98b}.  Unlike
traditional predicate logics, Alonzo admits partial functions and
undefined expressions in accordance with the approach employed in
mathematical practice that we call the \emph{traditional approach to
undefinedness}~\cite{Farmer04}.  Since partial functions naturally
arise from theory morphisms~\cite{Farmer94}, the little theories
method works best with a logic like Alonzo that supports partial
functions.

Alonzo has a simple syntax with a \emph{formal notation} for machines
and a \emph{compact notation} for humans that closely resembles the
notation found in mathematical practice.  The compact notation is
defined by the extensive set of \emph{notational definitions and
conventions} given in~\cite{Farmer25}.  Alonzo has two semantics, one
for mathematics based on \emph{standard models} and one for logic
based on Henkin-style \emph{general models}~\cite{Henkin50}.  By
virtue of its syntax and semantics, Alonzo is exceptionally well
suited for expressing and reasoning about mathematical ideas and for
specifying mathematical structures.  A~brief overview of the syntax
and semantics of Alonzo is given in Section~\ref{sec:alonzo}.
See~\cite{Farmer25} for a full presentation of Alonzo.

We have chosen \emph{monoid theory} --- the concepts, properties, and
facts about monoids --- as a sample body of mathematical knowledge to
formalize in Alonzo.  A \emph{monoid} is a mathematical structure
consisting of a nonempty set, an associative binary function on the
set, and a member of the set that is an identity element with respect
to the function.  Monoids are abundant in mathematics and computing.
Single-object categories are monoids.  Groups are monoids in which
every element has an inverse.  And several algebraic structures, such
as rings, fields, Boolean algebras, and vector spaces, contain monoids
as substructures.

Since a monoid is a significantly simpler algebraic structure than a
group, monoid theory lacks the rich structure of group theory.  We are
formalizing monoid theory in Alonzo, instead of group theory, since it
has just enough structure to adequately illustrate how a body of
mathematical knowledge can be formalized in Alonzo.  We will see that
employing the little theories method in the formalization of monoid
theory in Alonzo naturally leads to a robust theory graph.

Alonzo is equipped with a set of \emph{mathematical knowledge modules}
(\emph{modules} for short) for constructing various kinds of
mathematical knowledge units.  For example, it has modules for
constructing ``theories'' and ``theory morphisms''.  A \emph{language}
(or \emph{signature}) of Alonzo is a pair $L = (\sB,\sC)$, where $\sB$
is a finite set of base types and $\sC$ is a set of constants, that
specifies a set of expressions.  A \emph{theory} of Alonzo is a pair
$T = (L,\Gamma)$ where $L$ is a language called the \emph{language of
$T$} and $\Gamma$ is a set of sentences of $L$ called the \emph{axioms
of $T$}.  And a \emph{theory morphism} of Alonzo from a theory $T_1$
to a theory $T_2$ is a mapping of the expressions of $T_1$ to the
expressions of $T_2$ such that (1) base types are mapped to types and
closed quasitypes (expressions that denote sets of values), (2)
constants are mapped to closed expressions of appropriate type, and
(3)~valid sentences are mapped to valid sentences.

Alonzo also has modules for constructing ``developments'' and
``development morphisms''.  A \emph{theory development} (or
\emph{development} for short) of Alonzo is a pair $D = (T,\Xi)$ where
$T$ is a theory and $\Xi$ is a (possibly empty) sequence of
definitions and theorems presented, respectively, as definition and
theorem packages (see~\cite[Section 12.1]{Farmer25}).  $T$ is called
the \emph{bottom theory} of $D$, and $T'$, the extension of $T$
obtained by adding the definitions in $\Xi$ to~$T$, is called the
\emph{top theory} of $D$.  We say that $D$ is a \emph{development
of~$T$}.  A \emph{development morphism} from a development $D_1$ to a
development $D_2$ is a partial mapping from the expressions of $D_1$
to the expressions of $D_2$ that restricts to a theory morphism from
the bottom theory of $D_1$ to the bottom theory of $D_2$ and that
canonically extends to a theory morphism from the top theory of $D_1$
to the top theory of $D_2$ (see~\cite[Section 14.4.1]{Farmer25}).
Theories and theory morphisms are special cases of developments and
development morphisms, respectively, since we identify a theory $T$
with the trivial development $(T,\mList{\,})$.

The modules for constructing developments and development morphisms
provide the means to represent knowledge in the form of a
\emph{development graph}, a richer kind of theory graph, in which the
nodes are developments and the directed edges are development
morphisms.  Alonzo includes modules for transporting definitions and
theorems from one development to another via development morphisms.
The design of Alonzo's module system is inspired by the IMPS
implementation of the little theories method
\cite{FarmerEtAl92b,FarmerEtAl93,FarmerEtAl98b}.

The formalization of monoid theory presented in this paper exemplifies
an \emph{alternative approach to formal mathematics}.  We validate
the definitions and theorems in a development using traditional
(nonformal) mathematical proof.  However, we extensively use the
axioms, rules of inference, and metatheorems of $\mathfrak{A}$ --- the
formal proof system for Alonzo presented in~\cite{Farmer25} which is
derived from Andrews' proof system for {\qzero}~\cite{Andrews02} ---
in these traditional proofs.  The proofs are not included in the
modules used to construct developments.  Instead, they are given
separately in Appendix~\ref{app:validation}.

We produced the formalization of monoid theory with just a minimal
amount of software support.  We used the set of LaTeX macros and
environments for Alonzo given in~\cite{Farmer23a} plus a few macros
created specifically for this paper.  The macros are for presenting
Alonzo types and expressions in both the formal and compact notations.
The environments are for presenting Alonzo mathematical knowledge
modules.  The Alonzo modules are printed in brown~color.

The overarching goal of this paper is to demonstrate that, using the
little theories method and the alternative approach to formal
mathematics, mathematical knowledge can be very effectively formalized
in a version of simple type theory like Alonzo.  Specifically, we
want to show the following:

\be

  \item How the little theories method can be used to organize
    mathematical knowledge so that clarity is maximized and redundancy
    is minimized.

  \item How formal libraries of mathematical knowledge that prioritize
    communication over certification can be built using the
    alternative approach to formal mathematics with tools that are
    much simpler to learn and use than a proof assistant.

  \item How Alonzo is exceptionally well suited for expressing and
    reasoning about mathematical ideas and for specifying mathematical
    structures in a direct and natural manner.

\ee

The paper is organized as follows.  We present in
Section~\ref{sec:alternative-approach} the alternative approach to
formal mathematics and argue that this kind of approach can better
serve the average mathematics practitioner than the standard approach.
Section~\ref{sec:alonzo} gives a brief presentation of the syntax and
semantics of Alonzo.
Sections~\ref{sec:monoids}--\ref{sec:monoids-with-reals} present
developments of theories of monoids, commutative monoids,
transformation monoids, monoid actions, monoid homomorphisms, and
monoids over real number arithmetic plus some supporting developments.
These developments have been constructed to be illustrative; they are
not intended to be complete in any sense.
Sections~\ref{sec:monoids}--\ref{sec:monoids-with-reals} also present
various development morphisms that are used to transport definitions
and theorems from one development to another.
Section~\ref{sec:strings} shows how our formalization of monoid theory
can support a theory of strings.  Related work is discussed in
Section~\ref{sec:related-work}.  The paper concludes in
Section~\ref{sec:conc} with a summary and some final remarks.  The
definitions and theorems of the developments we have constructed are
validated by traditional mathematical proofs presented in
Appendix~\ref{app:validation}.  Appendix~\ref{app:misc-thms} contains
some miscellaneous theorems needed for the proofs in
Appendix~\ref{app:validation}.

\newpage

\section{Alternative Approach to Formal Mathematics}
\label{sec:alternative-approach}

A \emph{formal logic} (\emph{logic} for short) is a \emph{family of
languages} such that:

\be

  \item The languages of the logic have a \emph{common precise syntax}.

  \item The languages of the logic have a \emph{common precise
  semantics with a notion of logical consequence}.

  \item There is a \emph{sound formal proof system} for the logic in
    which proofs can be syntactically constructed.

\ee 
Examples of formal logics for mathematics are the various versions of
first-order logic, set theory, simple type theory, and dependent type
theory.

There are five big benefits of formal mathematics, i.e., doing
mathematics within a formal logic.

First, \emph{mathematics can be done with greater rigor}.  All
mathematical ideas are expressed and reasoned about in a theory $T$ of
a formal logic.  Mathematical concepts and statements are expressed as
expressions and sentences of the language of $T$.  All of these
expressions and sentences have a precise, unambiguous meaning.  The
assumptions underlying the reasoning about the mathematical ideas are
made explicit as axioms of the theory.  The theorems of theory are
precisely defined as the logical consequences of the axioms of the
theory.  And, finally, the theory is constructed so that it contains
only the vocabulary and assumptions that are needed for the task at
hand; irrelevant details are abstracted away.

Second, \emph{conceptual errors can be systematically discovered}.  In
formal mathematics, all concepts and statements must be expressed in a
language of a formal logic that has a precise semantics.  The process
of expressing mathematical ideas in a formal logic naturally leads to
many conceptual errors being caught similarly to how type errors are
caught in a modern programming language by type checking.  Thus
conceptual errors can be discovered systematically in formal
mathematics in a way that is largely not possible in traditional
mathematics.  As a result, formal mathematics often yields a deeper
understanding of the mathematics being explored than traditional
mathematics.

Third, \emph{mathematics can be done with software support}.  Since
the languages of a formal logic have a precise common syntax, the
expressions and sentences of a language can be represented as data
structures.  The expressions and sentences can then be analyzed,
manipulated, and processed via their representations as data
structures.  This, in turn, enables the study, discovery,
communication, and certification of mathematics to be done with the
aid of software.  Since the languages also have a precise common
semantics, there is a precise basis for verifying the correctness of
this software.

Fourth, \emph{results can be mechanically checked}.  Formal proofs can
be represented as data structures, and software can be used to check
that one of these data structures represents an actual proof in the
formal proof system of the logic.  Software can also be used to help
construct the formal proofs.  Since the software needed to check the
correctness of the formal proofs is often very simple and easily
verified itself, it is possible to verify the correctness of the
formal proofs with a very high level of assurance.

Fifth, \emph{we can regard mathematical knowledge as a formal
structure consisting of a network of interconnected theories}.  A
library of mathematical knowledge that represents this formal
structure can be built by creating theories, defining new concepts,
stating and proving theorems, and connecting one theory to another
with theory morphisms that map the theorems of one theory to the
theorems of another theory.  The knowledge embodied in a structured
library of this kind can be studied, managed, searched, and presented
using software.

The benefits of formal mathematics are huge.  Greater rigor and
discovering conceptual errors have been principal goals of
mathematicians for thousands of years.  Software support can greatly
extend the reach and productivity of mathematics practitioners.
Mechanically checked results can drive mathematics forward in areas
where the ideas are poorly understood (often due to their novelty) or
highly complex.  And mathematical knowledge as a formal structure can
enable the techniques and tools of mathematics and computing to be
applied to mathematical knowledge itself.

The standard approach to formal mathematics, in which mathematics is
done with the help of a proof assistant and all details are formally
proved and mechanically checked, has three major strengths:

\be

  \item It achieves all five benefits of formal mathematics mentioned
    above.

  \item All theorems are verified by machine-checked formal proofs.
    Thus there is a very high level of assurance that the results produced
    are correct.

  \item There are several powerful proof assistants available, such as
    HOL~\cite{GordonMelham93}, HOL Light~\cite{Harrison09}, ,
    Isabelle/HOL~\cite{Paulson94}, Lean~\cite{deMouraEtAl15},
    Metamath/ZFC~\cite{MetamathWebSite},
    Mizar~\cite{NaumowiczKornilowicz09}, and Rocq (formerly
    Coq)~\cite{RocqWebSite}, that support the approach.

\ee

\newpage

\noindent
It also has two important weaknesses:

\be

  \item It prioritizes certification over communication.  For the
    average mathematics practitioner, communicating mathematical ideas
    is usually much more important than certifying mathematical
    results when the mathematics is well understood.

  \item It is not accessible to the great majority of mathematics
    practitioners.  Having to learn a strange logic and work with a
    complex proof assistant that utilizes unfamiliar ways of
    expressing and reasoning about mathematics is very often a bridge
    too far for the average mathematics practitioner.

\ee

We strongly believe, as an alternative to the standard approach, an
approach to formal mathematics is needed that focuses on two goals,
communication and accessibility, the weaknesses of the standard
approach.  To achieve these goals the alternative approach should
satisfy the following requirements:

\be

  \item[R1.] \emph{The underlying logic is fully formal and supports
  standard mathematical practice.}  Supporting mathematical practice
    makes the logic easier to learn and use and makes formalization a
    more natural process.

  \item[R2.] \emph{Proofs can be traditional, formal, or a combination
  of the two.}  This flexibility in how proofs are written enables
    proofs to be a vehicle for communication as well as certification.

  \item[R3.] \emph{There is support for organizing mathematical
  knowledge using the little theories method.}  This enables
    mathematical knowledge to be formalized to maximize clarity and
    minimize redundancy.

  \item[R4.] \emph{There are several levels of supporting software.}
    The levels can range from just LaTeX support to a full proof
    assistant.  The user can thus choose the level of software support
    they want to have and the level of investment in learning the
    software they want to make.

\ee

The alternative approach can achieve all five benefits of formal
mathematics mentioned above, but it cannot achieve the same level of
assurance as the standard approach that the results produced are
correct.  This is because the alternative approach prioritizes
communication and accessibility over certification.  Since most
mathematics practitioners are usually more concerned about
communication and accessibility than certification, the alternative
approach is on average a better approach to formal mathematics than
the standard approach.  This is particularly true for applications
that involve well-understood mathematics, the kind of mathematics that
arises in mathematics education and routine applications.  However,
when the certification of results is the most important concern, the
standard approach will often be a better choice than the alternative
approach.

%%, especially for poorly understood, novel, or highly complex mathematics.

\iffalse
The alternative approach can achieve all four benefits of formal
mathematics mentioned above, but it cannot achieve the same level of
assurance as the standard approach that the results produced are
correct.  This is because the alternative approach prioritizes
communication and accessibility over certification.  Thus, compared to
the standard approach, the alternative approach is not as well suited
for poorly understood, novel, or highly complex mathematics, but it is
quite well suited for well-understood mathematics, the kind of
mathematics that arises in mathematics education and routine
applications.
\fi

This paper employs an implementation of the alternative approach based
on Alonzo that satisfies the first three requirements and partially
satisfies the fourth requirement.  Alonzo is a form of predicate
logic, which is widely familiar to mathematics practitioners.
Moreover, it supports the reasoning instruments that are most common
in mathematical practice including functions, sets, tuples, and lists;
mathematical structures; higher-order and restricted quantification;
definite description; theories and theory morphisms; definitional and
other kinds of conservative extensions; inductive sets; notational
definitions and conventions, and undefined expressions.  Thus Alonzo
satisfies R1 as well or better than almost any other logic.

R2 is satisfied by our implementation of the alternative approach
since proofs can be traditional or formal.  Thus communication can be
prioritized over certification in proofs when the mathematics is well
understood.  In this paper, all the proofs are traditional, but some
make use of the axioms, rules of inference, and metatheorems of
$\mathfrak{A}$, the proof system of Alonzo.  

R3 is satisfied since Alonzo is equipped with a module system for
organizing mathematical knowledge using the little theories method.

Our implementation of the alternative approach provides only the
simplest level of software support: LaTeX macros for presenting Alonzo
types and expressions and LaTeX environments for presenting Alonzo
modules.  Other levels of software support are possible; see the
discussion in Chapter~16 of~\cite{Farmer25}.  Alonzo has not been
implemented in a proof assistant, but since it is closely related to
LUTINS~\cite{Farmer90,Farmer93b,Farmer94}, the logic of the IMPS proof
assistant~\cite{FarmerEtAl93,FarmerEtAl98b}, it could be implemented
in much the same way that LUTINS is implemented in IMPS.  Thus R4 is
only partially satisfied now, but it could be fully satisfied with the
addition of more levels of software support.

The great majority of mathematics practitioners --- including
mathematicians --- are much more interested in communicating
mathematical ideas than in formally certifying mathematical results.
Hence, the alternative approach --- with support for standard
mathematical practice, traditional proofs, the little theories method,
and several levels of software --- is likely to serve the needs of the
average mathematics practitioner much better than the standard
approach.  This is especially true when the mathematical knowledge
involved is well understood (such as monoid theory) and certification
via traditional proof is adequate for the purpose at~hand.

In summary, we believe that the alternative approach is not a
replacement for the standard approach, but it would be more useful,
accessible, and natural than the standard approach for the vast
majority of mathematics practitioners.

\section{Alonzo}\label{sec:alonzo}

Alonzo is fully presented in~\cite{Farmer25}.  Due to space
limitations, we cannot duplicate the entire presentation of Alonzo in
this paper.  Ideally, the reader should be familiar with the syntax
and semantics of Alonzo presented in Chapters 4--7; the proof system
for Alonzo presented in Chapter 8 and Appendices A--C; the tables of
notational definitions found in Chapters 4, 6, 11, and 13; the
notational conventions presented in Chapters 4 and 6; and the various
kinds of (mathematical knowledge) modules of Alonzo presented in
Chapters 9, 10, 12, and 14.  However, we will give in this section a
brief presentation of the syntax and semantics of Alonzo with most of
the text taken from Chapters 4--6 of~\cite{Farmer25}.

\subsection{Syntax}

The syntax of Alonzo consists of ``types'' that denote nonempty sets
of values and ``expressions'' that either denote values (when they are
defined) or denote nothing at all (when they are undefined).  We
present the syntax of Alonzo types and expressions with the compact
notation, an ``external'' syntax intended for humans.  The reader is
referred to~\cite{Farmer25} for the formal syntax, an ``internal''
syntax intended for machines.  The compact notation for types and
expressions is given below.  Additional compact notation is introduced
using \emph{notational definitions} and \emph{notational conventions}.
A \emph{notational definition} has the form \[A \textrel{stands for}
B,\] where $A$ and $B$ are notations that present types or
expressions; it defines $A$ to be an alternate --- and usually more
compact, convenient, or standard --- notation for presenting the type
or expression that $B$ presents.  The meaning of $A$ is the meaning of
$B$.  The notational definitions are given in tables with boxes
surrounding the definitions, and the notational conventions are
assigned names of the form ``Notational Convention $n$''.

\bsp Let $\sS_{\sf bt}$\index{5s2@$\sS_{\sf bt}$}, $\sS_{\sf
  var}$\index{5s4@$\sS_{\sf var}$}, $\sS_{\sf
  con}$\index{5s3@$\sS_{\sf con}$} be fixed countably infinite sets of
symbols that will serve as names of base types, variables, and
constants, respectively.  We assume that $\sS_{\sf bt}$ contains the
symbols $A,B,C\ldots,X,Y,Z,$ etc., $\sS_{\sf var}$ contains the
symbols $a,b,c\ldots,x,y,z,$ etc., and $\sS_{\sf con}$ contains the
symbols $A,B,C\ldots,X,Y,Z,$ etc., numeric symbols, nonalphanumeric
symbols, and words in lowercase sans sarif font.\footnote{An
expression like ``$u,v,w,$ etc.''  means the set of symbols that
includes $u$, $v$, and $w$, and all possible annotated forms of $u$,
$v$, and $w$ such as $u'$, $v_1$, and~$\widetilde{w}$.}  We will
employ the following syntactic variables for these symbols as well as
types and expressions which are defined just below:
\be

  \item $\mathbf{a}, \mathbf{b}$, etc.\ range over $\sS_{\sf bt}$.

  \item $\mathbf{f}, \mathbf{g}, \mathbf{h}, \mathbf{i}, \mathbf{j},
    \mathbf{k}, \mathbf{m}, \mathbf{n}, \mathbf{u}, \mathbf{v},
    \mathbf{w}, \mathbf{x}, \mathbf{y}, \mathbf{z}$, etc.\ range over
    $\sS_{\sf var}$.

  \item $\mathbf{c}, \mathbf{d}$, etc.\ range over $\sS_{\sf con}$.

  \item $\alpha,\beta,\gamma,\delta$, etc.\ range over types.

  \item $\mathbf{A}_\alpha, \mathbf{B}_\alpha,
    \mathbf{C}_\alpha,\ldots,\mathbf{X}_\alpha, \mathbf{Y}_\alpha,
    \mathbf{Z}_\alpha$, etc.\ range over expressions of type $\alpha$.

\ee
\esp

A \emph{type} of Alonzo is a string of symbols defined inductively by
the following formation rules:
\bi

  \item[T1.] \emph{Type of truth values}: $\cBoolTy$ is a type.

  \item[T2.] \emph{Base type}: $\cBaseTy {\mathbf{a}}$ is a type.

  \item[T3.] \emph{Function type}: $\cFunTy {\alpha} {\beta}$ is a
    type.

  \item[T4.] \emph{Product type}: $\cProdTy {\alpha} {\beta}$ is a
    type.

\ei
Let $\sT$ denote the set of types of Alonzo.  We assume $\cBoolTy
\not\in \sS_{\sf bt}$.  

When there is no loss of meaning, matching pairs of parentheses in the
compact notation for types may be omitted (Notational Convention 1).
We assume that function type formation associates to the right so
that, e.g., a type of the
form \[\cFunTy{\alpha}{\cFunTy{\beta}{\gamma}}\] may be written more
simply as $\cFunTyBX{\alpha}{\beta}{\gamma}$ (Notational Convention
2).

A type $\alpha$ denotes a nonempty set $D_\alpha$ of values.  $\cB$
denotes the set $D_\cB = \bB$ of the Boolean (truth) values $\FALSE$
and $\TRUE$.  $\cFunTy{\alpha}{\beta}$ denotes some set
$D_{\cFunTyX{\alpha}{\beta}}$ of (partial and total) functions from
$D_\alpha$ to $D_\beta$.  $\cProdTy{\alpha}{\beta}$ denotes the
Cartesian product $D_{\cProdTyX{\alpha}{\beta}} = D_\alpha \times
D_\beta$.  We will use base types to denote the base domains of
mathematical structures.

An \emph{expression of type $\alpha$} of Alonzo is a string of symbols
defined inductively by the following formation rules:
\bi

  \item[E1.] \emph{Variable}: $\cVar {\mathbf{x}} {\alpha}$ is an
    expression of type~$\alpha$.

  \item[E2.] \emph{Constant}: $\cCon {\mathbf{c}} {\alpha}$ is an
    expression of type~$\alpha$.

  \item[E3.] \emph{Equality}: $\cEq {\mathbf{A}_\alpha}
    {\mathbf{B}_\alpha}$ is an expression of type~$\cBoolTy$.

  \item[E4.] \emph{Function application}: $\cFunApp
    {\mathbf{F}_{\cFunTy {\alpha} {\beta}}} {\mathbf{A}_\alpha}$ is an
    expression of type~$\beta$.

  \item[E5.] \emph{Function abstraction}: $\cFunAbs {\mathbf{x}}
    {\alpha} {\mathbf{B}_\beta}$ is an expression of type~$\cFunTy
    {\alpha} {\beta}$.

  \item[E6.] \emph{Definite description}: $\cDefDes {\mathbf{x}}
    {\alpha} {\mathbf{A}_{\cBoolTy}}$ is an expression of
    type~$\alpha$ where $\alpha \not= \cBoolTy$.

  \item[E7.] \emph{Ordered pair}: $\cOrdPair {\mathbf{A}_\alpha}
    {\mathbf{B}_\beta}$ is an expression of type~$\cProdTy {\alpha}
    {\beta}$.

\ei
Let $\sE$ denote the set of expressions of Alonzo.  A \emph{formula}
is an expression of type $\cBoolTy$, and a \emph{sentence} is a closed
formula.

When there is no loss of meaning, matching pairs of parentheses in
expressions may be omitted (Notational Convention 3).  We assume that
function application formation associates to the left so that, e.g.,
an expression of the form
$\cFunApp{\cFunApp{\mathbf{G}_{\cFunTyBX{\alpha}{\beta}{\gamma}}}
  {\mathbf{A}_\alpha}}{\mathbf{B}_\beta}$ may be written more simply
as $\cFunAppBX {\mathbf{G}_{\cFunTyBX{\alpha}{\beta}{\gamma}}}
{\mathbf{A}_\alpha} {\mathbf{B}_\beta}$ (Notational Convention 4).
When the type $\alpha$ of a constant $\cCon {\mathbf{c}} {\alpha}$ is
known from the context of the constant, we will very often write the
constant as simply $\mathbf{c}$ (Notational Convention 5).  A variable
$\cVar{\mathbf{x}}{\alpha}$ occurring in the body $\mathbf{B}_\beta$
of $\cFunAbsX{\mathbf{x}}{\alpha}{\mathbf{B}_\beta}$ or in the body
$\mathbf{A}_\cB$ of $\cDefDesX{\mathbf{x}}{\alpha}{\mathbf{A}_\cB}$
may be written as just $\mathbf{x}$ if there is no resulting ambiguity
(Notational Convention 6).  So, for example,
$\cFunAbsX{\mathbf{x}}{\alpha}{\cVar{\mathbf{x}}{\alpha}}$ may be
written more simply as $\cFunAbsX{\mathbf{x}}{\alpha}{\mathbf{x}}$.
We will employ this convention for the other variable binders of
Alonzo introduced later by notational definitions (Notational
Convention~7).  A variable $\cVar{\mathbf{x}}{\alpha}$ occurring in
$\mathbf{B}_\beta$ may be written as just $\mathbf{x}$ if the type
$\alpha$ is known from the context of the occurrence of
$\cVar{\mathbf{x}}{\alpha}$ in $\mathbf{B}_\beta$ (Notational
Convention~8).  For example, $\cEqX {\mathbf{A}_\alpha} {\cVar
  {\textbf{x}} {\alpha}}$ may be written as $\cEqX {\mathbf{A}_\alpha}
{\textbf{x}}$.

An expression of type $\alpha$ is always defined if $\alpha = \cB$ and
may be either defined or undefined if $\alpha \not= \cB$.  If defined,
it denotes a value in $D_\alpha$, the denotation of $\alpha$.  If
undefined, it denotes nothing at all.  We will use constants to denote
the distinguished values of mathematical structures.

As previously defined, a \emph{language} (or \emph{signature}) of
Alonzo is a pair $L = (\sB,\sC)$ where $\sB$ is a finite set of base
types and $\sC$ is a set of constants $\cCon {\mathbf{c}} {\alpha}$
where each base type occurring in $\alpha$ is a member of $\sB$.  A
type $\alpha$ is a \emph{type of $L$} if all the base types occurring
in $\alpha$ are members of $\sB$, and an expression
$\mathbf{A}_\alpha$ is an \emph{expression of $L$} if all the base
types occurring in $\mathbf{A}_\alpha$ are members of $\sB$ and all
the constants occurring in $\mathbf{A}_\alpha$ are members of $\sC$.
Let $\sT(L) \subseteq \sT$ denote the set of types of $L$ and $\sE(L)
\subseteq \sE$ denote the set of expressions of $L$.  Notice that
$\sB$ and $\sC$ may be empty, but $\sT(L)$ and $\sE(L)$ are always
nonempty since $\cB \in \sT(L)$.

\subsection{Semantics}

Let $L = (\sB,\sC)$ be a language of Alonzo.  We will now define the
semantics of $L$.

A \emph{frame} for $L$ is a collection $\sD = \mSet{D_\alpha \mid
  \alpha \in \sT(L)}$ of nonempty domains (sets) of values such that:
\bi

  \item[F1.] \emph{Domain of truth values}: $D_\cB = \bB =
    \mSet{\FALSE,\TRUE}$.

  \item[F2.] \emph{Predicate domain}: $D_{\cFunTyX{\alpha}{\cB}}$ is a
    set of \emph{some} total functions from $D_\alpha$ to $D_\cB$ for
    $\alpha \in \sT(L)$.

  \item[F3.] \emph{Function domain}: $D_{\cFunTyX{\alpha}{\beta}}$ is
    a set of \emph{some} partial and total functions from $D_\alpha$
    to $D_\beta$ for $\alpha, \beta \in \sT(L)$ with $\beta \not=
    \cB$.

  \item[F4.] \emph{Product domain}: $D_{\cProdTyX{\alpha}{\beta}} =
    D_\alpha \times D_\beta$ for $\alpha,\beta \in
    \sT(L)$.

\ei
A predicate domain $D_{\cFunTyX{\alpha}{\cB}}$ is \emph{full} if it is
the set of \emph{all} total functions from $D_\alpha$ to $D_\cB$, and
a function domain $D_{\cFunTyX{\alpha}{\beta}}$ with $\beta \not= \cB$
is \emph{full} if it is the set of \emph{all} partial and total
functions from $D_\alpha$ to $D_\beta$.  The frame is \emph{full} if
$D_{\cFunTyX{\alpha}{\beta}}$ is full for all $\alpha,\beta \in
\sT(L)$.  Notice that the only restriction on a \emph{base domain},
i.e., $D_{\mathbf{a}}$ for some $\mathbf{a} \in \sB$, is that it is
nonempty and that the frame is completely determined by its base
domains when the frame is full.  An \emph{interpretation} of $L$ is a
pair $M = (\sD,I)$ where $\sD = \mSet{D_\alpha \mid \alpha \in
  \sT(L)}$ is a frame for $L$ and $I$~is an \emph{interpretation
function} that maps each constant in $\sC$ of type $\alpha$ to an
element of $D_\alpha$.  Notice that
\[(\mSet{D_{\mathbf{a}} \mid \mathbf{a} \in \sB}, 
\mSet{I(\cCon {\mathbf{c}} {\alpha}) \mid \cCon {\mathbf{c}} {\alpha}
  \in \sC})\] is a mathematical structure.  Hence an interpretation of
a language \emph{defines} (1) a mathematical structure and (2) a
mapping of the base types and constants of the language to the base
domains and distinguished values, respectively, of the mathematical
structure.

Let $\sD = \mSet{D_\alpha \mid \alpha \in \sT(L)}$ be a frame for $L$.
An \emph{assignment into} $\sD$ is a function $\phi$ whose domain is
the set of variables of $L$ such that $\phi(\cVar{\mathbf{x}}{\alpha})
\in D_\alpha$ for each variable $\cVar{\mathbf{x}}{\alpha}$ of $L$.
Given an assignment $\phi$, a variable $\cVar{\mathbf{x}}{\alpha}$ of
$L$, and $d \in D_\alpha$, let $\phi[\cVar{\mathbf{x}}{\alpha} \mapsto
  d]$\index{23a@$\phi[\cVar{\mathbf{x}}{\alpha} \mapsto d]$} be the
assignment $\psi$ in $\sD$ such that $\psi(\cVar{\mathbf{x}}{\alpha})
= d$ and $\psi(\cVar{\mathbf{y}}{\beta}) =
\phi(\cVar{\mathbf{y}}{\beta})$ for all variables
$\cVar{\mathbf{y}}{\beta}$ of $L$ distinct from
$\cVar{\mathbf{x}}{\alpha}$.  Given an interpretation $M$ of $L$, let
$\mName{assign}(M)$ be the set of assignments into the frame of $M$.

Let $\sD = \mSet{D_\alpha \mid \alpha \in \sT(L)}$ be a frame for $L$
and $M = (\sD, I)$ be an interpretation of $L$.  $M$ is a
\emph{general model} of $L$ if there is a partial binary
\emph{valuation function} $V^{M}$ such that, for all assignments $\phi
\in \mName{assign}(M)$ and expressions $\mathbf{C}_\gamma$ of $L$, (1)
either $V^{M}_{\phi}(\mathbf{C}_\gamma) \in D_\gamma$ or
$V^{M}_{\phi}(\mathbf{C}_\gamma)$ is undefined\footnote{We write
$V^{M}_{\phi}(\mathbf{C}_\gamma)$ instead of
$V^{M}(\phi,\mathbf{C}_\gamma)$.} and (2) each of the following
conditions is satisfied:
\bi

  \item[V1.] $V^{M}_{\phi}(\cVar{\mathbf{x}}{\alpha}) =
    \phi(\cVar{\mathbf{x}}{\alpha})$.

  \item[V2.] $V^{M}_{\phi}(\cCon {\mathbf{c}} {\alpha}) = I(\cCon
    {\mathbf{c}} {\alpha})$.

  \item[V3.] $V^{M}_{\phi}(\cEqX {\mathbf{A}_\alpha}
    {\mathbf{B}_\alpha}) = \TRUE$ if $V^{M}_{\phi}(\mathbf{A}_\alpha)$
    is defined, $V^{M}_{\phi}(\mathbf{B}_\alpha)$ is defined, and
    $V^{M}_{\phi}(\mathbf{A}_\alpha) =
    V^{M}_{\phi}(\mathbf{B}_\alpha)$.  Otherwise, $V^{M}_{\phi}(\cEqX
    {\mathbf{A}_\alpha} {\mathbf{B}_\alpha}) = \FALSE$.

  \item[V4.] \bsp $V^{M}_{\phi}(\cFunAppX
    {\mathbf{F}_{\cFunTyX{\alpha}{\beta}}} {\mathbf{A}_\alpha}) =
    V^{M}_{\phi}(\mathbf{F}_{\cFunTyX{\alpha}{\beta}})
    (V^{M}_{\phi}(\mathbf{A}_\alpha))$ if
    $V^{M}_{\phi}(\mathbf{F}_{\cFunTyX{\alpha}{\beta}})$ is defined,
    $V^{M}_{\phi}(\mathbf{A}_\alpha)$ is defined, and
    $V^{M}_{\phi}(\mathbf{F}_{\cFunTyX{\alpha}{\beta}})$ is defined at
    $V^{M}_{\phi}(\mathbf{A}_\alpha)$.  Otherwise,
    $V^{M}_{\phi}(\cFunAppX {\mathbf{F}_{\cFunTyX{\alpha}{\beta}}}
    {\mathbf{A}_\alpha}) = \FALSE$ if $\beta = \cB$ and
    $V^{M}_{\phi}(\cFunAppX {\mathbf{F}_{\cFunTyX{\alpha}{\beta}}}
    {\mathbf{A}_\alpha})$ is undefined if $\beta \not= \cB$. \esp

  \item[V5.] $V^{M}_{\phi}(\cFunAbsX {\mathbf{x}} {\alpha}
    {\mathbf{B}_\beta})$ is the (partial or total) function $f \in
    D_{\cFunTyX{\alpha}{\beta}}$ such that, for each $d \in D_\alpha$,
    $f(d) = V^{M}_{\phi[\cVar{\mathbf{x}}{\alpha} \mapsto  d]}(\mathbf{B}_\beta)$ 
    if $V^{M}_{\phi[\cVar{\mathbf{x}}{\alpha} \mapsto d]}(\mathbf{B}_\beta)$ 
    is defined and $f(d)$ is undefined if
    $V^{M}_{\phi[\cVar{\mathbf{x}}{\alpha} \mapsto d]}
    (\mathbf{B}_\beta)$ is undefined.

  \item[V6.] $V^{M}_{\phi}(\cDefDesX {\mathbf{x}} {\alpha} {\mathbf{A}_\cB})$ 
    is the $d \in D_\alpha$ such that
    $V^{M}_{\phi[\cVar{\mathbf{x}}{\alpha} \mapsto d]}(\mathbf{A}_\cB) =
    \TRUE$ if there is exactly one such $d$.  Otherwise,
    $V^{M}_{\phi}(\cDefDesX {\mathbf{x}} {\alpha} {\mathbf{A}_\cB})$
    is undefined.

  \item[V7.] $V^{M}_{\phi}(\cOrdPair {\mathbf{A}_\alpha} 
    {\mathbf{B}_\beta}) =
    (V^{M}_{\phi}(\mathbf{A}_\alpha),V^{M}_{\phi}(\mathbf{B}_\beta))$
    if $V^{M}_{\phi}(\mathbf{A}_\alpha)$ and
    $V^{M}_{\phi}(\mathbf{B}_\beta)$ are defined.  Otherwise,
    $V^{M}_{\phi}(\cOrdPair {\mathbf{A}_\alpha} 
    {\mathbf{B}_\beta})$ is undefined.

\ei 
$V^{M}$ is unique when it exists.  $V^{M}_{\phi}(\mathbf{C}_\gamma)$
is called the \emph{value of $\mathbf{C}_\gamma$ in $M$ with respect
to $\phi$} when $V^{M}_{\phi}(\mathbf{C}_\gamma)$ is defined.
$\mathbf{C}_\gamma$~is said to have no value in $M$ with respect to
$\phi$ when $V^{M}_{\phi}(\mathbf{C}_\gamma)$ is undefined.

An interpretation $M = (\sD, I)$ of $L$ is a \emph{standard model} of
$L$ if $\sD$ is full.  Every standard model of $L$ is a general model
of $L$.

\subsection{Additional Compact Notation}

The compact notation for Alonzo types and expressions given above is
extended in~\cite{Farmer25} with a variety of operators, binders, and
abbreviations.  Equipped with this additional compact notation, Alonzo
becomes a practical logic in which mathematical ideas can be expressed
naturally and succinctly.  The compact notation that we need in this
paper from Chapter~6 of~\cite{Farmer25} is presented in
Tables~\ref{tab:nd-boolean}--\ref{tab:nd-quasitypes}.  To make the
notational definitions as readable as possible we have omitted
matching parentheses in the right-hand side of the definitions when
there is no loss of meaning and it is obvious where they should occur.

In Table~\ref{tab:nd-boolean}, we present notation for the truth
values and the standard Boolean operators.  The notation $\cAndPC$ is
an example of a \emph{pseudoconstant}.  It is not a real constant of
Alonzo, but it stands for an expression $\mathbf{C}_\gamma$ that can
be used just like a constant $\cCon {\mathbf{c}} {\gamma}$.  Unlike a
normal constant, $\cAndPC$ and most other pseudoconstants can be
employed in any language.  Thus they serve as logical constants.  The
same symbols that are used to write constants are used to write
pseudoconstants and parametric pseudoconstants (which are defined
below) (Notational Convention~9).

\begin{table}
\bc
\begin{tabular}{|lll|}
\hline

  $\cT$
  \index{8ca@$\cT$}
& stands for
& $\cEqX {\cFunAbs {x} {\cB} {x}} {\cFunAbs {x} {\cB} {x}}$.\\

  $\cF$ 
  \index{8cb@$\cF$}
& stands for
& $\cEqX {\cFunAbs {x} {\cB} {\cT}} {\cFunAbs {x} {\cB} {x}}$.\\

  $\cAndPC$
  \index{8cc@$\cAndPC$}
& stands for
& $\cFunAbsX {x} {\cB} {\cFunAbsX {y} {\cB}}$\\
& 
& \hspace*{2ex}${\cEqX {\cFunAbs {g} {\cFunTyBX {\cB} {\cB} {\cB}} 
  {\cFunAppBX {g} {\cT} {\cT}}} {}}$\\
&
& \hspace*{2ex}${\cFunAbs {g} {\cFunTyBX {\cB} {\cB} {\cB}} 
  {\cFunAppBX {g} {x} {y}}}$.\\

  $\cAnd {\mathbf{A}_\cB} {\mathbf{B}_\cB}$
  \index{8cd@$\cAnd {\mathbf{A}_\cB} {\mathbf{B}_\cB}$}
& stands for
& $\cFunAppBX {\cAndPC} {\mathbf{A}_\cB} {\mathbf{B}_\cB}$.\\

  $\cImpliesPC$
  \index{8ce@$\cImpliesPC$}
& stands for
& $\cFunAbsX {x} {\cB} {\cFunAbsX {y} {\cB} {\cEqX {x} {\cAnd {x} {y}}}}.$\\ 

  $\cImplies {\mathbf{A}_\cB} {\mathbf{B}_\cB}$
  \index{8cf@$\cImplies {\mathbf{A}_\cB} {\mathbf{B}_\cB}$}
& stands for
& $\cFunAppBX {\cImpliesPC} {\mathbf{A}_\cB} {\mathbf{B}_\cB}$.\\

  $\cNegPC$
  \index{8cg@$\cNegPC$}
& stands for
& $\cFunAbsX {x} {\cB} {\cEqX {x} {\cF}}$.\\

  $\cNeg {\mathbf{A}_\cB}$ 
  \index{8ch@$\cNeg {\mathbf{A}_\cB}$}
& stands for
& $\cFunAppX {\cNegPC} {\mathbf{A}_\cB}$.\\

  $\cOrPC$
  \index{8ci@$\cOrPC$}
& stands for
& $\cFunAbsX {x} {\cB} {\cFunAbsX {y} {\cB}
  {\cNegX {\cAnd {\cNegX{x}} {\cNegX{y}}}}}$.\\

  $\cOr {\mathbf{A}_\cB} {\mathbf{B}_\cB}$
  \index{8cj@$\cOr {\mathbf{A}_\cB} {\mathbf{B}_\cB}$}
& stands for
& $\cFunAppBX {\cOrPC} {\mathbf{A}_\cB} {\mathbf{B}_\cB}$.\\

\hline
\end{tabular}
\ec
\caption{Notational Definitions for Boolean Operators}
\label{tab:nd-boolean}
\end{table}

In Table~\ref{tab:nd-bin-op}, we present notation for binary
operators.  We will occasionally use implicit notational definitions
analogous to the notational definitions in Table~\ref{tab:nd-bin-op}
for the infix operators $<$, $>$, and $\ge$ corresponding to $\le$ for
other weak order operators such as $\subseteq$ and $\sqsubseteq$
(Notational Convention~10).

\begin{table}[t]
\bc
\begin{tabular}{|lll|}
\hline

  $\cBin {\mathbf{A}_\alpha} {\mathbf{c}} {\mathbf{B}_\alpha}$
  \index{8da@$\cBin {\mathbf{A}_\alpha} {\mathbf{c}} {\mathbf{B}_\alpha}$}
& stands for 
& $\cFunAppBX {\cCon {\mathbf{c}} {\cFunTyBX {\alpha} {\alpha} {\beta}}}
  {\mathbf{A}_\alpha} {\mathbf{B}_\alpha}$ \;or\;
  $\cFunAppX {\cCon {\mathbf{c}} {\cFunTyX {\cProdTy {\alpha} {\alpha}} {\beta}}}
  {\cOrdPair {\mathbf{A}_\alpha} {\mathbf{B}_\alpha}}$.\\

  $\cIff {\mathbf{A}_\cB} {\mathbf{B}_\cB}$
  \index{8db@$\cIff {\mathbf{A}_\cB} {\mathbf{B}_\cB}$}
& stands for   
& $\cEqX {\mathbf{A}_\cB} {\mathbf{B}_\cB}$.\\

  $\cNotEq {\mathbf{A}_\alpha} {\mathbf{B}_\alpha}$
  \index{8dc@$\cNotEq {\mathbf{A}_\alpha} {\mathbf{B}_\alpha}$}
& stands for 
& $\cNegX {\cEq {\mathbf{A}_\alpha} {\mathbf{B}_\alpha}}$.\\

  $\cBin {\mathbf{A}_\alpha} {<} {\mathbf{B}_\alpha}$
  \index{8dd@$\cBin {\mathbf{A}_\alpha} {<} {\mathbf{B}_\alpha}$}
& stands for 
& $\cAndX {\cFunAppB {{\le}_{\cFunTyBX {\alpha} {\alpha} {\cB}}}
  {\mathbf{A}_\alpha} {\mathbf{B}_\alpha}}
  {\cNotEq {\mathbf{A}_\alpha} {\mathbf{B}_\alpha}}$.\\

  $\cBin {\mathbf{A}_\alpha} {>} {\mathbf{B}_\alpha}$
  \index{8de@$\cBin {\mathbf{A}_\alpha} {>} {\mathbf{B}_\alpha}$}
& stands for 
& $\cBinX {\mathbf{B}_\alpha} {<} {\mathbf{A}_\alpha}$.\\

  $\cBin {\mathbf{A}_\alpha} {\ge} {\mathbf{B}_\alpha}$
  \index{8df@$\cBin {\mathbf{A}_\alpha} {\ge} {\mathbf{B}_\alpha}$}
& stands for 
& $\cBinX {\mathbf{B}_\alpha} {\le} {\mathbf{A}_\alpha}$.\\

  $\cBinB {\mathbf{A}_\alpha} {=} {\mathbf{B}_\alpha} {=} {\mathbf{C}_\alpha}$
  \index{8dg@$\cBinB {\mathbf{A}_\alpha} {=} {\mathbf{B}_\alpha} {=} {\mathbf{C}_\alpha}$}
& stands for 
& $\cAndX {\cEq {\mathbf{A}_\alpha} {\mathbf{B}_\alpha}}
  {\cEq {\mathbf{B}_\alpha} {\mathbf{C}_\alpha}}$.\\

  $\cBinB {\mathbf{A}_\alpha} {\mathbf{c}} {\mathbf{B}_\alpha} {\mathbf{d}} {\mathbf{C}_\alpha}$
  \index{8dh@$\cBinB {\mathbf{A}_\alpha} {\mathbf{c}} {\mathbf{B}_\alpha} {\mathbf{d}} {\mathbf{C}_\alpha}$}
& stands for 
& $\cAndX 
  {\cBin {\mathbf{A}_\alpha} {\mathbf{c}} {\mathbf{B}_\alpha}}
  {\cBin {\mathbf{B}_\alpha} {\mathbf{d}} {\mathbf{C}_\alpha}}$.\\

\hline
\end{tabular}
\ec
\caption{Notational Definitions for Binary Operators}
\label{tab:nd-bin-op}
\end{table}

In Table~\ref{tab:nd-quantifiers}, we present notation for the
universal and existential quantifiers.  We will usually write a
sequence of universal quantifiers and a sequence of existential
quantifiers in a more compact form with a single quantifier
(Notational Convention 11).  Thus, for example,
\[\cForallX {\mathbf{x}} {\alpha} {\cForallX {\mathbf{y}} {\alpha} 
{\cForallX {\mathbf{z}} {\beta} {\mathbf{A}_\cB}}}\] 
will be written as
\[\cForallBX {\mathbf{x}, \mathbf{y}} {\alpha} {\mathbf{z}} {\beta} 
{\mathbf{A}_\cB}.\] We will also use this form with quasitypes (which
are introduced below) (Notational Convention 12).

\begin{table}
\bc
\bigskip
\begin{tabular}{|lll|}
\hline

  $\cForall {\mathbf{x}} {\alpha} {\mathbf{A}_\cB}$
  \index{8ea@$\cForall {\mathbf{x}} {\alpha} {\mathbf{A}_\cB}$}
& stands for
& $\cEqX {\cFunAbs {x} {\alpha} {\cT}}
  {\cFunAbs {\mathbf{x}} {\alpha} {\mathbf{A}_\cB}}$.\\

  $\cForsome {\mathbf{x}} {\alpha} {\mathbf{A}_\cB}$
  \index{8eb@$\cForsome {\mathbf{x}} {\alpha} {\mathbf{A}_\cB}$}
& stands for
& $\cNegX {\cForall {\mathbf{x}} {\alpha} {\cNegX {\mathbf{A}_\cB}}}$.\\

\iffalse
  $\cForsomeUnique {\mathbf{x}} {\alpha} {\mathbf{A}_\cB}$
  \index{8ec@$\cForsomeUnique {\mathbf{x}} {\alpha} {\mathbf{A}_\cB}$}
& stands for
& $\cForsomeX {y} {\alpha} {\cEqX
  {\cFunAbs {\mathbf{x}} {\alpha} {\mathbf{A}_\cB}}
  {\cFunAbs {\mathbf{x}} {\alpha} {\cEqX {\mathbf{x}} {y}}}}$\\
&
& where $y$ does not occur in ${\cFunAbs {\mathbf{x}} {\alpha} {\mathbf{A}_\cB}}$.\\
\fi

\hline
\end{tabular}
\ec
\caption{Notational Definitions for Quantifiers}
\label{tab:nd-quantifiers}
\end{table}

In Table~\ref{tab:nd-definedness}, we present notation for expressions
involving definedness.  $\cBotPC {\cB}$~is a canonical ``undefined''
formula.  $\cBotPC {\alpha}$~is a canonical undefined expression of
type $\alpha \not= \cB$.  $\cEmpFunPC {\alpha} {\beta}$ is the empty
function of type $\cFunTyX {\alpha} {\beta}$ (where $\beta \not=
\cB$).  $\cIsDef {\mathbf{A}_\alpha}$ and $\cIsUndef {\mathbf{A}_\alpha}$ 
assert that the expression $\mathbf{A}_\alpha$ is defined and
undefined, respectively.  $\cQuasiEq {\mathbf{A}_\alpha}
{\mathbf{B}_\alpha}$ asserts that the expressions $\mathbf{A}_\alpha$
and $\mathbf{B}_\alpha$ are \emph{quasi-equal}, i.e., they are both
defined and equal or both undefined.  And $\cIf {\mathbf{A}_\cB}
{\mathbf{B}_\alpha} {\mathbf{C}_\alpha}$ is a conditional expression
that denotes the value of $\mathbf{B}_\alpha$ if $\mathbf{A}_\cB$
holds and otherwise denotes the value of~$\mathbf{C}_\alpha$.

\begin{table}[t]
\bc
\begin{tabular}{|lll|}
\hline

  $\cBotPC {\cB}$
  \index{8fa@$\cBotPC {\cB}$}
& stands for
& $\cF$.\\

  $\cBotPC {\alpha}$
  \index{8fb@$\cBotPC {\alpha}$}
& stands for
& $\cDefDesX {x} {\alpha} {\cNotEqX {x} {x}}${\dblsp}where $\alpha \not= \cB$.\\

  $\cEmpFunPC {\alpha} {\beta}$
  \index{8fc@$\cEmpFunPC {\alpha} {\beta}$}
& stands for
& $\cFunAbsX {x} {\alpha} {\cBotPC {\beta}}${\dblsp}where $\beta \not= \cB$.\\

  $\cIsDef {\mathbf{A}_\alpha}$
  \index{8fd@$\cIsDef {\mathbf{A}_\alpha}$}
& stands for
& $\cEqX {\mathbf{A}_\alpha} {\mathbf{A}_\alpha}$.\\

  $\cIsUndef {\mathbf{A}_\alpha}$
  \index{8fe@$\cIsUndef {\mathbf{A}_\alpha}$}
& stands for
& $\cNegX {\cIsDef {\mathbf{A}_\alpha}}$.\\

  $\cQuasiEq {\mathbf{A}_\alpha} {\mathbf{B}_\alpha}$
  \index{8ff@$\cQuasiEq {\mathbf{A}_\alpha} {\mathbf{B}_\alpha}$}
& stands for
& $\cImpliesX {\cOr {\cIsDefX {\mathbf{A}_\alpha}} 
  {\cIsDefX {\mathbf{B}_\alpha}}}
  {\cEqX {\mathbf{A}_\alpha} {\mathbf{B}_\alpha}}$.\\

  $\cNotQuasiEq {\mathbf{A}_\alpha} {\mathbf{B}_\alpha}$
  \index{8fg@$\cNotQuasiEq {\mathbf{A}_\alpha} {\mathbf{B}_\alpha}$}
& stands for
& $\cNegX {\cQuasiEq {\mathbf{A}_\alpha} {\mathbf{B}_\alpha}}$.\\

  $\cIfThenElse {\mathbf{A}_\cB} {\mathbf{B}_\cB} {\mathbf{C}_\cB}$
  \index{8fk@$\cIfThenElse {\mathbf{A}_\cB} {\mathbf{B}_\cB} {\mathbf{C}_\cB}$}
& stands for
& $\cAndX 
     {\cImplies {\mathbf{A}_\cB} {\mathbf{B}_\cB}} 
     {\cImplies {\cNegX {\mathbf{A}_\cB}} {\mathbf{C}_\cB}}$.\\

  $\cIfThenElse {\mathbf{A}_\cB} {\mathbf{B}_\alpha} {\mathbf{C}_\alpha}$
  \index{8fl@$\cIfThenElse {\mathbf{A}_\cB} {\mathbf{B}_\alpha} {\mathbf{C}_\alpha}$}
& stands for
& $\cDefDesX {x} {\alpha} {}$\\
&
& \hspace*{2ex}
  $\cAndX 
     {\cImplies {\mathbf{A}_\cB} {\cEqX {x} {\mathbf{B}_\alpha}}} 
     {\cImplies {\cNegX {\mathbf{A}_\cB}} {\cEqX {x} {\mathbf{C}_\alpha}}}$\\
&
& where $\alpha \not= \cB$.\\

  $\cIf {\mathbf{A}_\cB} {\mathbf{B}_\alpha} {\mathbf{C}_\alpha}$
  \index{8fm@$\cIf {\mathbf{A}_\cB} {\mathbf{B}_\alpha} {\mathbf{C}_\alpha}$}
& stands for
& $\cIfThenElse {\mathbf{A}_\cB} {\mathbf{B}_\alpha} {\mathbf{C}_\alpha}$.\\

\hline
\end{tabular}
\ec
\caption{Notational Definitions for Definedness}
\label{tab:nd-definedness}
\end{table}

The notation $\cBotPC {\alpha}$ is an example of a \emph{parametric
pseudoconstant}.  It stands for an expression $\mathbf{C}_\alpha$
where $\alpha$ is a \emph{parametric type} with the syntactic variable
$\alpha$ serving as a parameter that can be freely replaced with any
type.  Thus $\cBotPC {\alpha}$ is polymorphic in the sense that it can
be used with expressions of different types by simply replacing the
syntactic variable $\alpha$ with the type that is needed.  $\cEmpFunPC
{\alpha} {\beta}$ is similarly a parametric pseudoconstant.

The notational definitions of $\cIfThenElse {\mathbf{A}_\cB}
{\mathbf{B}_\cB} {\mathbf{C}_\cB}$ and $\cIfThenElse {\mathbf{A}_\cB}
{\mathbf{B}_\alpha} {\mathbf{C}_\alpha}$ (where $\alpha \not= \cB$)
are \emph{(parameterized) abbreviations} of the form
\[A(\mathbf{B}^{1}_{\alpha_1},\ldots,\mathbf{B}^{n}_{\alpha_n})\textrel{stands for} C\] 
where $A$ is a name, $n \ge 0$, and the syntactic variables
$\mathbf{B}^{1}_{\alpha_1},\ldots,\mathbf{B}^{1}_{\alpha_1}$ appear in
the expression $C$.  $A$ is written in uppercase sans sarif font to
distinguish it from the name of a constant or pseudoconstant
(Notational Convention~13).  We will always assume that the bound
variables introduced in $C$ are chosen so that they are not free in
$\mathbf{B}^{1}_{\alpha_1},\ldots,\mathbf{B}^{1}_{\alpha_1}$(Notational
Convention~14).  For example, the bound variable $\cVar {x} {\alpha}$
in the RHS of the notational definition of $\cIfThenElse
{\mathbf{A}_\cB} {\mathbf{B}_\alpha} {\mathbf{C}_\alpha}$ (where
$\alpha \not= \cB$) in Table~\ref{tab:nd-definedness} is chosen so
that it is not free in ${\mathbf{A}_\cB}$, ${\mathbf{B}_\alpha}$,
or~${\mathbf{C}_\alpha}$.

Since we can identify a set $S \subseteq U$ with the predicate $p_S :
U \tarrow \bB$ such that $a \in S$ iff $p_S(a)$, we will introduce a
power set type of $\alpha$, i.e., a type of the subsets of $\alpha$,
as the type $\cFunTyX {\alpha} {\cB}$ of predicates on $\alpha$.  The
compact notation for $\cFunTyX {\alpha} {\cB}$ is $\cSetTy {\alpha}$.
We introduce this notation and compact notation for the common set
operators in Table~\ref{tab:nd-sets}.  $\cEmpSetPC {\alpha}$ and
$\cUnivSetPC {\alpha}$ are parametric pseudoconstants that denote the
empty set and the universal set, respectively, of the members in the
domain of $\alpha$.

\begin{table}[t]
\bc
\begin{tabular}{|lll|}
\hline

  $\cSetTy {\alpha}$
  \index{8ga@$\cSetTy {\alpha}$}
& stands for
& $\cFunTyX {\alpha} {\cB}$.\\

  $\cIn {\mathbf{A}_\alpha} {\mathbf{B}_{\cSetTy {\alpha}}}$
  \index{8gb@$\cIn {\mathbf{A}_\alpha} {\mathbf{B}_{\cSetTy {\alpha}}}$}
& stands for
& $\cFunAppX {\mathbf{B}_{\cSetTy {\alpha}}} {\mathbf{A}_\alpha}$.\\

  $\cNotIn {\mathbf{A}_\alpha} {\mathbf{B}_{\cSetTy {\alpha}}}$
  \index{8gc@$\cNotIn {\mathbf{A}_\alpha} {\mathbf{B}_{\cSetTy {\alpha}}}$}
& stands for
& $\cNegX {\cIn {\mathbf{A}_\alpha} {\mathbf{B}_{\cSetTy {\alpha}}}}$.\\

  $\cSet {\mathbf{x}} {\alpha} {\mathbf{A}_\cB}$
  \index{8gd@$\cSet {\mathbf{x}} {\alpha} {\mathbf{A}_\cB}$}
& stands for
& $\cFunAbsX {\mathbf{x}} {\alpha} {\mathbf{A}_\cB}$.\\

  $\cEmpSetPC {\alpha}$
  \index{8ge@$\cEmpSetPC {\alpha}$}
& stands for
& $\cFunAbsX {x} {\alpha} {\cF}$.\\

  $\cEmpSetAltPC {\alpha}$
  \index{8gf@$\cEmpSetAltPC {\alpha}$}
& stands for
& $\cEmpSetPC {\alpha}$.\\

  $\cUnivSetPC {\alpha}$
  \index{8gg@$\cUnivSetPC {\alpha}$}
& stands for
& $\cFunAbsX {x} {\alpha} {\cT}$.\\

  $\cFinSet {n} {\alpha}$
  \index{8gh@$\cFinSet {n} {\alpha}$}
& stands for
& $\LambdaApp x_1 : \alpha \mDot \cdots \mDot \LambdaApp x_n : \alpha \mDot 
  \LambdaApp x : \alpha \mDot$\\
&
& \hspace*{2ex}$x = x_1 \Or \cdots \Or x = x_n${\dblsp}where $n \ge 1$.\\

  $\cFinSetL{\mathbf{A}^{1}_{\alpha}, \ldots, \mathbf{A}^{n}_{\alpha}}$
  \index{8gi@$\cFinSetL{\mathbf{A}^{1}_{\alpha}, \ldots, \mathbf{A}^{n}_{\alpha}}$}
& stands for
& $\cFinSet {n} {\alpha}\,\mathbf{A}^{1}_{\alpha}\,\cdots\,\mathbf{A}^{n}_{\alpha}${\dblsp}where $n \ge 1$.\\

  $\cSubseteqPC {\alpha}$
  \index{8gj@$\cSubseteqPC {\alpha}$}
& stands for
& $\cFunAbsX {a} {\cSetTy {\alpha}} {\cFunAbsX {b} {\cSetTy {\alpha}}}$\\
& 
& \hspace*{2ex}$\cForallX {x} {\alpha} 
  {\cImpliesX {\cInX {x} {a}} {\cInX {x} {b}}}$.\\

  $\cUnionPC {\alpha}$
  \index{8gk@$\cUnionPC {\alpha}$}
& stands for
& $\cFunAbsX {a} {\cSetTy {\alpha}} {\cFunAbsX {b} {\cSetTy {\alpha}}}$\\
& 
& \hspace*{2ex}$\cSet {x} {\alpha} {\cOrX {\cInX {x} {a}} {\cInX {x} {b}}}$.\\

  $\cIntersPC {\alpha}$
  \index{8gl@$\cIntersPC {\alpha}$}
& stands for
& $\cFunAbsX {a} {\cSetTy {\alpha}} {\cFunAbsX {b} {\cSetTy {\alpha}}}$\\
& 
& \hspace*{2ex}$\cSet {x} {\alpha} {\cAndX {\cInX {x} {a}} {\cInX {x} {b}}}$.\\

  $\cComplPC {\alpha}$
  \index{8gm@$\cComplPC {\alpha}$}
& stands for
& $\cFunAbsX {a} {\cSetTy \alpha} {\cSet {x} {\alpha} {\cNotInX {x} {a}}}$.\\

  $\cComplX {\mathbf{A}_{\cSetTy {\alpha}}}$
  \index{8gn@$\cCompl {\mathbf{A}_{\cSetTy {\alpha}}}$}
& stands for
& $\cFunAppX {\compl{\,\cdot\,}_{\cFunTyX {\cSetTy {\alpha}} {\cSetTy {\alpha}}}}
  {\mathbf{A}_{\cSetTy {\alpha}}}$.\\

  $\cSetDiffPC {\alpha}$
  \index{8go@$\cSetDiffPC {\alpha}$}
& stands for
& $\cFunAbsX {a} {\cSetTy {\alpha}} {\cFunAbsX {b} {\cSetTy {\alpha}}
     {\cIntersX {a} {\cComplX {b}}}}$.\\

\hline
\end{tabular}
\ec
\caption{Notational Definitions for Sets}
\label{tab:nd-sets}
\end{table}

We introduce notation for product types, tuples, and the accessors for
ordered pairs in Table~\ref{tab:nd-tuples}.

\begin{table}[t]
\bc
\begin{tabular}{|lll|}
\hline

  $\cTupleTyL {\alpha}$
  \index{8ha@$\cTupleTyL {\alpha}$}
& stands for
& $\alpha$.\\

  $\cTupleTyL {\alpha_1 \times \cdots \times \alpha_n}$
  \index{8hb@$\cTupleTyL {\alpha_1 \times \cdots \times \alpha_n}$}
& stands for
& $\cProdTy {\alpha_1} {\cTupleTyL {\alpha_2 \times \cdots \times \alpha_n}}$
  {\dblsp}where $n \ge 2$.\\

  $\cTupleL {\mathbf{A}_\alpha}$
  \index{8hc@$\cTupleL {\mathbf{A}_\alpha}$}
& stands for
& $\mathbf{A}_\alpha$.\\

  $\cTupleL {\mathbf{A}^{1}_{\alpha_1}, \ldots, \mathbf{A}^{n}_{\alpha_n}}$
  \index{8hd@$\cTupleL {\mathbf{A}^{1}_{\alpha_1}, \ldots, \mathbf{A}^{n}_{\alpha_n}}$}
& stands for
& $\cOrdPair {\mathbf{A}^{1}_{\alpha_1}} 
  {\cTupleL {\mathbf{A}^{2}_{\alpha_1}, \ldots, \mathbf{A}^{n}_{\alpha_n}}}$
  {\dblsp}where $n \ge 2$.\\

  $\cFstPC {\alpha} {\beta}$
  \index{8he@$\cFstPC {\alpha} {\beta}$}
& stands for
& $\cFunAbsX {p} {\cProdTyX {\alpha} {\beta}} {\cDefDesX {x} {\alpha} 
  {\cForsomeX {y} {\beta} {\cEqX {p} {\cOrdPair {x} {y}}}}}$.\\

  $\cSndPC {\alpha} {\beta}$
  \index{8hf@$\cSndPC {\alpha} {\beta}$}
& stands for
& $\cFunAbsX {p} {\cProdTyX {\alpha} {\beta}} {\cDefDesX {y} {\beta} 
  {\cForsomeX {x} {\alpha} {\cEqX {p} {\cOrdPair {x} {y}}}}}$.\\

\hline
\end{tabular}
\ec
\caption{Notational Definitions for Tuples}
\label{tab:nd-tuples}
\end{table}

Some convenient notation for functions is found in
Table~\ref{tab:nd-functions}.

\begin{table}[t]
\bc
\begin{tabular}{|lll|}
\hline

  $\cIdFunPC {\alpha}$
  \index{8ia@$\cIdFunPC {\alpha}$}
& stands for
& $\cFunAbsX {x} {\alpha} {x}$.\\

  $\cDomPC {\alpha} {\beta}$
  \index{8ib@$\cDomPC {\alpha} {\beta}$}
& stands for
& $\cFunAbsX {f} {\cFunTyX {\alpha} {\beta}} {}$\\
&
& \hspace*{2ex}${\cSet {x} {\alpha} {\cIsDefX {\cFunApp {f} {x}}}}$.\\

  $\cRanPC {\alpha} {\beta}$
  \index{8ic@$\cRanPC {\alpha} {\beta}$}
& stands for
& $\cFunAbsX {f} {\cFunTyX {\alpha} {\beta}} {}$\\
&
& \hspace*{2ex}${\cSet {y} {\beta} 
  {\cForsomeX {x} {\alpha} {\cEqX {\cFunAppX {f} {x}} {y}}}}$.\\

  $\cTotal {\mathbf{F}_{\cFunTyX {\alpha} {\beta}}}$
  \index{8ja@$\cTotal {\cdots}$}
& stands for
& $\cForallX {x} {\alpha} {\cIsDefX 
  {\cFunApp {\mathbf{F}_{\cFunTyX {\alpha} {\beta}}} {x}}}$.\\

\iffalse
  $\cSubfuneqPC {\alpha} {\beta}$
  \index{8id@$\cSubfuneqPC {\alpha} {\beta}$}
& stands for
& $\cFunAbsX {f} {\cFunTyX {\alpha} {\beta}} 
  {\cFunAbsX{g} {\cFunTyX {\alpha} {\beta}} {}}$\\
&
& \hspace*{2ex}${\cForallX {x} {\alpha} 
  {\cImpliesX {\cInX {x} {\cFunAppX {\cDomPC {\alpha} {\beta}} {f}}} {}}}$\\
&
& \hspace*{2ex}${\cEqX {\cFunAppX {f} {x}} {\cFunAppX {g} {x}}}$.\\

  $\cFunCompPC {\alpha} {\beta} {\gamma}$
  \index{8ie@$\cFunCompPC {\alpha} {\beta} {\gamma}$}
& stands for
& $\cFunAbsX {f} {\cFunTyX {\alpha} {\beta}} 
  {\cFunAbsX{g} {\cFunTyX {\beta} {\gamma}} {}}$\\
&
& \hspace*{2ex}${\cFunAbsX {x} {\alpha} {\cFunAppX {g} {\cFunApp {f} {x}}}}$.\\

  $\cFunComp 
   {\mathbf{F}_{\cFunTyX {\alpha} {\beta}}} 
   {\mathbf{G}_{\cFunTyX {\beta} {\gamma}}}$
  \index{8if@$\cFunComp {\mathbf{F}_{\cFunTyX {\alpha} {\beta}}} {\mathbf{G}_{\cFunTyX {\beta} {\gamma}}}$}
& stands for
& $\cFunAppBX {\cFunCompPC {\alpha} {\beta} {\gamma}}
  {\mathbf{F}_{\cFunTyX {\alpha} {\beta}}} 
  {\mathbf{G}_{\cFunTyX {\beta} {\gamma}}}$.\\
\fi

  $\cRestrictPC {\alpha} {\beta}$
  \index{8ig@$\cRestrictPC {\alpha} {\beta}$}
& stands for
& $\cFunAbsX {f} {\cFunTyX {\alpha} {\beta}} 
     {\cFunAbsX {s} {\cSetTy {\alpha}} {}}$\\
&
& \hspace*{2ex}${\cFunAbsX {x} {\alpha} 
  {\cIfX {\cInX {x} {s}} {\cFunAppX {f} {x}} {\cBotPC {\beta}}}}$.\\

  $\cRestrict {\mathbf{F}_{\cFunTyX {\alpha} {\beta}}} {\mathbf{A}_{\cSetTy {\alpha}}}$
  \index{8ih@$\cRestrict {\mathbf{F}_{\cFunTyX {\alpha} {\beta}}} {\mathbf{A}_{\cSetTy {\alpha}}}$}
& stands for
& $\cFunAppBX {\cRestrictPC {\alpha} {\beta}} 
     {\mathbf{F}_{\cFunTyX {\alpha} {\beta}}} {\mathbf{A}_{\cSetTy {\alpha}}}$.\\

\hline
\end{tabular}
\ec
\caption{Notational Definitions for Functions}
\label{tab:nd-functions}
\end{table}

A \emph{quasitype within type $\alpha \in \sT$} is any expression of
type $\cSetTy {\alpha} = \cFunTyX {\alpha} {\cB}$.  A quasitype
$\mathbf{Q}_{\cSetTy {\alpha}}$ denotes a subset of the domain denoted
by $\alpha$.  Thus quasitypes represent subtypes\index{Type!subtype} and
are useful for specifying subdomains of a domain.  Unlike a type, a
quasitype may denote an empty domain.  Notice that an expression
$\mathbf{A}_{\cFunTyX {\alpha} {\cB}}$ is simultaneously an expression
of type $\cFunTyX {\alpha} {\cB}$, an expression of type of $\cSetTy
{\alpha}$, and a quasitype within type $\alpha$.  So
$\mathbf{A}_{\cFunTyX {\alpha} {\cB}}$ (or $\mathbf{A}_{\cSetTy
  {\alpha}}$) can be used as a function, as a set, and like a type as
shown below.

In Table~\ref{tab:nd-quasitypes}, we introduce various notations for
using quasitypes in place of types.  Quasitypes can be used to
restrict the range of a variable bound by a binder.  For example,
$\cFunAbsQTy {x} {\mathbf{Q}_{\cSetTy {\alpha}}} {\mathbf{B}_\beta}$
denotes the function denoted by $\cFunAbsX {x} {\alpha}
{\mathbf{B}_\beta}$ weakly restricted to the domain denoted by
$\mathbf{Q}_{\cSetTy {\alpha}}$.  Quasitypes can also be used to
sharpen definedness statements.  For example, $\cIsDefInQTy
{\mathbf{A}_\alpha} {\mathbf{Q}_{\cSetTy {\alpha}}}$, read as
\emph{$\mathbf{A}_\alpha$ is defined in $\mathbf{Q}_{\cSetTy
  {\alpha}}$}, asserts that the value of $\mathbf{A}_\alpha$ is
defined and is a member of the set denoted by $\mathbf{Q}_{\cSetTy
  {\alpha}}$.  $\cFunQTy {\mathbf{Q}_{\cSetTy {\alpha}}}
{\mathbf{R}_{\cSetTy {\beta}}}$ is a quasitype within $\cFunTyX
{\alpha} {\beta}$ that denotes the function space from the denotation
of $\mathbf{Q}_{\cSetTy {\alpha}}$ to the denotation of
${\mathbf{R}_{\cSetTy {\beta}}}$, and $\cProdQTy {\mathbf{Q}_{\cSetTy
    {\alpha}}} {\mathbf{R}_{\cSetTy {\beta}}}$ is a quasitype within
$\cProdTyX {\alpha} {\beta}$ that denotes the Cartesian product of the
denotation of $\mathbf{Q}_{\cSetTy {\alpha}}$ and the denotation of
${\mathbf{R}_{\cSetTy {\beta}}}$.  

\begin{table}[t]\footnotesize
\bc
\begin{tabular}{|lll|}
\hline

\iffalse
\multicolumn{3}{|l|}{Let $\cVar {x} {\alpha}$, $\cVar {x'} {\alpha}$,
  $\cVar {y} {\beta}$, $\cVar {y'} {\beta}$, $\cVar {z} {\gamma}$,
  $\cVar {s} {\cSetTy {\alpha}}$, and $\cVar {f} {\cFunTyX {\alpha}
    {\alpha}}$ not be free in}\\

\multicolumn{3}{|l|}{$\mathbf{F}_{\cFunTyX {\alpha} {\beta}}$,
  $\mathbf{F}_{\cFunTyBX {\alpha} {\beta} {\gamma}}$,
  $\mathbf{Q}_{\cSetTy {\alpha}}$, $\mathbf{R}_{\cSetTy {\beta}}$, and
  $\mathbf{S}_{\cSetTy {\gamma}}$.}\\

\multicolumn{3}{|l|}{}\\
\fi

  $\cFunAbsQTy {\mathbf{x}} {\mathbf{Q}_{\cSetTy {\alpha}}} {\mathbf{B}_\beta}$
  \index{8ka@$\cFunAbsQTy {\mathbf{x}} {\mathbf{Q}_{\cSetTy {\alpha}}} {\mathbf{B}_\beta}$}
& stands for
& $\cFunAbsX {\mathbf{x}} {\alpha} {\cIf
  {\cInX {\mathbf{x}} {\mathbf{Q}_{\cSetTy {\alpha}}}} 
  {\mathbf{B}_\beta} {\cBotPC {\beta}}}$.\\

  $\cForallQTy {\mathbf{x}} {\mathbf{Q}_{\cSetTy {\alpha}}} {\mathbf{B}_\cB}$
  \index{8kb@$\cForallQTy {\mathbf{x}} {\mathbf{Q}_{\cSetTy {\alpha}}} {\mathbf{B}_\cB}$}
& stands for
& $\cForallX {\mathbf{x}} {\alpha} {\cImplies
  {\cInX {\mathbf{x}} {\mathbf{Q}_{\cSetTy {\alpha}}}} {\mathbf{B}_\cB}}$.\\

  $\cForsomeQTy {\mathbf{x}} {\mathbf{Q}_{\cSetTy {\alpha}}} {\mathbf{B}_\cB}$
  \index{8kc@$\cForsomeQTy {\mathbf{x}} {\mathbf{Q}_{\cSetTy {\alpha}}} {\mathbf{B}_\cB}$}
& stands for
& $\cForsomeX {\mathbf{x}} {\alpha} {\cAnd
  {\cInX {\mathbf{x}} {\mathbf{Q}_{\cSetTy {\alpha}}}} {\mathbf{B}_\cB}}$.\\

  $\cDefDesQTy {\mathbf{x}} {\mathbf{Q}_{\cSetTy {\alpha}}} {\mathbf{B}_\cB}$
  \index{8kd@$\cDefDesQTy {\mathbf{x}} {\mathbf{Q}_{\cSetTy {\alpha}}} {\mathbf{B}_\cB}$}
& stands for
& $\cDefDesX {\mathbf{x}} {\alpha} {\cAnd
  {\cInX {\mathbf{x}} {\mathbf{Q}_{\cSetTy {\alpha}}}} {\mathbf{B}_\cB}}$.\\

  $\cIsDefInQTy {\mathbf{A}_\alpha} {\mathbf{Q}_{\cSetTy {\alpha}}}$
  \index{8ke@$\cIsDefInQTy {\mathbf{A}_\alpha} {\mathbf{Q}_{\cSetTy {\alpha}}}$}
& stands for
& $\cAndX {\cIsDefX {\mathbf{A}_\alpha}}
  {\cInX {\mathbf{A}_\alpha} {\mathbf{Q}_{\cSetTy {\alpha}}}}$.\\

  $\cIsUndefInQTy {\mathbf{A}_\alpha} {\mathbf{Q}_{\cSetTy {\alpha}}}$
  \index{8kf@$\cIsUndefInQTy {\mathbf{A}_\alpha} {\mathbf{Q}_{\cSetTy {\alpha}}}$}
& stands for
& $\cNegX {\cIsDefInQTy {\mathbf{A}_\alpha} {\mathbf{Q}_{\cSetTy {\alpha}}}}$.\\

  $\cFunQTyPC {\alpha} {\beta}$
  \index{8kg@$\cFunQTyPC {\alpha} {\beta}$}
& stands for
& $\cFunAbsX {s} {\cSetTy {\alpha}} {\cFunAbsX {t} {\cSetTy {\beta}}}$\\
&
& \hspace*{2ex}$\{f : {\cFunTyX {\alpha} {\beta}} \mid 
  \cForallX {x} {\alpha}$\\
&
& \hspace*{4ex}${\cImpliesX {\cIsDefX {\cFunApp {f} {x}}}
  {\cAnd {\cInX {x} {s}} {\cInX {\cFunAppX {f} {x}} {t}}}}\}$\\
&
& where $\beta \not= \cB$.\\

  $\cProdQTyPC {\alpha} {\beta}$
  \index{8kh@$\cProdQTyPC {\alpha} {\beta}$}
& stands for
& $\cFunAbsX {s} {\cSetTy {\alpha}} {\cFunAbsX {t} {\cSetTy {\beta}}}$\\
&
& \hspace*{2ex}$\{p : {\cProdTyX {\alpha} {\beta}} \mid $\\
&
& \hspace*{4ex}${\cInX {\cFunAppX 
  {\cFstPC {\alpha} {\beta}} {p}} {s}} \And {}$\\
&
& \hspace*{4ex}${\cInX {\cFunAppX
  {\cSndPC {\alpha} {\beta}} {p}} {t}}\}$\\

  $\cFunQTy {\mathbf{Q}_{\cSetTy {\alpha}}} {\cB}$
  \index{8ko@$\cFunQTy {\mathbf{Q}_{\cSetTy {\alpha}}} {\cB}$}
& stands for  
& $\cSet 
     {s} 
     {\cSetTy {\alpha}}
     {\cSubseteqX {s} {\mathbf{Q}_{\cSetTy {\alpha}}}}$.\\

  $\cSetQTy {\mathbf{Q}_{\cSetTy {\alpha}}}$
  \index{8kp@$\cSetQTy {\mathbf{Q}_{\cSetTy {\alpha}}}$}
& stands for  
& $\cFunQTyX {\mathbf{Q}_{\cSetTy {\alpha}}} {\cB}$.\\

  $\cFunQTy {\mathbf{Q}_{\cSetTy {\alpha}}} {\mathbf{R}_{\cSetTy {\beta}}}$
  \index{8ki@$\cFunQTy {\mathbf{Q}_{\cSetTy {\alpha}}} {\mathbf{R}_{\cSetTy {\beta}}}$}
& stands for
& $\cFunAppBX {\cFunQTyPC {\alpha} {\beta}}
  {\mathbf{Q}_{\cSetTy {\alpha}}} {\mathbf{R}_{\cSetTy {\beta}}}$\\
&
& where $\beta \not= \cB$.\\

  $\cFunQTy {\alpha} {\mathbf{R}_{\cSetTy {\beta}}}$
  \index{8kj@$\cFunQTy {\alpha} {\mathbf{R}_{\cSetTy {\beta}}}$}
& stands for  
& $\cFunQTyX {\cUnivSetPC {\alpha}} {\mathbf{R}_{\cSetTy {\beta}}}$ 
  {\dblsp}where $\beta \not= \cB$.\\

  $\cFunQTy {\mathbf{Q}_{\cSetTy {\alpha}}} {\beta}$
  \index{8kk@$\cFunQTy {\mathbf{Q}_{\cSetTy {\alpha}}} {\beta}$}
& stands for  
& $\cFunQTyX {\mathbf{Q}_{\cSetTy {\alpha}}} {\cUnivSetPC {\beta}}$
  {\dblsp}where $\beta \not= \cB$.\\

  $\cProdQTy {\mathbf{Q}_{\cSetTy {\alpha}}} {\mathbf{R}_{\cSetTy {\beta}}}$
  \index{8kl@$\cProdQTy {\mathbf{Q}_{\cSetTy {\alpha}}} {\mathbf{R}_{\cSetTy {\beta}}}$}
& stands for
& $\cFunAppBX {\cProdQTyPC {\alpha} {\beta}}
  {\mathbf{Q}_{\cSetTy {\alpha}}} {\mathbf{R}_{\cSetTy {\beta}}}.$\\

  $\cProdQTy {\alpha} {\mathbf{R}_{\cSetTy {\beta}}}$
  \index{8km@$\cProdQTy {\alpha} {\mathbf{R}_{\cSetTy {\beta}}}$}
& stands for  
& $\cProdQTyX {\cUnivSetPC {\alpha}} {\mathbf{R}_{\cSetTy {\beta}}}$.\\

  $\cProdQTy {\mathbf{Q}_{\cSetTy {\alpha}}} {\beta}$
  \index{8kn@$\cProdQTy {\mathbf{Q}_{\cSetTy {\alpha}}} {\beta}$}
& stands for  
& $\cProdQTyX {\mathbf{Q}_{\cSetTy {\alpha}}} {\cUnivSetPC {\beta}}$.\\

  $\cTotalOn {\mathbf{F}_{\cFunTyX {\alpha} {\beta}}}
  {\mathbf{Q}_{\cSetTy {\alpha}}} {\mathbf{R}_{\cSetTy {\beta}}}$
  \index{8kq@$\textsf{TOTAL-ON}(\cdots)$}
& stands for
& $\cForallX {x} {\mathbf{Q}_{\cSetTy {\alpha}}} {\cIsDefInQTyX
  {\cFunApp {\mathbf{F}_{\cFunTyX {\alpha} {\beta}}} {x}} 
  {\mathbf{R}_{\cSetTy {\beta}}}}$.\\

\hline
\end{tabular}
\ec
\caption{Notational Definitions for Quasitypes}\label{tab:nd-quasitypes}
\end{table}

\iffalse
The name of a (nontrivial) development of a
theory named \textsf{X} will be given a name of the form
\textsf{X-$n$} where $n$ is a positive integer.
\fi

\section{Monoids}\label{sec:monoids}

A \emph{monoid} is a mathematical structure $(m,\cdot,e)$ where $m$ is
a nonempty set of values, $\cdot : (m \times m) \tarrow m$ is an
associative function, and $e \in m$ is an identity element with
respect to $\cdot$.  Mathematics and computing are replete with
examples of monoids such as $(\bN,+,0)$, $(\bN,*,1)$, and
$(\Sigma^*,\mathrel{++},\epsilon)$ where $\Sigma^*$ is the set of
strings over an alphabet $\Sigma$, $\mathrel{++}$ is string
concatenation, and $\epsilon$ is the empty string.

Table~\ref{tab:monoids-pc} defines some parametric pseudoconstants
that we will need for monoids, and Table~\ref{tab:monoids-abbr} defines
several useful abbreviations for monoids.

\begin{table}
\bc
\bigskip
\begin{tabular}{|l|}
\hline

${\cSetProdPC {\textsf{set-op}} {\alpha} {\beta} {\gamma}}$\\[1ex]
\hspace{2ex} stands for\\[1ex]
${\cFunAbsX 
    {f} 
    {\cFunTyX {\cProdTy {\alpha} {\beta}} {\gamma}}
    {\cFunAbsX
       {p}
       {\cProdTyX {\cSetTy {\alpha}} {\cSetTy {\beta}}} {}}}$\\
       \hspace*{2ex}
      ${\cSet
          {z}
          {\gamma}
          {\cForsomeBX
             {x}
             {\cFunAppX {\mName {fst}} {p}}
             {y}
             {\cFunAppX {\mName {snd}} {p}}
             {\cEqX
                {z}
                {\cFunAppX {f} {\cOrdPair {x} {y}}}}}}$.\\\hline

${\cFunCompPairPC {\alpha} {\beta} {\gamma}}$\\[1ex]
\hspace{2ex} stands for\\[1ex]
${\cFunAbsX
    {p}
    {\cProdTyX {\cFunTy {\alpha} {\beta}} {\cFunTy {\beta} {\gamma}}}
    {\cFunAbsX
       {x}
       {\alpha}
       {\cFunAppX
          {\cFunApp {\mName{snd}} {p}}
          {\cFunApp 
             {\cFunApp {\mName{fst}} {p}}
             {x}}}}}$.\\\hline

${\cFunComp 
    {\mathbf{F}_{\cFunTyX {\alpha} {\beta}}} 
    {\mathbf{G}_{\cFunTyX {\beta} {\gamma}}}}$\\[1ex]
\hspace{2ex} stands for\\[1ex]
${\cFunAppX 
    {\cFunCompPairPC {\alpha} {\beta} {\gamma}}
    {\cOrdPair
       {\mathbf{F}_{\cFunTyX {\alpha} {\beta}}} 
       {\mathbf{G}_{\cFunTyX {\beta} {\gamma}}}}}$.\\\hline

${\cFunAppPairPC {\alpha} {\beta}}$\\[1ex]
\hspace{2ex} stands for\\[1ex]
${\cFunAbsX
    {p}
    {\cProdTyX {\cFunTy {\alpha} {\beta}} {\alpha}}
    {\cFunAppX
       {\cFunApp {\mName{fst}} {p}}
       {\cFunApp {\mName{snd}} {p}}}}$.\\\hline

\hline
\end{tabular}
\ec
\caption{Notational Definitions for Monoids: Pseudoconstants}
\label{tab:monoids-pc}
\end{table}

\begin{table}
\bc
\begin{tabular}{|l|}
\hline

${\cMonoid 
   {\mathbf{M}_{\cSetTy {\alpha}}}
   {\mathbf{F}_{\cFunTyX {\cProdTy {\alpha} {\alpha}} {\alpha}}}
   {\mathbf{E}_{\alpha}}}$\\[1ex]
\hspace{2ex} stands for\\[1ex]
${\cIsDefX {\mathbf{M}_{\cSetTy {\alpha}}}} \And {}$\\
${\cNotEqX {\mathbf{M}_{\cSetTy {\alpha}}} {\cEmpSetPC {\alpha}}} \And {}$\\
${\cIsDefInQTyX
    {\mathbf{F}_{\cFunTyX {\cProdTy {\alpha} {\alpha}} {\alpha}}}
    {\cFunQTyX
       {\cProdQTy {\mathbf{M}_{\cSetTy {\alpha}}} {\mathbf{M}_{\cSetTy {\alpha}}}}
       {\mathbf{M}_{\cSetTy {\alpha}}}}} \And {}$\\
${\cIsDefInQTyX {\mathbf{E}_{\alpha}} {\mathbf{M}_{\cSetTy {\alpha}}}} \And {}$\\
${\cForallX 
    {x, y, z} 
    {\mathbf{M}_{\cSetTy {\alpha}}} {}}$\\
    \hspace*{2ex}
   ${\cEqX 
        {\cFunAppX
           {\mathbf{F}_{\cFunTyX {\cProdTy {\alpha} {\alpha}} {\alpha}}}
           {\cOrdPair
              {x}
              {\cFunAppX 
                 {\mathbf{F}_{\cFunTyX {\cProdTy {\alpha} {\alpha}} {\alpha}}}
                 {\cOrdPair {y} {z}}}}}
       {\cFunAppX
          {\mathbf{F}_{\cFunTyX {\cProdTy {\alpha} {\alpha}} {\alpha}}}
          {\cOrdPair          
             {\cFunAppX 
                {\mathbf{F}_{\cFunTyX {\cProdTy {\alpha} {\alpha}} {\alpha}}}
                {\cOrdPair {x} {y}}}
              {z}}}} \And {}$\\
${\cForallX
    {x}
    {\mathbf{M}_{\cSetTy {\alpha}}}
    {\cBinBX
      {\cFunAppX 
         {\mathbf{F}_{\cFunTyX {\cProdTy {\alpha} {\alpha}} {\alpha}}}
         {\cOrdPair {\mathbf{E}_{\alpha}} {x}}}
      {=}
      {\cFunAppX
         {\mathbf{F}_{\cFunTyX {\cProdTy {\alpha} {\alpha}} {\alpha}}}
         {\cOrdPair {x} {\mathbf{E}_{\alpha}}}}
      {=}
      {x}}}$.\\\hline

${\cComMonoid 
   {\mathbf{M}_{\cSetTy {\alpha}}}
   {\mathbf{F}_{\cFunTyX {\cProdTy {\alpha} {\alpha}} {\alpha}}}
   {\mathbf{E}_{\alpha}}}$\\[1ex]
\hspace{2ex} stands for\\[1ex]
${\cMonoid 
    {\mathbf{M}_{\cSetTy {\alpha}}}
    {\mathbf{F}_{\cFunTyX {\cProdTy {\alpha} {\alpha}} {\alpha}}}
    {\mathbf{E}_{\alpha}}} \And {}$\\
${\cForallX 
    {x, y} 
    {\mathbf{M}_{\cSetTy {\alpha}}}
    {\cEqX
       {\cFunAppX
            {\mathbf{F}_{\cFunTyX {\cProdTy {\alpha} {\alpha}} {\alpha}}}
            {\cOrdPair {x} {y}}}
       {\cFunAppX
            {\mathbf{F}_{\cFunTyX {\cProdTy {\alpha} {\alpha}} {\alpha}}}
            {\cOrdPair {y} {x}}}}}$\\\hline

${\cMonAction
   {\mathbf{M}_{\cSetTy {\alpha}}}
   {\mathbf{S}_{\cSetTy {\beta}}}
   {\mathbf{F}_{\cFunTyX {\cProdTy {\alpha} {\alpha}} {\alpha}}}
   {\mathbf{E}_{\alpha}}
   {\mathbf{G}_{\cFunTyX {\cProdTy {\alpha} {\beta}} {\beta}}}}$\\[1ex]
\hspace{2ex} stands for\\[1ex]
${\cMonoid 
    {\mathbf{M}_{\cSetTy {\alpha}}}
    {\mathbf{F}_{\cFunTyX {\cProdTy {\alpha} {\alpha}} {\alpha}}}
    {\mathbf{E}_{\alpha}}} \And {}$\\
${\cIsDefX {\mathbf{S}_{\cSetTy {\beta}}}} \And {}$\\
${\cNotEqX {\mathbf{S}_{\cSetTy {\beta}}} {\cEmpSetPC {\beta}}} \And {}$\\
${\cIsDefInQTyX
    {\mathbf{G}_{\cFunTyX {\cProdTy {\alpha} {\beta}} {\beta}}}
    {\cFunQTyX
       {\cProdQTy {\mathbf{M}_{\cSetTy {\alpha}}} {\mathbf{S}_{\cSetTy {\beta}}}}
       {\mathbf{S}_{\cSetTy {\beta}}}}} \And {}$\\
${\cForallBX 
    {x, y} 
    {\mathbf{M}_{\cSetTy {\alpha}}}
    {s}
    {\mathbf{S}_{\cSetTy {\beta}}} {}}$\\
      \hspace*{2ex}
      ${\cEqX
          {\cFunAppX
             {\mathbf{G}_{\cFunTyX {\cProdTy {\alpha} {\beta}} {\beta}}}
             {\cOrdPair
                {x}
                {\cFunAppX 
                   {\mathbf{G}_{\cFunTyX {\cProdTy {\alpha} {\beta}} {\beta}}}
                   {\cOrdPair
                     {y}
                     {s}}}}} 
          {\cFunAppX
            {\mathbf{G}_{\cFunTyX {\cProdTy {\alpha} {\beta}} {\beta}}}
            {\cOrdPair          
               {\cFunAppX 
                  {\mathbf{F}_{\cFunTyX {\cProdTy {\alpha} {\alpha}} {\alpha}}}
                  {\cOrdPair
                     {x}
                     {y}}}
               {s}}}} \And {}$\\
${\cForallX 
    {s} 
    {\mathbf{S}_{\cSetTy {\beta}}}
    {\cEqX 
       {\cFunAppX
          {\mathbf{G}_{\cFunTyX {\cProdTy {\alpha} {\beta}} {\beta}}}
          {\cOrdPair
             {\mathbf{E}_{\alpha}}
             {s}}}
       {s}}}$.\\\hline

${\cMonHomom 
   {\mathbf{M}^{1}_{\cSetTy {\alpha}}}
   {\mathbf{M}^{2}_{\cSetTy {\beta}}}
   {\mathbf{F}^{1}_{\cFunTyX {\cProdTy {\alpha} {\alpha}} {\alpha}}}
   {\mathbf{E}^{1}_{\alpha}} 
   {\mathbf{F}^{2}_{\cFunTyX {\cProdTy {\beta} {\beta}} {\beta}}}
   {\mathbf{E}^{2}_{\beta}}
   {\mathbf{H}_{\cFunTyX {\alpha} {\beta}}}}$\\[1ex]
\hspace{2ex} stands for\\[1ex]
${\cMonoid 
    {\mathbf{M}^{1}_{\cSetTy {\alpha}}}
    {\mathbf{F}^{1}_{\cFunTyX {\cProdTy {\alpha} {\alpha}} {\alpha}}}
    {\mathbf{E}^{1}_{\alpha}}} \And {}$\\
${\cMonoid 
    {\mathbf{M}^{2}_{\cSetTy {\beta}}}
    {\mathbf{F}^{2}_{\cFunTyX {\cProdTy {\beta} {\beta}} {\beta}}}
    {\mathbf{E}^{2}_{\beta}}} \And {}$\\
${\cIsDefInQTyX
    {\mathbf{H}_{\cFunTyX {\alpha} {\beta}}}
    {\cFunQTyX
       {\mathbf{M}^{1}_{\cSetTy {\alpha}}}
       {\mathbf{M}^{2}_{\cSetTy {\beta}}}}} \And {}$\\
${\cForallX 
    {x,y} 
    {\mathbf{M}^{1}_{\cSetTy {\alpha}}}
    {\cEqX
       {\cFunAppX 
          {\mathbf{H}_{\cFunTyX {\alpha} {\beta}}} 
          {\cFunApp  
             {\mathbf{F}^{1}_{\cFunTyX {\cProdTy {\alpha} {\alpha}} {\alpha}}}
             {\cOrdPair {x} {y}}}}
       {\cFunAppX
          {\mathbf{F}^{2}_{\cFunTyX {\cProdTy {\beta} {\beta}} {\beta}}}
          {\cOrdPair
             {\cFunAppX {\mathbf{H}_{\cFunTyX {\alpha} {\beta}}} {x}} 
             {\cFunAppX {\mathbf{H}_{\cFunTyX {\alpha} {\beta}}} {y}}}}}} \And {}$\\
${\cEqX 
    {\cFunAppX {\mathbf{H}_{\cFunTyX {\alpha} {\beta}}} {\mathbf{E}^{1}_{\alpha}}}
    {\mathbf{E}^{2}_{\beta}}}$\\

\hline
\end{tabular}
\ec
\caption{Notational Definitions for Monoids: Abbreviations}
\label{tab:monoids-abbr}
\end{table}

Let $T = (L,\Gamma)$ be a theory\footnote{A \emph{theory} of Alonzo
and related notions are presented in Chapter 9 of~\cite{Farmer25}.} of
Alonzo.  Consider a tuple
\[(\zeta_\alpha, \mathbf{F}_{\cFunTyX {\cProdTy {\alpha} {\alpha}}
  {\alpha}}, \mathbf{E}_{\alpha})\] where (1) $\zeta_\alpha$ is either
a type $\alpha$ of $L$ or a closed quasitype $\mathbf{Q}_{\cSetTy
  {\alpha}}$ of $L$ and (2)~$\mathbf{F}_{\cFunTyX {\cProdTy {\alpha}
    {\alpha}} {\alpha}}$ and $\mathbf{E}_{\alpha}$ are closed
expressions of $L$.  Let $\mathbf{X}_\cB$ be the sentence \[\cMonoid
{\mathbf{M}_{\cSetTy {\alpha}}} {\mathbf{F}_{\cFunTyX {\cProdTy
      {\alpha} {\alpha}} {\alpha}}} {\mathbf{E}_{\alpha}},\] where
\textsf{MONOID} is the abbreviation introduced by the notational
definition given in Table~\ref{tab:monoids-abbr} and
$\mathbf{M}_{\cSetTy {\alpha}}$ is $\cUnivSetPC {\alpha}$ if
$\zeta_\alpha$ is $\alpha$ and is $\mathbf{Q}_{\cSetTy {\alpha}}$
otherwise.  If $T \vDash \mathbf{X}_\cB$, then $(\zeta_\alpha,
\mathbf{F}_{\cFunTyX {\cProdTy {\alpha} {\alpha}} {\alpha}},
\mathbf{E}_{\alpha})$ denotes a monoid $(m,\cdot,e)$ in $T$.  Stated more precisely,
if $T \vDash \mathbf{X}_\cB$, then, for all general models $M$ of $T$
and all assignments $\phi \in \mName{assign}(M)$, $(\zeta_\alpha,
\mathbf{F}_{\cFunTyX {\cProdTy {\alpha} {\alpha}} {\alpha}},
\mathbf{E}_{\alpha})$ denotes the
monoid \[(V^{M}_{\phi}(\mathbf{M}_{\cSetTy {\alpha}}),
V^{M}_{\phi}(\mathbf{F}_{\cFunTyX {\cProdTy {\alpha} {\alpha}}
  {\alpha}}), V^{M}_{\phi}(\mathbf{E}_{\alpha})).\]

Thus we can show that $(\zeta_\alpha, \mathbf{F}_{\cFunTyX {\cProdTy
    {\alpha} {\alpha}} {\alpha}}, \mathbf{E}_{\alpha})$ denotes a
monoid in $T$ by proving $T \vDash \mathbf{X}_\cB$.  However, we may
need general definitions and theorems about monoids to prove
properties in~$T$ about the monoid denoted by $(\zeta_\alpha,
\mathbf{F}_{\cFunTyX {\cProdTy {\alpha} {\alpha}} {\alpha}},
\mathbf{E}_{\alpha})$.  It would be extremely inefficient to state
these definitions and prove these theorems in $T$ since instances of
these same definitions and theorems could easily be needed for other
triples in $T$, as well as in other theories, that denote monoids.

Instead of developing part of a monoid theory in $T$, we should apply
the little theories method and develop a ``little theory'' $T_{\rm
  mon}$ of monoids, separate from~$T$, that has the most convenient
level of abstraction and the most convenient vocabulary for talking
about monoids.  The general definitions and theorems of monoids can
then be introduced in a development\footnote{A \emph{development} of
Alonzo and related notions are presented in Chapter 12
of~\cite{Farmer25}.}  $D_{\rm mon}$ of $T_{\rm mon}$ in a universal
abstract form.  When these definitions and theorems are needed in a
development $D$, a development morphism\footnote{A \emph{theory
morphism} and a \emph{development morphism} of Alonzo are presented in
Sections~14.3 and 14.4, respectively, of~\cite{Farmer25}.} from
$D_{\rm mon}$ to $D$ can be created and then used to transport the
abstract definitions and theorems in $D_{\rm mon}$ to concrete
instances of them in $D$.  The validity of these concrete definitions
and theorems in $D$ is guaranteed by the fact that the abstract
definitions and theorems are valid in the top theory of $D_{\rm mon}$
and the development morphism used to transport them preserves
validity.

We can verify that $(\zeta_\alpha, \mathbf{F}_{\cFunTyX {\cProdTy
    {\alpha} {\alpha}} {\alpha}}, \mathbf{E}_{\alpha})$ denotes a
monoid in $T$ by simply constructing an appropriate theory morphism
$\Phi$ from $T_{\rm mon}$ to $T$.  As a bonus, we can use $\Phi$ to
transport the abstract definitions and theorems in $D_{\rm mon}$ to
concrete instances of them in a development of $T$ whenever they are
needed.  Moreover, we do not have to explicitly prove that a
particular property of $(\zeta_\alpha, \mathbf{F}_{\cFunTyX {\cProdTy
    {\alpha} {\alpha}} {\alpha}}, \mathbf{E}_{\alpha})$, such as
$\mathbf{X}_\cB$, that holds by virtue of $(\zeta_\alpha,
\mathbf{F}_{\cFunTyX {\cProdTy {\alpha} {\alpha}} {\alpha}},
\mathbf{E}_{\alpha})$ denoting a monoid is valid in~$T$; instead, we
only need to show that there is an abstract theorem of $T_{\rm mon}$
that $\Phi$ transports to this property.

The following theory definition module defines a suitably abstract
theory of monoids named \textsf{MON}:

\begin{theory-def}
{Monoids}
{\textsf{MON}.}
{$M$.}
{$\cdot_{\cFunTyX {\cProdTy {M} {M}} {M}},\;  \mathsf{e}_M$.}
{
\be

  \item ${\cForallX {x,y,z} {M} {\cEqX {\cBinX {x} {\cdot} {\cBin {y}
    {\cdot} {z}}} {\cBinX {\cBin {x} {\cdot} {y}} {\cdot} {z}}}}$
    \hfill ($\cdot$ is associative).

  \item ${\cForallX {x} {M} {\cBinBX {\cBinX {\mathsf{e}} {\cdot} {x}}
      {=} {\cBinX {x} {\cdot} {\mathsf{e}}} {=} {x}}}$\hfill
    ($\mathsf{e}$ is an identity element with respect to $\cdot$).

\ee
}
\end{theory-def}

\noindent
Notice that we have employed several notational definitions and
conventions in the axioms --- including dropping the types of the
constants --- for the sake of brevity.  This theory specifies the set
of monoids exactly: The base type $M$, like all types, denotes a
nonempty set $m$; the constant $\cdot_{\cFunTyX {\cProdTy {M} {M}}
  {M}}$ denotes a function $\cdot : (m \times m) \tarrow m$ that is
associative; and the constant $\mathsf{e}_M$ denotes a member $e$ of
$m$ that is an identity element with respect to $\cdot$.

The following development definition module defines a development,
named \textsf{MON-1}, of the theory \textsf{MON}:

\begin{dev-def}
{Monoids 1} 
{\textsf{MON-1}.}  
{\textsf{MON}.}  
{
\bi

  \item[] $\mName{Thm1}$: ${\cMonoid {\cUnivSetPC {M}}
    {\cdot_{\cFunTyX {\cProdTy {M} {M}} {M}}} {\mathsf{e}_M}}$\\
  \phantom{x} \hfill (models of \textsf{MON} define monoids).

  \item[] $\mName{Thm2}$: ${\cTotal {\cdot_{\cFunTyX {\cProdTy {M}
          {M}} {M}}}}$ \hfill ($\cdot$ is total).

  \item[] $\mName{Thm3}$: ${\cForallX {x} {M} {\cImpliesX {\cForall
        {y} {M} {\cBinBX {\cBinX {x} {\cdot} {y}} {=} {\cBinX {y}
            {\cdot} {x}} {=} {y}}} {\cEqX {x}
        {\mathsf{e}}}}}$\\ \phantom{x} \hfill (uniqueness of identity
    element).

  \item[] $\mName{Def1}$: ${\cEqX {\mathsf{submonoid}_{\cFunTyX
        {\cSetTy {M}} {\cB}}} {}}$\\ \hspace*{2ex} ${\cFunAbsX {s}
    {\cSetTy {M}} {\cAndX {\cAndX {\cNotEqX {s} {\cEmpSetPC {M}}}
        {\cIsDefInQTy {\cRestrictX {\cdot} {\cProdTyX {s} {s}}}
          {\cFunTyX {\cProdTy {s} {s}} {s}}}} {\cInX {\mathsf{e}}
        {s}}}}$ \hfill (submonoid).

  \item[] $\mName{Thm4}$: ${\cForallX {s} {\cSetTy {M}} {\cImpliesX
      {\cFunAppX {\mathsf{submonoid}} {s}} {\cMonoid {s}
        {\cdot\wrestricted_{\cProdTyX {s} {s}}}
        {\mathsf{e}}}}}$\\ \phantom{x} \hfill (submonoids are
    monoids).

  \item[] $\mName{Thm5}$: ${\cFunAppX {\mathsf{submonoid}} {\cFinSetL
      {\mathsf{e}}}}$ \hfill (minimum submonoid).

  \item[] $\mName{Thm6}$: ${\cFunAppX {\mathsf{submonoid}}
    {\cUnivSetPC {M}}}$ \hfill (maximum submonoid).

  \item[] $\mName{Def2}$: ${\cEqX {\cdot^{\rm op}_{\cFunTyX {\cProdTy
          {M} {M}} {M}}} {\cFunAbsX {p} {\cProdTyX {M} {M}} {\cBinX
        {\cFunApp {\mathsf{snd}} {p}} {\cdot} {\cFunApp {\mathsf{fst}}
          {p}}}}}$ \hfill (opposite of $\cdot$).

  \item[] $\mName{Thm7}$: ${\cForallX {x,y,z} {M} {\cEqX {\cBinX {x}
        {\cdot^{\rm op}} {\cBin {y} {\cdot^{\rm op}} {z}}} {\cBinX
        {\cBin {x} {\cdot^{\rm op}} {y}} {\cdot^{\rm op}} {z}}}}$\\
    \phantom{x} \hfill ($\cdot^{\rm op}$ is associative).

  \item[] $\mName{Thm8}$: ${\cForallX {x} {M} {\cBinBX {\cBinX
        {\mathsf{e}} {\cdot^{\rm op}} {x}} {=} {\cBinX {x} {\cdot^{\rm
            op}} {\mathsf{e}}} {=} {x}}}$\\ 
   \phantom{x} \hfill ($\mathsf{e}$ is an identity element with
   respect to $\cdot^{\rm op}$).

  \item[] $\mName{Def3}$: ${\cEqX {\odot_{\cFunTyX {\cProdTy {\cSetTy
            {M}} {\cSetTy {M}}} {\cSetTy {M}}}} {\cFunAppX
      {\cSetProdPC {\textsf{set-op}} {M} {M} {M}} {\cdot}}}$\\ 
  \phantom{x} \hfill (set product).

  \item[] $\mName{Def4}$: ${\cEqX {\mathsf{E}_{\cSetTy {M}}}
    {\cFinSetL {\mathsf{e}_M}}}$ \hfill (set identity element).

  \item[] $\mName{Thm9}$: ${\cForallX {x,y,z} {\cSetTy {M}} {\cEqX
      {\cBinX {x} {\odot} {\cBin {y} {\odot} {z}}} {\cBinX {\cBin {x}
          {\odot} {y}} {\odot} {z}}}}$
  \hfill ($\odot$ is  associative).

  \item[] $\mName{Thm10}$: ${\cForallX {x} {\cSetTy {M}} {\cBinBX
      {\cBinX {\mathsf{E}} {\odot} {x}} {=} {\cBinX {x}
        {\odot} {\mathsf{E}}} {=} {x}}}$\\ \phantom{x}
    \hfill ($\mathsf{E}$ is an identity element with
    respect to $\odot$).

\ei 
}
\end{dev-def}

\noindent
${\cSetProdPC {\textsf{set-op}} {M} {M} {M}}$ is an instance of the
parametric pseudoconstant ${\cSetProdPC {\textsf{set-op}} {\alpha}
  {\beta} {\gamma}}$ defined in Table~\ref{tab:monoids-pc}.

\bsp $\mName{Thm1}$ states that each model of \textsf{MON} defines a
monoid.  $\mName{Thm2}$ states that the monoid's binary function is
total (which is implied by the first axiom of \textsf{MON}).
$\mName{Thm3}$ states that a monoid's identity element is
unique. $\mName{Def1}$ defines the notion of a \emph{submonoid} and
\textsf{Thm4--Thm6} are three theorems about submonoids.  Notice that
$\cdot\wrestricted_{\cProdTyX {s} {s}}$, the restriction of $\cdot$ to
$\cProdTyX {s} {s}$, denotes a partial function.  Notice also that
\[{\cIsDefInQTyX {\cRestrictX {\cdot} {\cProdTyX {s} {s}}}
 {\cFunTyX {\cProdTy {s} {s}} {s}}}\] in $\mName{Def1}$ asserts that
$s$ is closed under ${\cRestrictX {\cdot} {\cProdTyX {s} {s}}}$ since
$\cdot$ is total by $\mName{Thm2}$.  $\mName{Def2}$
defines~$\cdot^{\rm op}_{\cFunTyX {\cProdTy {M} {M}} {M}}$, the
\emph{opposite} of $\cdot$, and \textsf{Thm7--Thm8} are key theorems
about~$\cdot^{\rm op}$.  $\mName{Def3}$ defines~$\odot_{\cFunTyX
  {\cProdTy {\cSetTy {M}} {\cSetTy {M}}} {\cSetTy {M}}} $, the
\emph{set product} on ${\cSetTy {M}}$; $\mName{Def4}$
defines~$\mathsf{E}_{\cSetTy {M}}$, the identity element with respect
to $\odot$; and \textsf{Thm9--Thm10} are key theorems about $\odot$.
These four definitions and ten theorems require proofs that show the
RHS of each definition (i.e., the definition's definiens) is defined
and each theorem is valid.  The proofs are given in
Appendix~\ref{app:validation}.  \esp

\section{Transportation of Definitions and Theorems}\label{sec:transportation}

Let $T$ be a theory such that $T \vDash \mathbf{X}_\cB$ where
$\mathbf{X}_\cB$ is the sentence \[{\cMonoid {\mathbf{M}_{\cSetTy
      {\alpha}}} {\mathbf{F}_{\cFunTyX {\cProdTy {\alpha} {\alpha}}
      {\alpha}}} {\mathbf{E}_{\alpha}}},\] and assume that $D$ is some
development of $T$ (which could be $T$ itself).  We would like to show
how the definitions and theorems of the development \textsf{MON-1} can
be transported to $D$.\footnote{A \emph{transportation} is presented
in Subsection 14.4.2 of~\cite{Farmer25}.}

Before considering the general case, we will consider the special case
when $\mathbf{M}_{\cSetTy {\alpha}}$ is $\cUnivSetPC {\alpha}$, which
denotes the entire domain for the type $\alpha$, and
$\mathbf{F}_{\cFunTyX {\cProdTy {\alpha} {\alpha}} {\alpha}}$ and
$\mathbf{E}_{\alpha}$ are constants $\mathbf{c}_{\cFunTyX {\cProdTy
    {\alpha} {\alpha}} {\alpha}}$ and~$\mathbf{d}_{\alpha}$.  We start
by defining a theory morphism from \textsf{MON} to $T$ using a theory
translation definition module:

\begin{theory-trans-def}
{Special \textnormal{\textsf{MON}} to $T$}
{\textsf{special-MON-to-$T$}.}
{\textsf{MON}.}
{$T$.}
{
\be

  \item $M \mapsto \alpha$.

\ee
}
{
\be

  \item ${\cdot_{\cFunTyX {\cProdTy {M} {M}} {M}}} \mapsto
    {\mathbf{c}_{\cFunTyX {\cProdTy {\alpha} {\alpha}} {\alpha}}}$.

  \item ${\mathsf{e}_M} \mapsto {\mathbf{d}_{\alpha}}$.

\ee
}
\end{theory-trans-def}

Since \textsf{special-MON-to-$T$} is a normal translation\footnote{A
\emph{theory translation} and a \emph{development translation} of
Alonzo are presented in Subsections 14.3.1 and 14.4.1, respectively,
of~\cite{Farmer25}.}, it has no obligations of the first kind
by~\cite[Lemma 14.10]{Farmer25} and two obligations of the second kind
which are valid in $T$ by~\cite[Lemma 14.11]{Farmer25}.  It has two
obligations of the third kind corresponding to the two axioms of
\textsf{MON}.  $T \vDash \mathbf{X}_\cB$ implies that each of these
two obligations is valid in $T$.  Therefore,
\textsf{special-MON-to-$T$} is a theory morphism from \textsf{MON} to
$T$ by the Morphism Theorem \cite[Theorem
  14.16]{Farmer25}.\footnote{An \emph{obligation} of a theory
translation and the Morphism Theorem are presented in Subsection
14.3.2 of~\cite{Farmer25}.}

Now we can transport the definitions and theorems of \textsf{MON-1} to
$D$ via \textsf{special-MON-to-$T$} using definition and theorems
transportation modules.  For example, $\mName{Thm3}$ and
$\mName{Def1}$ can be transported using the following two modules:

\begin{thm-transport}
{Transport of \textnormal{\textsf{Thm3}} to $D$}
{\textsf{uniqueness-of-identity-element-via-special-MON-to-$D$}.}
{\textsf{MON-1}.}
{\textsf{$D$}.}
{\textsf{special-MON-to-$T$}.}
{
\bi

  \item[] $\mName{Thm3}$: ${\cForallX {x} {M} {\cImpliesX {\cForall
        {y} {M} {\cBinBX {\cBinX {x} {\cdot} {y}} {=} {\cBinX {y}
            {\cdot} {x}} {=} {y}}} {\cEqX {x}
        {\mathsf{e}}}}}$\\ \phantom{x} \hfill (uniqueness of identity
    element).

\ei
}
{
\bi

  \item[] \textsf{Thm3-via-special-MON-to-$T$}:\\
  \hspace*{2ex}
  ${\cForallX {x} {\alpha} {\cImpliesX {\cForall
        {y} {\alpha} {\cBinBX {\cBinX {x} {\mathbf{c}} {y}} {=} {\cBinX {y}
            {\mathbf{c}} {x}} {=} {y}}} {\cEqX {x}
        {\mathbf{d}}}}}$\\
  \phantom{x} \hfill (uniqueness of identity element).

\ei
}
{\textsf{$D'$}.}
\end{thm-transport}

\begin{def-transport}
{Transport of \textnormal{\textsf{Def1}} to $D'$}
{\textsf{submonoid-via-special-MON-to-$D'$}.}
{\textsf{MON-1}.}
{$D'$.}
{\textsf{special-MON-to-$T$}.}
{
\bi

  \item[] $\mName{Def1}$: 
  ${\cEqX 
      {\mathsf{submonoid}_{\cFunTyX {\cSetTy {M}} {\cB}}} {}}$\\ 
      \hspace*{2ex} 
      ${\cFunAbsX 
          {s}
          {\cSetTy {M}} 
          {\cAndX 
             {\cAndX 
                {\cNotEqX {s} {\cEmpSetPC {M}}}
                {\cIsDefInQTy 
                   {\cRestrictX {\cdot} {\cProdTyX {s} {s}}}
                   {\cFunTyX {\cProdTy {s} {s}} {s}}}} 
             {\cInX {\mathsf{e}} {s}}}}$ \hfill (submonoid).

\ei
}
{
\bi

  \item[] $\textsf{Def1-via-special-MON-to-$T$}$: 
  ${\cEqX 
      {\mathsf{submonoid}_{\cFunTyX {\cSetTy {\alpha}} {\cB}}} {}}$\\ 
      \hspace*{2ex} 
      ${\cFunAbsX 
          {s}
          {\cSetTy {\alpha}} 
          {\cAndX 
             {\cAndX 
                {\cNotEqX {s} {\cEmpSetPC {\alpha}}}
                {\cIsDefInQTy 
                   {\cRestrictX {\mathbf{c}} {\cProdTyX {s} {s}}}
                   {\cFunTyX {\cProdTy {s} {s}} {s}}}} 
             {\cInX {\mathbf{d}} {s}}}}$ \hfill (submonoid).

\ei
}
{\textsf{$D''$}.}
{\textsf{special-MON-1-to-$D'$}.}
\end{def-transport}

We will next consider the general case when $\mathbf{M}_{\cSetTy
  {\alpha}}$ may be different from $\cUnivSetPC {\alpha}$ and
$\mathbf{F}_{\cFunTyX {\cProdTy {\alpha} {\alpha}} {\alpha}}$ and
$\mathbf{E}_{\alpha}$ may not be constants.  The general case is
usually more complicated and less succinct than the special case.
We start again by defining a theory morphism from \textsf{MON} to $T$
using a theory translation definition module:

\begin{theory-trans-def}
{General \textnormal{\textsf{MON}} to $T$}
{\textsf{general-MON-to-$T$}.}
{\textsf{MON}.}
{$T$.}
{
\be

  \item $M \mapsto {\mathbf{M}_{\cSetTy {\alpha}}}$.

\ee
}
{
\be

  \item ${\cdot_{\cFunTyX {\cProdTy {M} {M}} {M}}} \mapsto
    {\mathbf{F}_{\cFunTyX {\cProdTy {\alpha} {\alpha}} {\alpha}}}$.

  \item ${\mathsf{e}_M} \mapsto {\mathbf{E}_{\alpha}}$.

\ee
}
\end{theory-trans-def}

Let $\textsf{general-MON-to-$T$} = (\mu,\nu)$.  Then
\textsf{general-MON-to-$T$} has the following five obligations (one of
the first, two of the second, and two of the third kind):

\be

  \item $\overline{\nu}(\cNotEqX {\cUnivSetPC {M}} {\cEmpSetPC {M}})
    \SynEqual {\cNotEqX {\cFunAbs {x} {\mathbf{M}_{\cSetTy {\alpha}}}
        {\cT}} {\cFunAbs {x} {\mathbf{M}_{\cSetTy {\alpha}}} {\cF}}}$.

  \item $\overline{\nu}(\cIsDefInQTyX {\cdot_{\cFunTyX {\cProdTy {M}
        {M}} {M}}} {\cUnivSetPC {\cFunTyX {\cProdTy {M} {M}} {M}}})
    \SynEqual {}$\\ \hspace*{2ex} ${\cIsDefInQTyX
      {\mathbf{F}_{\cFunTyX {\cProdTy {\alpha} {\alpha}} {\alpha}}}
      {\cFunAbs {x} {\cFunQTyX {\cProdQTy {\mathbf{M}_{\cSetTy
                {\alpha}}} {\mathbf{M}_{\cSetTy {\alpha}}}}
          {\mathbf{M}_{\cSetTy {\alpha}}}} {\cT}}}$.

  \item $\overline{\nu}(\cIsDefInQTyX {\mathsf{e}_M} {\cUnivSetPC
    {M}})\SynEqual {\cIsDefInQTyX {\mathbf{E}_{\alpha}} {\cFunAbs {x}
      {\mathbf{M}_{\cSetTy {\alpha}}} {\cT}}}$.

  \item $\overline{\nu}(\cForallX {x,y,z} {M} {\cEqX {\cBinX {x}
      {\cdot} {\cBin {y} {\cdot} {z}}} {\cBinX {\cBin {x} {\cdot} {y}}
      {\cdot} {z}}}) \SynEqual {}$\\
  \hspace*{2ex}
   ${\cForallX 
      {x, y, z} 
      {\mathbf{M}_{\cSetTy {\alpha}}} {}}$\\
      \hspace*{4ex}
      ${\cEqX 
          {\cFunAppX
             {\mathbf{F}_{\cFunTyX {\cProdTy {\alpha} {\alpha}} {\alpha}}}
             {\cOrdPair
                {x}
                {\cFunAppX 
                   {\mathbf{F}_{\cFunTyX {\cProdTy {\alpha} {\alpha}} {\alpha}}}
                   {\cOrdPair {y} {z}}}}}
          {\cFunAppX
             {\mathbf{F}_{\cFunTyX {\cProdTy {\alpha} {\alpha}} {\alpha}}}
             {\cOrdPair          
                {\cFunAppX 
                   {\mathbf{F}_{\cFunTyX {\cProdTy {\alpha} {\alpha}} {\alpha}}}
                   {\cOrdPair {x} {y}}}
                {z}}}}$.

  \item $\overline{\nu}(\cForallX {x} {M} {\cBinBX {\cBinX
      {\mathsf{e}} {\cdot} {x}} {=} {\cBinX {x} {\cdot} {\mathsf{e}}}
    {=} {x}}) \SynEqual {}$\\
  \hspace*{2ex}
  ${\cForallX
      {x}
      {\mathbf{M}_{\cSetTy {\alpha}}}
      {\cBinBX
         {\cFunAppX 
            {\mathbf{F}_{\cFunTyX {\cProdTy {\alpha} {\alpha}} {\alpha}}}
            {\cOrdPair {\mathbf{E}_{\alpha}} {x}}}
         {=}
         {\cFunAppX
            {\mathbf{F}_{\cFunTyX {\cProdTy {\alpha} {\alpha}} {\alpha}}}
            {\cOrdPair {x} {\mathbf{E}_{\alpha}}}}
         {=}
         {x}}}$.

\ee

\noindent
$\mathbf{A}_\alpha \SynEqual \mathbf{B}_\alpha$ means the expressions
denoted by $\mathbf{A}_\alpha$ and $\mathbf{B}_\alpha$ are identical.

$T \vDash \mathbf{X}_\cB$ implies that each of these obligations is
valid in $T$ as follows.  The first and second conjuncts of
$\mathbf{X}_\cB$ imply that the first obligation is valid in $T$ by
part 3 of \cite[Lemma 14.9]{Farmer25}.  The first and third conjuncts
imply that the second obligation is valid in $T$ by part 5 of
\cite[Lemma 14.9]{Farmer25}.  The first and fourth conjuncts imply
that the third obligation is valid in $T$ by part 5 of \cite[Lemma
  14.9]{Farmer25}.  And the fifth and sixth conjuncts imply,
respectively, that the fourth and fifth obligations are valid in~$T$.
Therefore, \textsf{general-MON-to-$T$} is a theory morphism by the
Morphism Theorem \cite[Theorem 14.16]{Farmer25}.

We can now transport, as before, the definitions and theorems of
\mbox{\textsf{MON-1}} to $D$ via \textsf{general-MON-to-$T$} using
definition and theorem transportation modules, but we can also
transport them using a group transportation module\footnote{This kind
of module transports a set of definitions and theorems as a group in
which order does not matter.  A group transportation has nothing to do
with the algebraic structure called a group.}. For example,
$\mName{Thm3}$ and $\mName{Def1}$ can be transported as a group using
the following group transportation module:

\begin{group-transport}
{Transport of \textnormal{\textsf{Thm3}} and \textnormal{\textsf{Def1}} to $D$}
{\textsf{uniqueness-of-identity-element-and-submonoid-to-$D$}.}
{\textsf{MON-1}.}
{$D$.}
{\textsf{general-MON-to-$T$}.}
{
\bi

  \item[] $\mName{Thm3}$: ${\cForallX {x} {M} {\cImpliesX {\cForall
        {y} {M} {\cBinBX {\cBinX {x} {\cdot} {y}} {=} {\cBinX {y}
            {\cdot} {x}} {=} {y}}} {\cEqX {x}
        {\mathsf{e}}}}}$\\ \phantom{x} \hfill (uniqueness of identity
    element).

  \item[] $\mName{Def1}$: 
  ${\cEqX 
      {\mathsf{submonoid}_{\cFunTyX {\cSetTy {M}} {\cB}}} {}}$\\ 
      \hspace*{2ex} 
      ${\cFunAbsX 
          {s}
          {\cSetTy {M}} 
          {\cAndX 
             {\cAndX 
                {\cNotEqX {s} {\cEmpSetPC {M}}}
                {\cIsDefInQTy 
                   {\cRestrictX {\cdot} {\cProdTyX {s} {s}}}
                   {\cFunTyX {\cProdTy {s} {s}} {s}}}} 
             {\cInX {\mathsf{e}} {s}}}}$ \hfill (submonoid).

\ei
}
{
\bi

  \item[] \textsf{Thm3-via-general-MON-to-$T$}:\\
  \hspace*{2ex}
  ${\cForallX
      {x}
      {\mathbf{M}_{\cSetTy {\alpha}}} {}}$\\
      \hspace*{4ex}
      ${\cImpliesX
          {\cForall
             {y}
             {\mathbf{M}_{\cSetTy {\alpha}}}
             {\cBinBX
                {\cFunAppX 
                   {\mathbf{F}_{\cFunTyX {\cProdTy {\alpha} {\alpha}} {\alpha}}}
                   {\cOrdPair {x} {y}}}
                {=}
                {\cFunAppX
                   {\mathbf{F}_{\cFunTyX {\cProdTy {\alpha} {\alpha}} {\alpha}}}
                   {\cOrdPair {y} {x}}}
                {=}
                {y}}}
          {\cEqX {x} {\mathbf{E}_{\alpha}}}}$\\ 
  \phantom{x} \hfill (uniqueness of identity element).

  \item[] $\textsf{Def1-via-general-MON-to-$T$}$: 
  ${\cEqX 
      {\mathsf{submonoid}_{\cFunTyX {\cSetTy {\alpha}} {\cB}}} {}}$\\ 
      \hspace*{2ex} 
      ${\cFunAbsX 
          {s}
          {\cSetQTy {\mathbf{M}_{\cSetTy {\alpha}}}} {}}$\\
          \hspace*{4ex}
         ${\cNotEqX 
            {s} 
            {\cFunAbs {x} {\mathbf{M}_{\cSetTy {\alpha}}} {\cF}}} \And {}$\\
          \hspace*{4ex}
         ${\cIsDefInQTy 
             {\cRestrictX
                {\mathbf{F}_{\cFunTyX {\cProdTy {\alpha} {\alpha}} {\alpha}}}
                {\cProdTyX {s} {s}}}
             {\cFunTyX {\cProdTy {s} {s}} {s}}} \And {}$\\
          \hspace*{4ex} 
         ${\cInX {\mathbf{E}_{\alpha}} {s}}$ \hfill (submonoid).

\ei
}
{$D'$.}
{\textsf{general-MON-1-to-$D'$}.}
\end{group-transport} 

\noindent
The abbreviation ${\cSetQTy {\mathbf{M}_{\cSetTy {\alpha}}}}$, which
denotes the power set of $\mathbf{M}_{\cSetTy {\alpha}}$, is defined in
Table~\ref{tab:nd-quasitypes}.

\section{Opposite and Set Monoids}\label{sec:op-and-set-monoids}

For every monoid $(m,\cdot,e)$, there is (1) an associated monoid
$(m,\cdot^{\rm op},e)$, where $\cdot^{\rm op}$ is the opposite of
$\cdot$, called the \emph{opposite monoid} of $(m,\cdot,e)$ and (2)~a
monoid $(\sP(m),\odot,\mSet{e})$, where $\sP(m)$ is the power set of
$m$ and $\odot$ is the set product on $\sP(m)$, called the \emph{set
monoid} of $(m,\cdot,e)$.

We will construct a development morphism named
\textsf{MON-to-opposite-monoid} from the theory \textsf{MON} to its
development \textsf{MON-1} that maps \[(M,\cdot_{\cFunTyX {\cProdTy
    {M} {M}} {M}},\mathsf{e})\] to \[(M, {\cdot^{\rm op}_{\cFunTyX
    {\cProdTy {M} {M}} {M}}}, {\mathsf{e}}).\] Then we will be able to
use this morphism to transport abstract definitions and theorems about
monoids to more concrete definitions and theorems about opposite
monoids.  Here is the definition of \textsf{MON-to-opposite-monoid}:

\newpage

\begin{dev-trans-def}
{\textnormal{\textsf{MON}} to Op.~Monoid}
{\textsf{MON-to-opposite-monoid}.}
{\textsf{MON}.}
{\textsf{MON-1}.}
{
\be

  \item $M \mapsto M$.

\ee
}
{
\be

  \item ${\cdot_{\cFunTyX {\cProdTy {M} {M}} {M}}} \mapsto {\cdot^{\rm
      op}_{\cFunTyX {\cProdTy {M} {M}} {M}}}$.

  \item ${\mathsf{e}_M} \mapsto {\mathsf{e}_M}$.

\ee
}
\end{dev-trans-def}

Since \textsf{MON-to-opposite-monoid} is a normal translation, it has
no obligations of the first kind by~\cite[Lemma 14.10]{Farmer25} and
two obligations of the second kind which are valid in the top theory
of \textsf{MON-1} by~\cite[Lemma 14.11]{Farmer25}.  It has two
obligations of the third kind corresponding to the two axioms of
\textsf{MON}.  These two obligations are logically equivalent to
\textsf{Thm7} and \textsf{Thm8}, respectively, in
\mbox{\textsf{MON-1}}, and so these two theorems are obviously valid
in the top theory of \mbox{\textsf{MON-1}}.  Therefore,
\textsf{MON-to-opposite-monoid} is a development morphism from
\textsf{MON} to \textsf{MON-1} by the Morphism Theorem
\cite[Theorem~14.16]{Farmer25}.

We can now transport \textsf{Thm1} via \textsf{MON-to-opposite-monoid}
to show that opposite monoids are indeed monoids:

\begin{thm-transport}
{Transport of \textnormal{\textsf{Thm1}} to \textnormal{\textsf{MON-1}}}
{\textsf{monoid-via-MON-to-opposite-monoid}.}
{\textsf{MON}.}
{\textsf{MON-1}.}
{\textsf{MON-to-opposite-monoid}.}
{
\bi

  \item[] $\mName{Thm1}$: ${\cMonoid {\cUnivSetPC {M}}
    {\cdot_{\cFunTyX {\cProdTy {M} {M}} {M}}} {\mathsf{e}_M}}$\\
  \phantom{x} \hfill (models of \textsf{MON} define monoids).

\ei
}
{
\bi

  \item[] \textsf{Thm11 (Thm1-via-MON-to-opposite-monoid)}:\\
  \hspace*{2ex} ${\cMonoid {\cUnivSetPC {M}} {\cdot^{\rm op}_{\cFunTyX
        {\cProdTy {M} {M}} {M}}} {\mathsf{e}_M}}$ 
  \hfill (opposite monoids are monoids).

\ei
}
{\textsf{MON-2}.}
\end{thm-transport}

\bsp Similarly, we will construct a development morphism named
\textsf{MON-to-set-monoid} from the theory \textsf{MON} to its
development \textsf{MON-2} that maps \[(M,{\cdot_{\cFunTyX {\cProdTy
      {M} {M}} {M}}}, {\mathsf{e}_M})\] to \[({\cSetTy {M}},
{\odot_{\cFunTyX {\cProdTy {\cSetTy {M}} {\cSetTy {M}}} {\cSetTy
      {M}}}}, {\mathsf{E}_{\cSetTy {M}}}).\] Then we will be able to
use this morphism to transport abstract definitions and theorems about
monoids to more concrete definitions and theorems about set monoids.
Here is the definition of \textsf{MON-to-set-monoid}: \esp

\begin{dev-trans-def}
{\textnormal{\textsf{MON}} to Set Monoid}
{\textsf{MON-to-set-monoid}.}
{\textsf{MON}.}
{\textsf{MON-2}.}
{
\be
 
  \item $M \mapsto {\cSetTy {M}}$.

\ee
}
{
\be

  \item $\cdot_{\cFunTyX {\cProdTy {M} {M}} {M}} \mapsto
    {\odot_{\cFunTyX {\cProdTy {\cSetTy {M}} {\cSetTy {M}}} {\cSetTy
          {M}}}}$.

  \item $\mathsf{e}_M \mapsto {\mathsf{E}_{\cSetTy {M}}}$.

\ee
}
\end{dev-trans-def}

Since \textsf{MON-to-set-monoid} is a normal translation, it has no
obligations of the first kind by~\cite[Lemma 14.10]{Farmer25}.  It has
two obligations of the second kind.  The first one is valid in the top
theory of \mbox{\textsf{MON-2}} by part 4 of \cite[Lemma
  14.9]{Farmer25} since ${\odot_{\cFunTyX {\cProdTy {\cSetTy {M}}
      {\cSetTy {M}}} {\cSetTy {M}}}}$ beta-reduces by~\cite[Axiom
  A4]{Farmer25} to a function abstraction which is defined
by~\cite[Axiom A5.11]{Farmer25}.  The second one is valid in the top
theory of \mbox{\textsf{MON-2}} by part 4 of \cite[Lemma
  14.9]{Farmer25} since ${\mathsf{E}_{\cSetTy {M}}}$ is a function
abstraction which is defined by~\cite[Axiom A5.11]{Farmer25}.  It has
two obligations of the third kind corresponding to the two axioms of
\textsf{MON}.  These two obligations are \textsf{Thm9} and
\textsf{Thm10}, respectively, in \mbox{\textsf{MON-2}}, and so these
two theorems are obviously valid in the top theory of
\mbox{\textsf{MON-2}}.  Therefore, \textsf{MON-to-set-monoid} is a
development morphism from \textsf{MON} to \textsf{MON-2} by the
Morphism Theorem \cite[Theorem~14.16]{Farmer25}.

We can now transport \textsf{Thm1} via \textsf{MON-to-set-monoid} to
show that set monoids are indeed monoids:

\begin{thm-transport}
{Transport of \textnormal{\textsf{Thm1}} to \textnormal{\textsf{MON-2}}}
{\textsf{monoid-via-MON-to-set-monoid}.}
{\textsf{MON}.}
{\textsf{MON-2}.}
{\textsf{MON-to-set-monoid}.}
{
\bi

  \item[] $\mName{Thm1}$: ${\cMonoid {\cUnivSetPC {M}}
    {\cdot_{\cFunTyX {\cProdTy {M} {M}} {M}}} {\mathsf{e}_M}}$\\
  \phantom{x} \hfill (models of \textsf{MON} define monoids).

\ei
}
{
\bi

  \item[] \textsf{Thm12 (Thm1-via-MON-to-set-monoid)}:\\
  \hspace*{2ex}
  ${\cMonoid
     {\cUnivSetPC {\cSetTy {M}}}
     {\odot_{\cFunTyX {\cProdTy {\cSetTy {M}} {\cSetTy {M}}} {\cSetTy {M}}}} 
     {\mathsf{E}_{\cSetTy {M}}}}$\\
  \phantom{x} \hfill (set  monoids are monoids).

\ei
}
{\textsf{MON-3}.}
\end{thm-transport}

\section{Commutative Monoids}\label{sec:com-monoids}

A monoid $(m,\cdot,e)$ is \emph{commutative} if $\cdot$ is
commutative.  

Let $\mathbf{Y}_\cB$ be the formula \[\cComMonoid {\mathbf{M}_{\cSetTy
    {\alpha}}} {\mathbf{F}_{\cFunTyX {\cProdTy {\alpha} {\alpha}}
    {\alpha}}} {\mathbf{E}_{\alpha}},\] where \textsf{COM-MONOID} is
the abbreviation introduced by the notational definition given in
Table~\ref{tab:monoids-abbr}.  $\mathbf{Y}_\cB$ asserts that the tuple
\[({\mathbf{M}_{\cSetTy {\alpha}}}, {\mathbf{F}_{\cFunTyX {\cProdTy
      {\alpha} {\alpha}} {\alpha}}}, {\mathbf{E}_{\alpha}})\] denotes
a commutative monoid $(m,\cdot,e)$.

We can define a theory of commutative monoids, named \textsf{COM-MON},
by adding an axiom that says $\cdot$ is commutative to the theory
\textsf{MON} using a theory extension module:

\begin{theory-ext}
{Commutative Monoids}
{\textsf{COM-MON}.}
{\textsf{MON}.}
{}
{}
{
\be \setcounter{enumi}{2}

  \item ${\cForallX {x,y} {M} {\cEqX {\cBinX {x} {\cdot} {y}} {\cBinX
        {y} {\cdot} {x}}}}$ \hfill ($\cdot$ is commutative).

\ee
}
\end{theory-ext}

Then we can develop the theory \textsf{COM-MON} using the following
development definition module:

\begin{dev-def}
{Commutative Monoids 1}
{\textsf{COM-MON-1}.}
{\textsf{COM-MON}.}
{
\bi

  \item[] $\mName{Thm13}$: ${\cComMonoid {\cUnivSetPC {M}}
    {\cdot_{\cFunTyX {\cProdTy {M} {M}} {M}}} {\mathsf{e}_M}}$\\
  \phantom{x} \hfill (models of \textsf{COM-MON} define commutative monoids).

  \item[] $\mName{Def5}$: ${\cEqX {\mathsf{\le}_{\cFunTyBX {M} {M}
        {\cB}}} {\cFunAbsX {x,y} {M} {\cForsomeX {z} {M} {\cEqX
          {\cBinX {x} {\cdot} {z}} {y}}}}}$ \hfill (weak order).

  \item[] $\mName{Thm14}$: ${\cForallX {x} {M} {\cBinX {x} {\le} {x}}}$\hfill
    (reflexivity).

  \item[] $\mName{Thm15}$: ${\cForallX {x,y,z} {M} {\cImpliesX {\cAnd
        {\cBinX {x} {\le} {y}} {\cBinX {y} {\le} {z}}} {\cBinX {x}
        {\le} {z}}}}$ \hfill (transitivity).

\ei
}
\end{dev-def}

\noindent
\textsf{Thm13} states that each model of \textsf{COM-MON} defines a
commutative monoid.  \textsf{Def5} defines a weak (nonstrict) order
that is a pre-order by \textsf{Thm14} and \textsf{Thm15}.  We could
have put \textsf{Def5}, \textsf{Thm14}, and \textsf{Thm15} in a
development of \textsf{MON} since \textsf{Thm14} and \textsf{Thm15} do
not require that $\cdot$ is commutative, but we have put these in
\textsf{COM-MON} instead since ${\mathsf{\le}_{\cFunTyBX {M} {M}
    {\cB}}}$ is more natural for commutative monoids than for
noncommutative monoids.

Since \textsf{COM-MON} is an extension of \textsf{MON}, there is an
inclusion (i.e.,~a theory morphism whose mapping is the identity
function) from \textsf{MON} to \textsf{COM-MON}.  This inclusion is
defined by the following theory translation definition module:

\newpage

\begin{theory-trans-def}
{\textnormal{\textsf{MON}} to \textnormal{\textsf{COM-MON}}}
{\textsf{MON-to-COM-MON}.}
{\textsf{MON}.}
{\textsf{COM-MON}.}
{
\be

  \item $M \mapsto M$.

\ee
}
{
\be

  \item ${\cdot_{\cFunTyX {\cProdTy {M} {M}} {M}}} \mapsto
    {\cdot_{\cFunTyX {\cProdTy {M} {M}} {M}}}$.

  \item ${\mathsf{e}_M} \mapsto {\mathsf{e}_M}$.

\ee
}
\end{theory-trans-def}

\noindent
We will assume that, whenever we define a theory extension $T'$ of a
theory $T$, we also simultaneously define the inclusion from $T$ to
$T'$.

Since \textsf{MON-to-COM-MON} is an inclusion from \textsf{MON} to
\textsf{COM-MON}, it is also a development morphism from
\textsf{MON-3} to \textsf{COM-MON-1} and the definitions and theorems
of \textsf{MON-3} can be freely transported verbatim to
\textsf{COM-MON-1}.  In the rest of the paper, when a theory $T'$ is
an extension of a theory $T$ and $D$ is a development of $T$, we will
assume that the definitions and theorems of $D$ are also definitions
and theorems of any trivial or nontrivial development of $T'$ without
explicitly transporting them via the inclusion from $T$ to $T'$ as
long as there are no name clashes.  This assumption is given the name
\emph{inclusion transportation convention} in~\cite[Subsection
  14.4.3]{Farmer25}.

\section{Transformation Monoids}\label{sec:trans-monoids}

A very important type of monoid is a monoid composed of
transformations of a set.  Let $s$ be a nonempty set.  Then
$(f,\circ,\mathsf{id})$, where $f$ is a set of (partial or total)
functions from $s$ to $s$, \[\circ: ((s \tarrow s) \times (s \tarrow
s)) \tarrow (s \tarrow s)\] is function composition, and $\mathsf{id}
: s \tarrow s$ is the identity function, is a \emph{transformation
monoid on $s$} if $f$ is closed under $\circ$ and $\mathsf{id} \in f$.
It is easy to verify that every transformation monoid is a monoid.  If
$f$ contains every function in the function space $s \tarrow s$, then
$(f,\circ,\mathsf{id})$ is clearly a transformation monoid which is
called the \emph{full transformation monoid on $s$}.  Let us say that
a transformation monoid $(f,\circ,\mathsf{id})$ is \emph{standard} if
$f$ contains only total functions.  In many developments, nonstandard
transformation monoids are ignored, but there is no reason to do that
here since Alonzo admits undefined expressions and partial functions.

Consider the following theory \textsf{ONE-BT} of one base type:

\begin{theory-def}
{One Base Type}
{\textsf{ONE-BT}.}
{$S$.}
{}
{}
\end{theory-def}

\noindent
We can define the notion of a transformation monoid in a development
of this theory, but we must first introduce some general facts about
function composition.  To do that, we need a theory \textsf{FUN-COMP}
with four base types in order to state the associativity theorem for
function composition in full generality:

\begin{theory-def}
{Function Composition}
{\textsf{FUN-COMP}.}
{$A,B,C,D$.}
{}
{}
\end{theory-def}

We introduce two theorems for function composition in a development of
\textsf{FUN-COMP}:

\begin{dev-def}
{Function Composition 1}
{\textsf{FUN-COMP-1}.}
{\textsf{FUN-COMP}.}
{
\bi

  \item[] $\mName{Thm16}$: 
  ${\cForallCX 
     {f} {\cFunTyX {\cBaseTy A} {\cBaseTy B}} 
     {g} {\cFunTyX {\cBaseTy B} {\cBaseTy C}} 
     {h} {\cFunTyX {\cBaseTy C} {\cBaseTy D}} 
     {\cEqX 
        {\cFunCompX {f} {\cFunComp {g} {h}}}
        {\cFunCompX {\cFunComp {f} {g}} {h}}}}$\\
  \phantom{x} \hfill ($\circ$ is associative).

  \item[] $\mName{Thm17}$:
  ${\cForallX 
     {f} {\cFunTyX {\cBaseTy A} {\cBaseTy B}} 
     {\cBinBX 
        {\cFunCompX {\cIdFunPC {A}} {f}}
        {=}
        {\cFunCompX {f} {\cIdFunPC {B}}}
        {=}
        {f}}}$\\
  \phantom{x} \hfill (identity functions are left and right identity elements).

\ei
}
\end{dev-def}

\noindent
The parametric pseudoconstants ${\cFunCompPairPC {\alpha} {\beta}
  {\gamma}}$ and ${\cIdFunPC {\alpha}}$ are defined in
Tables~\ref{tab:monoids-pc} and~\ref{tab:nd-functions}, respectively.
The infix notation for the application of
\[{\cFunCompPairPC {\alpha} {\beta} {\gamma}}\] is also defined in
Table~\ref{tab:monoids-pc}.

Next we define a theory morphism from \textsf{FUN-COMP} to
\textsf{ONE-BT}:

\begin{theory-trans-def}
{\textnormal{\textsf{FUN-COMP}} to \textnormal{\textsf{ONE-BT}}}
{\textsf{FUN-COMP-to-ONE-BT}.}
{\textsf{FUN-COMP}.}
{\textsf{ONE-BT}.}
{
\be
 
  \item $A \mapsto S$.

  \item $B \mapsto S$.

  \item $C \mapsto S$.

  \item $D \mapsto S$.

\ee
}
{}
\end{theory-trans-def}

The translation \textsf{FUN-COMP-to-ONE-BT} is clearly a theory
morphism by the Morphism Theorem \cite[Theorem 14.16]{Farmer25} since
it is a normal translation and \textsf{FUN-COMP} contains no constants
or axioms.  So we can transport the theorems of \textsf{FUN-COMP-1} to
\textsf{ONE-BT} via \textsf{FUN-COMP-to-ONE-BT}:

\begin{group-transport}
{\small Transport of \textnormal{\textsf{Thm16--Thm17}} to
  \textnormal{\textsf{ONE-BT}}}
{\textsf{function-composition-theorems-via-FUN-COMP-to-ONE-BT}.}
{\textsf{FUN-COMP-1}.}  {\textsf{ONE-BT}.}
{\textsf{FUN-COMP-to-ONE-BT}.}  { \bi

  \item[] $\mName{Thm16}$:
  ${\cForallCX 
     {f} {\cFunTyX {\cBaseTy A} {\cBaseTy B}} 
     {g} {\cFunTyX {\cBaseTy B} {\cBaseTy C}} 
     {h} {\cFunTyX {\cBaseTy C} {\cBaseTy D}} 
     {\cEqX 
        {\cFunCompX {f} {\cFunComp {g} {h}}}
        {\cFunCompX {\cFunComp {f} {g}} {h}}}}$\\
  \phantom{x} \hfill ($\circ$ is associative). 

  \item[] $\mName{Thm17}$:
  ${\cForallX 
     {f} {\cFunTyX {\cBaseTy A} {\cBaseTy B}} 
     {\cBinBX 
        {\cFunCompX {\cIdFunPC {A}} {f}}
        {=}
        {\cFunCompX {f} {\cIdFunPC {B}}}
        {=}
        {f}}}$\\
  \phantom{x} \hfill (identity functions are left and right identity elements).

\ei
}
{
\bi

  \item[] \textsf{Thm18 (Thm16-via-FUN-COMP-to-ONE-BT)}:\\
  \hspace*{2ex}
  ${\cForallX 
     {f,g,h} 
     {\cFunTyX {\cBaseTy S} {\cBaseTy S}} 
     {\cEqX 
        {\cFunCompX {f} {\cFunComp {g} {h}}}
        {\cFunCompX {\cFunComp {f} {g}} {h}}}}$
  \hfill ($\circ$ is associative).

  \item[] \textsf{Thm19 (Thm17-via-FUN-COMP-to-ONE-BT)}:\\
  \hspace*{2ex}
  ${\cForallX 
     {f} {\cFunTyX {\cBaseTy S} {\cBaseTy S}} 
     {\cBinBX 
        {\cFunCompX {\cIdFunPC {S}} {f}}
        {=}
        {\cFunCompX {f} {\cIdFunPC {S}}}
        {=}
        {f}}}$\\
  \phantom{x} \hfill (${\cIdFunPC {S}}$ is an identity element with respect to $\circ$).

\ei
}
{\textsf{ONE-BT-1}.}
{\textsf{FUN-COMP-1-to-ONE-BT-1}.}
\end{group-transport}

We can obtain the theorem that all transformation monoids are monoids
almost for free by transporting results from \textsf{MON-1} to
\textsf{ONE-BT-1}.  We start by creating the theory morphism from
\textsf{MON} to \textsf{ONE-BT} that maps \[(M,\cdot_{\cFunTyX {\cProdTy
    {M} {M}} {M}},\mathbf{e}_M)\] to
\[(\cFunTyX {S} {S}, \cFunCompPairPC {S} {S} {S}, \cIdFunPC {S}):\]

\begin{theory-trans-def}
{\textnormal{\textsf{MON} to \textnormal{\textsf{ONE-BT}}}}
{\textsf{MON-to-ONE-BT}.}  {\textsf{MON}.}  {\textsf{ONE-BT}.}  { \be
 
  \item $M \mapsto {\cFunTyX {S} {S}}$.

\ee
}
{
\be

  \item $\cdot_{\cFunTyX {\cProdTy {M} {M}} {M}} \mapsto
    {\cFunCompPairPC {S} {S} {S}}$.

  \item $\mathsf{e}_M \mapsto \cIdFunPC {S}$.

\ee
}
\end{theory-trans-def}

\bsp The theory translation \textsf{MON-to-ONE-BT} is normal so that
it has no obligations of the first kind by~\cite[Lemma
  14.10]{Farmer25}.  It has two obligations of the second kind.  These
are valid in \textsf{ONE-BT} by part 4 of \cite[Lemma 14.9]{Farmer25}
since $\cFunCompPairPC {S} {S} {S}$ and $\cIdFunPC {S}$ are function
abstractions which are defined by~\cite[Axiom A5.11]{Farmer25}.  It
has two obligations of the third kind corresponding to the two axioms
of \textsf{MON}.  The two obligations are \textsf{Thm18} and
\textsf{Thm19}, respectively, in \textsf{ONE-BT-1}, and so these two
theorems are obviously valid in the top theory of \textsf{ONE-BT-1}.
Therefore, \textsf{MON-to-ONE-BT} is a theory morphism from
\textsf{MON} to \textsf{ONE-BT} by the Morphism Theorem
\cite[Theorem~14.16]{Farmer25}.
 
We can transport \textsf{Def1}, the definition of
${\mathsf{submonoid}_{\cFunTyX {\cSetTy {M}} {\cB}}}$, and
\textsf{Thm4}, the theorem that says all submonoids are monoids, to
\textsf{ONE-BT-1} via \mbox{\textsf{MON-to-ONE-BT}} by a group
transportation module:

\esp
\begin{group-transport}
{\small Transport of \textnormal{\textsf{Def1}} \& \textnormal{\textsf{Thm2}} 
to \textnormal{\textsf{ONE-BT-1}}}
{\textsf{submonoids-via-MON-to-ONE-BT}.}
{\textsf{MON-1}.}
{\textsf{ONE-BT-1}.}
{\textsf{MON-to-ONE-BT}.}
{
\bi

  \item[] $\mName{Def1}$: 
  ${\cEqX 
      {\mathsf{submonoid}_{\cFunTyX {\cSetTy {M}} {\cB}}} {}}$\\ 
      \hspace*{2ex} 
      ${\cFunAbsX 
          {s}
          {\cSetTy {M}} 
          {\cAndX 
             {\cAndX 
                {\cNotEqX {s} {\cEmpSetPC {M}}}
                {\cIsDefInQTy 
                   {\cRestrictX {\cdot} {\cProdTyX {s} {s}}}
                   {\cFunTyX {\cProdTy {s} {s}} {s}}}} 
             {\cInX {\mathsf{e}} {s}}}}$ \hfill (submonoid).

  \item[] $\mName{Thm4}$: 
  ${\cForallX 
      {s} 
      {\cSetTy {M}}    
      {\cImpliesX
         {\cFunAppX {\mathsf{submonoid}} {s}} 
         {\cMonoid 
            {s}
            {\cdot\wrestricted_{\cProdTyX {s} {s}}} 
            {\mathsf{e}}}}}$\\ 
  \phantom{x} \hfill (submonoids are monoids).

\ei
}
{
\bi

  \item[] \textsf{Def6 (Def1-via-MON-to-ONE-BT)}:
  ${\cEqX 
      {\textsf{trans-monoid}_{\cFunTyX {\cSetTy {\cFunTyX {S} {S}}} {\cB}}} {}}$\\
      \hspace*{2ex}
     ${\cFunAbsX 
         {s} 
         {\cSetTy {\cFunTyX {S} {S}}} {}}$\\
         \hspace*{4ex}
        ${\cNotEqX {s} {\cEmpSetPC {\cFunTyX {S} {S}}}} \And {}$\\
         \hspace*{4ex}
        ${\cIsDefInQTy 
            {\cRestrictX 
               {\cFunCompPairPC {S} {S} {S}}
               {\cProdTyX {s} {s}}}
            {\cFunTyX {\cProdTy {s} {s}} {s}}} \And {}$\\
         \hspace*{4ex}
        ${\cInX {\cIdFunPC {S}} {s}}$
  \hfill (transformation monoid).

  \item[] \textsf{Thm20 (Thm4-via-MON-to-ONE-BT)}:\\
  \hspace*{2ex}
  ${\cForallX 
      {s} 
      {\cSetTy {\cFunTyX {S} {S}}} {}}$\\
      \hspace*{4ex}
     ${\cImpliesX
         {\cFunAppX {\textsf{trans-monoid}} {s}} 
         {\cMonoid 
            {s}
            {{\cFunCompPairPC {S} {S} {S}}\wrestricted_{\cProdTyX {s} {s}}}
            {\cIdFunPC {S}}}}$\\ 
  \phantom{x} \hfill (transformation monoids are monoids).

\ei
}
{\textsf{ONE-BT-2}.}
{\textsf{MON-1-to-ONE-BT-2}.}
\end{group-transport} 

\noindent
\textsf{trans-monoid} is a predicate that is true when it is applied
to a set of functions of ${\cFunTyX {S} {S}}$ that forms a transformation
monoid.  \textsf{Thm20} says that every transformation monoid ---
including the full transformation monoid --- is a monoid.

\section{Monoid Actions}\label{sec:mon-actions}

A \emph{(left) monoid action} is a mathematical structure
$(m,s,\cdot,e,\mName{act})$ where $(m,\cdot,e)$ is a monoid and
$\mName{act} : (m \times s) \tarrow s$ is a function such that
\[(1)\sglsp x \mathrel{\mName{act}} (y \mathrel{\mName{act}} z) = (x \cdot y)
\mathrel{\mName{act}} z\] for all $x,y \in m$ and $z \in s$
and \[(2)\sglsp e \mathrel{\mName{act}} z = z\] for all $z \in s$.  We
say in this case that the monoid $(m,\cdot,e)$ \emph{acts on} the set
$s$ \emph{by} the function~$\mName{act}$.

Let $\mathbf{Z}_\cB$ be the formula \[\cMonAction {\mathbf{M}_{\cSetTy
    {\alpha}}} {\mathbf{S}_{\cSetTy {\beta}}} {\mathbf{F}_{\cFunTyX
    {\cProdTy {\alpha} {\alpha}} {\alpha}}} {\mathbf{E}_{\alpha}}
{\mathbf{G}_{\cFunTyX {\cProdTy {\alpha} {\beta}} {\beta}}},\] where
\textsf{MON-ACTION} is the abbreviation introduced by the notational
definition given in Table~\ref{tab:monoids-abbr}.  $\mathbf{Z}_\cB$
asserts that the tuple \[({\mathbf{M}_{\cSetTy {\alpha}}},
{\mathbf{S}_{\cSetTy {\beta}}}, {\mathbf{F}_{\cFunTyX {\cProdTy
      {\alpha} {\alpha}} {\alpha}}}, {\mathbf{E}_{\alpha}},
{\mathbf{G}_{\cFunTyX {\cProdTy {\alpha} {\beta}} {\beta}}})\] denotes
a monoid action $(m,s,\cdot,e,\mName{act})$.

A theory of monoid actions is defined as an extension of the theory of
monoids:

\begin{theory-ext}
{Monoid Actions} 
{\textsf{MON-ACT}.}  
{\textsf{MON}.}  
{$S$.}
{$\mName{act}_{\cFunTyX {\cProdTy {M} {S}} {S}}$.}  
{ 
\be \setcounter{enumi}{2}

  \item ${\cForallBX {x,y} {M} {s} {S} {\cEqX {\cBinX {x}
        {\mName{act}} {\cBin {y} {\mName{act}} {s}}} {\cBinX {\cBin
          {x} {\cdot} {y}} {\mName{act}} {s}}}}$\\ \phantom{x} \hfill
    ($\mName{act}$ is compatible with $\cdot$).

  \item ${\cForallX {s} {S} {\cEqX {\cBinX {\mName{e}} {\mName{act}}
        {s}} {s}}}$ \hfill ($\mName{act}$ is compatible with
    $\mName{e}$).

\ee
}
\end{theory-ext}

We begin a development of \textsf{MON-ACT} by adding the definitions
and theorems below:

\begin{dev-def}
{Monoid Actions 1}
{\textsf{MON-ACT-1}.}
{\textsf{MON-ACT}.}
{
\bi

  \item[] $\mName{Thm21}$: ${\cMonAction {\cUnivSetPC {M}}
    {\cUnivSetPC {S}} {\cdot_{\cFunTyX {\cProdTy {M} {M}} {M}}}
    {\mathsf{e}_M} {\mName{act}_{\cFunTyX {\cProdTy {M} {S}}
        {S}}}}$\\ 
  \phantom{x} \hfill (models of \textsf{MON-ACT} define monoid actions).

  \item[] $\mName{Thm22}$: ${\cTotal {\mName{act}_{\cFunTyX {\cProdTy
          {M} {S}} {S}}}}$ \hfill ($\mName{act}$ is total).

  \item[] $\mName{Def7}$: ${\cEqX {\mName{orbit}_{\cFunTyX {S}
        {\cSetTy {S}}}} {\cFunAbsX {s} {S} {\cSet {t} {S} {\cForsomeX
          {x} {M} {\cEqX {\cBinX {x} {\mName{act}} {s}}
            {t}}}}}}$\hfill (orbit).

  \item[] $\mName{Def8}$: ${\cEqX {\mName{stabilizer}_{\cFunTyX {S}
        {\cSetTy {M}}}} {\cFunAbsX {s} {S} {\cSet {x} {M} {\cEqX
          {\cBinX {x} {\mName{act}} {s}} {s}}}}}$\hfill (stabilizer).

  \item[] $\mName{Thm23}$: ${\cForallX {s} {S} {\cFunAppX
      {\mName{submonoid}} {\cFunApp {\mName {stabilizer}} {s}}}}$
    \hfill (stabilizers are submonoids).

\ei
}
\end{dev-def}

\noindent
$\mName{Thm21}$ states that each model of \textsf{MON-ACTION} defines
a monoid action.  $\mName{Thm22}$ says that $\mName{act}_{\cFunTyX
  {\cProdTy {M} {S}} {S}}$ is total (which is implied by the third
axiom of \textsf{MON-ACTION}).  $\mName{Def7}$ and $\mName{Def8}$
introduce the concepts of an orbit and a stabilizer.  And
$\mName{Thm23}$ states that a stabilizer of a monoid action
$(m,s,\cdot,e,\mName{act})$ is a submonoid of the monoid
$(m,\cdot,e)$.  The power of this machinery --- monoid actions with
orbits and stabilizers --- is low with arbitrary monoids but very high
with groups, i.e., monoids in which every element has an inverse.

Monoid actions are common in monoid theory.  We will present two
important examples of monoid actions.  The first is the monoid action
$(m,m,\cdot,e,\cdot)$ such that the monoid $(m,\cdot,e)$ acts on the
set $m$ of its elements by its function $\cdot$.  We formalize this by
creating the theory morphism from \textsf{MON-ACT} to \textsf{MON}
that maps \[(M,S,{\cdot_{\cFunTyX {\cProdTy {M} {M}} {M}}},
{\mathbf{e}_M}, {\mName{act}_{\cFunTyX {\cProdTy {M} {S}} {S}}})\]
to \[(M, M, {\cdot_{\cFunTyX {\cProdTy {M} {M}} {M}}}, {\mathbf{e}_M},
{\cdot_{\cFunTyX {\cProdTy {M} {M}} {M}}}):\]

\begin{theory-trans-def}
{\textnormal{\textsf{MON-ACT}} to \textnormal{\textsf{MON}}}
{\textsf{MON-ACT-to-MON}.}
{\textsf{MON-ACT}.}
{\textsf{MON}.}
{
\be

  \item $M \mapsto M$.

  \item $S \mapsto M$.

\ee
}
{
\be

  \item ${\cdot_{\cFunTyX {\cProdTy {M} {M}} {M}}} \mapsto
    {\cdot_{\cFunTyX {\cProdTy {M} {M}} {M}}}$.

  \item ${\mathsf{e}_M} \mapsto {\mathsf{e}_M}$.

  \item ${\mName{act}_{\cFunTyX {\cProdTy {M} {S}} {S}}} \mapsto
    {\cdot_{\cFunTyX {\cProdTy {M} {M}} {M}}}$.

\ee
}
\end{theory-trans-def}

\noindent
It is an easy exercise to verify, arguing as we have above, that
\textsf{MON-ACT-to-MON} is a theory morphism.

We can now transport \textsf{Thm21} from \textsf{MON-ACT} to
\textsf{MON-3} via \textsf{MON-ACT-to-MON} to show that the action of
a monoid $(m,\cdot,e)$ on $m$ by $\cdot$ is a monoid action:

\begin{thm-transport}
{Transport of \textnormal{\textsf{Thm21}} to \textnormal{\textsf{MON-3}}}
{\textsf{monoid-action-via-MON-ACT-to-MON}.}
{\textsf{MON-ACT}.}
{\textsf{MON-3}.}
{\textsf{MON-ACT-to-MON}.}
{
\bi

  \item[] $\mName{Thm21}$: ${\cMonAction {\cUnivSetPC {M}}
    {\cUnivSetPC {S}} {\cdot_{\cFunTyX {\cProdTy {M} {M}} {M}}}
    {\mathsf{e}_M} {\mName{act}_{\cFunTyX {\cProdTy {M} {S}}
        {S}}}}$\\ 
  \phantom{x} \hfill (models of \textsf{MON-ACT} define monoid actions).

\ei
}
{
\bi

  \item[] \textsf{Thm24 (Thm21-via-MON-ACT-to-MON)}:\\
  \hspace*{2ex} ${\cMonAction {\cUnivSetPC {M}} {\cUnivSetPC {M}}
    {\cdot_{\cFunTyX {\cProdTy {M} {M}} {M}}} {\mathsf{e}_M}
    {\cdot_{\cFunTyX {\cProdTy {M} {M}} {M}}}}$\\ 
  \phantom{x} \hfill (first example is a monoid action).

\ei
}
{\textsf{MON-4}.}
\end{thm-transport}

The second example is a standard transformation monoid
$(f,\circ,\mathsf{id})$ on $s$ acting on $s$ by the function that
applies a transformation to a member of $s$.  (Note that all the
functions in $f$ are total by virtue of the transformation monoid
being standard.)  We formalize this example as a theory morphism from
\textsf{MON-ACT} to \textsf{ONE-BT} extended with a set constant that
denotes a standard transformation monoid.  Here is the extension with
a set constant $\mathsf{F}_{\cSetTy {\cFunTyX {S} {S}}}$ and two
axioms:

\begin{theory-ext}
{One Base Type with a Set Constant}
{\textsf{ONE-BT-with-SC}.}
{\textsf{ONE-BT}.}
{}
{$\mathsf{F}_{\cSetTy {\cFunTyX {S} {S}}}$.}
{
\be 

  \item ${\cFunAppX {\textsf{trans-monoid}} {\mathsf{F}}}$ \hfill
    ($\mathsf{F}$ forms a transformation monoid).

  \item ${\cForallX {f} {\mathsf{F}} {\cTotal {f}}}$ \hfill (the
    members of $\mathsf{F}$ are total functions).

\ee
}
\end{theory-ext}

\noindent
And here is the theory morphism from \textsf{MON-ACT} to
\textsf{ONE-BT-with-SC} that maps \[(M,S,{\cdot_{\cFunTyX
    {\cProdTy {M} {M}} {M}}}, {\mathbf{e}_M}, {\mName{act}_{\cFunTyX
    {\cProdTy {M} {S}} {S}}})\] to \[(\mathsf{F}_{\cSetTy {\cFunTyX
    {S} {S}}}, S, {\circ_{\cFunTyX {\cProdTy {\cFunTy {S} {S}}
      {\cFunTy {S} {S}}} {\cFunTy {S} {S}}}}{\wrestricted_{\cProdQTyX
    {\mathsf{F}} {\mathsf{F}}}}, {\cIdFunPC {S}}, {\cFunAppPairPC {S}
  {S}}{\wrestricted_{\cProdQTyX {\mathsf{F}} {S}}}):\]

\begin{theory-trans-def}
{\textnormal{\textsf{MON-ACT}} to \textnormal{\textsf{ONE-BT-with-SC}}}
{\textsf{MON-ACT-to-ONE-BT-with-SC}.}
{\textsf{MON-ACT}.}
{\textsf{ONE-BT-with-SC}.}
{
\be

  \item $M \mapsto {\mathsf{F}_{\cSetTy {\cFunTyX {S} {S}}}}$.

  \item $S \mapsto S$.

\ee
}
{
\be

  \item ${\cdot_{\cFunTyX {\cProdTy {M} {M}} {M}}} \mapsto
    {\circ_{\cFunTyX {\cProdTy {\cFunTy {S} {S}} {\cFunTy {S} {S}}}
        {\cFunTy {S} {S}}}}{\wrestricted_{\cProdQTyX {\mathsf{F}}
        {\mathsf{F}}}}$.

  \item ${\mathsf{e}_M} \mapsto {\cIdFunPC {S}}$.

  \item ${\mName{act}_{\cFunTyX {\cProdTy {M} {S}} {S}}} \mapsto
    {\cFunAppPairPC {S} {S}}{\wrestricted_{\cProdQTyX {\mathsf{F}}
        {S}}}$.

\ee
}
\end{theory-trans-def}

\noindent
The parametric pseudoconstant ${\cFunAppPairPC {S}
  {S}}{\wrestricted_{\cProdQTyX {\mathsf{F}} {S}}}$ is defined in
Table~\ref{tab:monoids-pc}.  It is a straightforward exercise to
verify, arguing as we have above, that
\textsf{MON-ACT-to-ONE-BT-with-SC} is a theory morphism.

\bsp
We can now transport \textsf{Thm21} from \textsf{MON-ACT} to
\textsf{ONE-BT-with-S} via \textsf{MON-ACT-to-ONE-BT-with-SC} to show
that a standard transformation monoid $(f,\circ,\mathsf{id})$ on $s$
acting on $s$ by the function that applies a (total) transformation to
a member of $s$ is a monoid action:
\esp

\begin{thm-transport}
{\small Trans.~of \textnormal{\textsf{Thm21}} to \textnormal{\textsf{ONE-BT-with-SC}}}
{\textsf{monoid-action-via-MON-ACT-to-ONE-BT-with-SC}.}
{\textsf{MON-ACT}.}
{\textsf{ONE-BT-with-SC}.}
{\textsf{MON-ACT-to-ONE-BT-with-SC}.}
{
\bi

  \item[] $\mName{Thm21}$: ${\cMonAction {\cUnivSetPC {M}}
    {\cUnivSetPC {S}} {\cdot_{\cFunTyX {\cProdTy {M} {M}} {M}}}
    {\mathsf{e}_M} {\mName{act}_{\cFunTyX {\cProdTy {M} {S}}
        {S}}}}$\\ 
  \phantom{x} \hfill (models of \textsf{MON-ACT} define monoid actions).

\ei
}
{
\bi

  \item[] \textsf{Thm25 (Thm21-via-MON-ACT-to-ONE-BT-with-SC)}:\\
  \hspace*{2ex} 
  $\textsf{MON-ACTION}(
  {\mathsf{F}_{\cSetTy {\cFunTyX {S} {S}}}},\\
  \hspace*{18.2ex}
  {\cUnivSetPC {S}},\\ 
  \hspace*{18.2ex}
  {\circ_{\cFunTyX {\cProdTy {\cFunTy {S} {S}} {\cFunTy {S} {S}}}
        {\cFunTy {S} {S}}}}{\wrestricted_{\cProdQTyX {\mathsf{F}}
        {\mathsf{F}}}},\\
  \hspace*{18.2ex}
  {\cIdFunPC {S}},\\
  \hspace*{18.2ex}
  {\cFunAppPairPC {S} {S}}{\wrestricted_{\cProdQTyX {\mathsf{F}} {S}}})$\\
  \phantom{x} \hfill (second example is a monoid action).

\ei
}
{\textsf{ONE-BT-with-SC-1}.}
\end{thm-transport}

\section{Monoid Homomorphisms}\label{sec:mon-hom}

Roughly speaking, a \emph{monoid homomorphism} is a
structure-preserving mapping from one monoid to another.  

Let $\mathbf{W}_\cB$ be the formula \[{\cMonHomom 
   {\mathbf{M}^{1}_{\cSetTy {\alpha}}}
   {\mathbf{M}^{2}_{\cSetTy {\beta}}}
   {\mathbf{F}^{1}_{\cFunTyX {\cProdTy {\alpha} {\alpha}} {\alpha}}}
   {\mathbf{E}^{1}_{\alpha}} 
   {\mathbf{F}^{2}_{\cFunTyX {\cProdTy {\beta} {\beta}} {\beta}}}
   {\mathbf{E}^{2}_{\beta}}
   {\mathbf{H}_{\cFunTyX {\alpha} {\beta}}}},\] 
where \textsf{MON-HOMOM} is the abbreviation introduced by the
notational definition given in Table~\ref{tab:monoids-abbr}.
$\mathbf{W}_\cB$ asserts that the tuple \[({\mathbf{M}^{1}_{\cSetTy
    {\alpha}}}, {\mathbf{M}^{2}_{\cSetTy {\beta}}},
{\mathbf{F}^{1}_{\cFunTyX {\cProdTy {\alpha} {\alpha}} {\alpha}}},
{\mathbf{E}^{1}_{\alpha}}, {\mathbf{F}^{2}_{\cFunTyX {\cProdTy {\beta}
      {\beta}} {\beta}}}, {\mathbf{E}^{2}_{\beta}},
{\mathbf{H}_{\cFunTyX {\alpha} {\beta}}})\] denotes a mathematical
structure $(m_1,m_2,\cdot_1,e_1,\cdot_2,e_2,h)$ where
$(m_1,\cdot_1,e_1)$ is a monoid, $(m_2,\cdot_2,e_2)$ is a monoid, and
$h : m_1 \tarrow m_2$ is a monoid homomorphism from
$(m_1,\cdot_1,e_1)$ to $(m_2,\cdot_2,e_2)$.

The notion of a monoid homomorphism is captured in the theory
\textsf{MON-HOM}:

\begin{theory-def}
{{\bf Monoid Homomorphisms}} 
{\textsf{MON-HOM}.}  
{$M_1,M_2$.}
{$\cdot_{\cFunTyX {\cProdTy {M_1} {M_1}} {M_1}},\; \mathsf{e}_{M_1},\;
  \cdot_{\cFunTyX {\cProdTy {M_2} {M_2}} {M_2}},\; \mathsf{e}_{M_2},\;
  \mathsf{h}_{\cFunTyX {M_1} {M_2}}$.}
{
\be

  \item ${\cForallX {x,y,z} {M_1} {\cEqX {\cBinX {x} {\cdot} {\cBin {y}
    {\cdot} {z}}} {\cBinX {\cBin {x} {\cdot} {y}} {\cdot} {z}}}}$
    \hfill ($\cdot_{\cFunTyX {\cProdTy {M_1} {M_1}} {M_1}}$ is associative).

  \item ${\cForallX {x} {M_1} {\cBinBX {\cBinX {\mathsf{e}} {\cdot}
        {x}} {=} {\cBinX {x} {\cdot} {\mathsf{e}}} {=} {x}}}$ \hfill
    ($\mathsf{e}_{M_1}$ is an identity element).

  \item ${\cForallX {x,y,z} {M_2} {\cEqX {\cBinX {x} {\cdot} {\cBin {y}
    {\cdot} {z}}} {\cBinX {\cBin {x} {\cdot} {y}} {\cdot} {z}}}}$
    \hfill ($\cdot_{\cFunTyX {\cProdTy {M_2} {M_2}} {M_2}}$ is associative).

  \item ${\cForallX {x} {M_2} {\cBinBX {\cBinX {\mathsf{e}} {\cdot}
        {x}} {=} {\cBinX {x} {\cdot} {\mathsf{e}}} {=} {x}}}$ \hfill
    ($\mathsf{e}_{M_2}$ is an identity element).

  \item ${\cForallX 
            {x,y} 
            {M_1}
            {\cEqX
               {\cFunAppX {\mathsf{h}} {\cBin {x} {\cdot} {y}}}
               {\cBinX 
                  {\cFunApp {\mathsf{h}} {x}} 
                  {\cdot} 
                  {\cFunApp {\mathsf{h}} {y}}}}}$
            \hfill (first homomorphism property).

  \item ${\cEqX 
            {\cFunAppX {\mathsf{h}} {\mathsf{e}_{M_1}}}
            {\mathsf{e}_{M_2}}}$ 
            \hfill (second homomorphism property).

\ee
}
\end{theory-def}

\noindent
$\mathsf{h}_{\cFunTyX {M_1} {M_2}}$ denotes a monoid homomorphism from
the monoid denoted by
\[(M_1, \cdot_{\cFunTyX {\cProdTy {M_1} {M_1}} {M_1}}, \mathsf{e}_{M_1})\] 
to the monoid denoted by
\[(M_2, \cdot_{\cFunTyX {\cProdTy {M_2} {M_2}} {M_2}}, \mathsf{e}_{M_2}).\]

\noindent
Here is a simple development of \textsf{MON-HOM}:

\begin{dev-def}
{Monoid Homomorphisms 1}
{\textsf{MON-HOM-1}.}
{\textsf{MON-HOM}.}
{
\bi

  \item[] $\mName{Thm26}$:\\
  \hspace*{2ex} 
  $\textsf{MON-HOM}(
  {\cUnivSetPC {M_1}},\\
  \hspace*{15.2ex}
  {\cUnivSetPC {M_2}},\\
  \hspace*{15.2ex}
  {\cdot_{\cFunTyX {\cProdTy {M_1} {M_1}} {M_1}}},\\
  \hspace*{15.2ex}
  {\mathsf{e}_{M_1}},\\
  \hspace*{15.2ex}
  {\cdot_{\cFunTyX {\cProdTy {M_2} {M_2}} {M_2}}},\\
  \hspace*{15.2ex}
  {\mathsf{e}_{M_2}},\\
  \hspace*{15.2ex}
  {\mathsf{h}_{\cFunTyX {M_1} {M_2}}})$\\ 
  \phantom{x} \hfill (models of \textsf{MON-HOM} define monoid homomorphisms).

  \item[] $\mName{Thm27}$: ${\cTotal {\mathsf{h}_{\cFunTyX {M_1}
        {M_2}}}}$ \hfill (${\mathsf{h}_{\cFunTyX {M_1} {M_2}}}$ is
    total).

\ei
}
\end{dev-def}

There are embeddings (i.e., theory morphisms whose mappings are
injective) from \textsf{MON} to the two copies of \textsf{MON} within
\textsf{MON-HOM} defined by the following two theory translation
definitions:

\begin{theory-trans-def}
{First \textnormal{\textsf{MON}} to \textnormal{\textsf{MON-HOM}}}
{\textsf{first-MON-to-MON-HOM}.}
{\textsf{MON}.}
{\textsf{MON-HOM}.}
{
\be

  \item $M \mapsto M_1$.

\ee
}
{
\be

  \item ${\cdot_{\cFunTyX {\cProdTy {M} {M}} {M}}} \mapsto
    {\cdot_{\cFunTyX {\cProdTy {M_1} {M_1}} {M_1}}}$.

  \item ${\mathsf{e}_M} \mapsto {\mathsf{e}_{M_1}}$.

\ee
}
\end{theory-trans-def}

\begin{theory-trans-def}
{Second \textnormal{\textsf{MON}} to \textnormal{\textsf{MON-HOM}}}
{\textsf{second-MON-to-MON-HOM}.}
{\textsf{MON}.}
{\textsf{MON-HOM}.}
{
\be

  \item $M \mapsto M_2$.

\ee
}
{
\be

  \item ${\cdot_{\cFunTyX {\cProdTy {M} {M}} {M}}} \mapsto
    {\cdot_{\cFunTyX {\cProdTy {M_2} {M_2}} {M_2}}}$.

  \item ${\mathsf{e}_M} \mapsto {\mathsf{e}_{M_2}}$.

\ee
}
\end{theory-trans-def}

An example of a monoid homomorphism from the monoid denoted by
\[(M,\cdot_{\cFunTyX {\cProdTy {M} {M}} {M}}, \mathsf{e}_M)\] to the
monoid denoted by \[(\cSetTy {M}, \odot_{\cFunTyX {\cProdTy {\cSetTy
      {M}} {\cSetTy {M}}} {\cSetTy {M}}}, \mathsf{E}_{\cSetTy {M}})\]
is the function that maps a member $x$ of the denotation of $M$ to the
singleton~$\mSet{x}$.  This monoid homomorphism is formalized by the
following development morphism: 

\begin{dev-trans-def}
{\textnormal{\textsf{MON-HOM}} to \textnormal{\textsf{MON}}}
{\textsf{MON-HOM-to-MON-4}.}
{\textsf{MON-HOM}.}
{\textsf{MON-4}.}
{
\be

  \item $M_1 \mapsto M$.

  \item $M_2 \mapsto \cSetTy {M}$.

\ee
}
{
\be

  \item $\cdot_{\cFunTyX {\cProdTy {M_1} {M_1}} {M_1}} \mapsto
    \cdot_{\cFunTyX {\cProdTy {M} {M}} {M}}$.

  \item $\mathsf{e}_{M_1} \mapsto \mathsf{e}_{M}$.

  \item $\cdot_{\cFunTyX {\cProdTy {M_2} {M_2}} {M_2}} \mapsto
    \odot_{\cFunTyX {\cProdTy {\cSetTy {M}} {\cSetTy {M}}} {\cSetTy {M}}}$.

  \item $\mathsf{e}_{M_2} \mapsto \mathsf{E}_{\cSetTy {M}}$.

  \item $\mathsf{h}_{\cFunTyX {M_1} {M_2}} \mapsto \cFunAbsX {x} {M}
    {\cFinSetL {x}}$.

\ee
}
\end{dev-trans-def}

\noindent
It is a straightforward exercise to verify that
\textsf{HOM-MON-to-MON-4} is a theory morphism by the arguments we
employed above.

We can now transport \textsf{Thm26} from \textsf{MON-HOM} to
\textsf{MON-4} via \textsf{MON-HOM-to-MON-4} to show the example is a
monoid homomorphism:

\begin{thm-transport}
{Transport of \textnormal{\textsf{Thm26}} to \textnormal{\textsf{MON-4}}}
{\textsf{monoid-action-via-MON-HOM-to-MON-4}.}
{\textsf{MON-HOM}.}
{\textsf{MON-4}.}
{\textsf{MON-HOM-to-MON-4}.}
{
\bi

  \item[] $\mName{Thm26}$: \\
  \hspace*{2ex} 
  $\textsf{MON-HOM}(
  {\cUnivSetPC {M_1}},\\
  \hspace*{15.2ex}
  {\cUnivSetPC {M_2}},\\
  \hspace*{15.2ex}
  {\cdot_{\cFunTyX {\cProdTy {M_1} {M_1}} {M_1}}},\\
  \hspace*{15.2ex}
  {\mathsf{e}_{M_1}},\\
  \hspace*{15.2ex}
  {\cdot_{\cFunTyX {\cProdTy {M_2} {M_2}} {M_2}}},\\
  \hspace*{15.2ex}
  {\mathsf{e}_{M_2}},\\
  \hspace*{15.2ex}
  {\mathsf{h}_{\cFunTyX {M_1} {M_2}}})$\\ 
  \phantom{x} \hfill (models of \textsf{MON-HOM} define monoid homomorphisms).

\ei
}
{
\bi

  \item[] \textsf{Thm28 (Thm26-via-MON-HOM-to-MON-4)}\\
  \hspace*{2ex} 
  $\textsf{MON-HOM}(
  {\cUnivSetPC {M}},\\
  \hspace*{15.2ex}
  {\cUnivSetPC {\cSetTy {M}}},\\
  \hspace*{15.2ex}
  {\cdot_{\cFunTyX {\cProdTy {M} {M}} {M}}},\\
  \hspace*{15.2ex}
  {\mathsf{e}_{M}},\\
  \hspace*{15.2ex}
  {\odot_{\cFunTyX {\cProdTy {\cSetTy {M}} {\cSetTy {M}}} {\cSetTy {M}}}},\\
  \hspace*{15.2ex}
  {\mathsf{E}_{\cSetTy {M}}},\\
  \hspace*{15.2ex}
  {\cFunAbsX {x} {M} {\cFinSetL {x}}})$
  \hfill (example is a monoid homomorphism).

\ei
}
{\textsf{MON-5}.}
\end{thm-transport}

\section{Monoids over Real Number Arithmetic}\label{sec:monoids-with-reals}

We need machinery concerning real number arithmetic to express some
concepts about monoids.  For instance, an iterated product operator
for monoids involves integers.  To formalize these kinds of concepts,
we need a theory of monoids that includes real number arithmetic.
Chapter 13 of~\cite{Farmer25} presents~\textsf{COF}, a theory of
complete ordered fields.  \textsf{COF} is categorical in the standard
sense (see~\cite{Farmer25}).  That is, it has a single standard model
up to isomorphism that defines the structure of real number
arithmetic.

We define a theory of monoids over \textsf{COF} by extending
\textsf{COF} with the language and axioms of \textsf{MON}:

\begin{theory-ext}
{Monoids over \textnormal{\textsf{COF}}}
{\textsf{MON-over-COF}.}
{\textsf{COF}.}
{$M$.}
{$\cdot_{\cFunTyX {\cProdTy {M} {M}} {M}},\;  \mathsf{e}_M$.}
{
\be \setcounter{enumi}{18}

  \item ${\cForallX {x,y,z} {M} {\cEqX {\cBinX {x} {\cdot} {\cBin {y}
    {\cdot} {z}}} {\cBinX {\cBin {x} {\cdot} {y}} {\cdot} {z}}}}$
    \hfill ($\cdot$ is associative).

  \item ${\cForallX {x} {M} {\cBinBX {\cBinX {\mathsf{e}} {\cdot} {x}}
      {=} {\cBinX {x} {\cdot} {\mathsf{e}}} {=} {x}}}$\hfill
    ($\mathsf{e}$ is an identity element).

\ee
}
\end{theory-ext}

We can now define an iterated product operator for monoids in a
development of \textsf{MON-over-COF-1}:

\begin{dev-def}
{Monoids over \textnormal{\textsf{COF}} 1}
{\textsf{MON-over-COF-1}.}
{\textsf{MON-over-COF}.}
{
\bi

  \item[] $\mName{Def9}$: ${\cEqX {\mName{prod}_{\cFunTyCX {R} {R} {\cFunTy {R} {M}} {M}}} {}}$\\
    \hspace*{2ex}
   ${\cDefDesX {f} 
    {\cFunTyCX {Z_{\cSetTy {R}}} {Z_{\cSetTy {R}}} {\cFunTy {Z_{\cSetTy {R}}} {M}} {M}} {}}$\\
    \hspace*{4ex}${\cForallBX {m,n} {Z_{\cSetTy {R}}} {g} {\cFunTyX {Z_{\cSetTy {R}}} {M}}
    {\cQuasiEqX {\cFunAppCX {f} {m} {n} {g}} {}}}$\\
    \hspace*{6ex}${\cIf {\cBinX {m} {>} {n}} {\mathsf{e}} 
    {\cBinX {\cFunAppC {f} {m} {\cBin {n} {-} {1}} {g}} {\cdot} 
    {\cFunApp {g} {n}}}}$ \hfill (iterated product).

  \item[] $\mName{Thm29}$: ${\cForallBX {m} {Z_{\cSetTy {R}}} {g}
    {\cFunTyX {Z_{\cSetTy {R}}} {M}} {\cQuasiEqX {\cProd {i} {m} {m}
        {\cFunAppX {g} {i}}}{\cFunAppX {g} {m}}}}$\\
    \phantom{x} \hfill (trivial product).

  \item[] $\mName{Thm30}$: 
  ${\cForallBX 
      {m,k,n} 
      {Z_{\cSetTy {R}}} 
      {g}
      {\cFunTyX {Z_{\cSetTy {R}}} {M}} {}}$\\
      \hspace*{2ex}
     ${\cImpliesX
         {\cBinBX {m} {<} {k} {<} {n}}
         {\cQuasiEqX 
           {\cBinX
              {\cProd {i} {m} {k} {\cFunAppX {g} {i}}}
              {\cdot}
              {\cProd {i} {k+1} {n} {\cFunAppX {g} {i}}}}
           {\cProdX {i} {m} {n} {\cFunAppX {g} {i}}}}}$\\
     \phantom{x} \hfill (extended iterated product).

\ei
}
\end{dev-def}

\noindent
We are utilizing the notation for the iterated product operator
defined in Table~\ref{tab:monoids-product}.  $Z_{\cSetTy {R}}$ is a
quasitype defined in the development \textsf{COF-dev-2} of
\textsf{COF} found in~\cite{Farmer25} that denotes the set of
integers.  ($Z_{\cSetTy {R}}$ is automatically available in
\textsf{MON-over-COF} by the inclusion transportation convention
presented in Section~\ref{sec:com-monoids}.)  $\mName{Def9}$ defines
the iterated product operator, and $\mName{Thm29}$ and $\mName{Thm30}$
are two theorems about the operator.

\begin{table}
\bc
\begin{tabular}{|lll|}
\hline

  $\cProd {\mathbf{i}} {\mathbf{M}_R} {\mathbf{N}_R} {\mathbf{A}_M}$
& stands for
& $\cFunAppCX 
     {\mName{prod}_{\cFunTyCX {R} {R} {\cFunTy {R} {M}} {M}}} {} {} {}$\vspace{-1.5ex}\\
&
&    \hspace*{2ex}
    ${\mathbf{M}_R}\,
     {\mathbf{N}_R}\, 
     {\cFunAbs {\mathbf{i}} {R} {\mathbf{A}_M}}$.\\

\hline
\end{tabular}
\ec
\caption{Notational Definition for Monoids: Iterated Product Operator}
\label{tab:monoids-product}
\end{table}

We can similarly define extensions of \textsf{MON} over \textsf{COF}.
For example, here is a theory of commutative monoids over \textsf{COF}
and a development of it:

\begin{theory-ext}
{Commutative Monoids over \textnormal{\textsf{COF}}}
{\textsf{COM-MON-over-COF}.}
{\textsf{MON-over-COF}.}
{}
{}
{
\be \setcounter{enumi}{20}

  \item ${\cForallX {x,y} {M} {\cEqX {\cBinX {x} {\cdot} {y}} {\cBinX
    {y} {\cdot} {x}}}}$ \hfill ($\cdot$ is commutative).

\ee
}
\end{theory-ext}

\begin{dev-def}
{Com.\ Monoids over \textnormal{\textsf{COF}} 1}
{\textsf{COM-MON-over-COF-1}.}
{\textsf{COM-MON-over-COF}.}
{
\bi

  \item[] $\mName{Thm31}$: 
  ${\cForallBX 
      {m,n} 
      {Z_{\cSetTy {R}}} 
      {g,h}
      {\cFunTyX {Z_{\cSetTy {R}}} {M}} {}}$\\
      \hspace*{2ex}
     ${\cQuasiEqX 
        {\cBinX
           {\cProd {i} {m} {n} {\cFunAppX {g} {i}}}
           {\cdot}
           {\cProd {i} {m} {n} {\cFunAppX {h} {i}}}}
        {\cProdX
           {i} 
           {m} 
           {n} 
           {\cBinX
              {\cFunApp {g} {i}}
              {\cdot}
              {\cFunApp {h} {i}}}}}$\\
     \phantom{x} \hfill (product of iterated products).

\ei
}
\end{dev-def}

\noindent
Notice that this theorem holds only if $\cdot$ is commutative.

For another example, here is a theory of commutative monoid actions
over \textsf{COF} and a development of it:

\begin{theory-ext}
{\small Commutative Monoid Actions over \textnormal{\textsf{COF}}}
{\textsf{COM-MON-ACT-over-COF}.}
{\textsf{COM-MON-over-COF}.}
{$S$.}
{$\mName{act}_{\cFunTyX {\cProdTy {M} {S}} {S}}$.} 
{
\be \setcounter{enumi}{21}

  \item ${\cForallBX {x,y} {M} {s} {S} {\cEqX {\cBinX {x}
        {\mName{act}} {\cBin {y} {\mName{act}} {s}}} {\cBinX {\cBin
          {x} {\cdot} {y}} {\mName{act}} {s}}}}$\\ \phantom{x} \hfill
    ($\mName{act}$ is compatible with $\cdot$).

  \item ${\cForallX {s} {S} {\cEqX {\cBinX {\mName{e}} {\mName{act}}
        {s}} {s}}}$ \hfill ($\mName{act}$ is compatible with
    $\mName{e}$).

\ee
}
\end{theory-ext}

\begin{dev-def}
{\small Com.~Monoid Actions over \textnormal{\textsf{COF}} 1}
{\textsf{COM-MON-ACT-over-COF-1}.}
{\textsf{COM-MON-ACT-over-COF}.}
{
\bi

  \item[] $\mName{Thm32}$: 
  ${\cForallBX      
     {x,y} 
     {M} 
     {s} 
     {S} 
     {\cEqX    
        {\cBinX {x} {\mName{act}} {\cBin {y} {\mName{act}} {s}}} 
        {\cBinX {y} {\mName{act}} {\cBin {x} {\mName{act}} {s}}}}}$\\
   \phantom{x} \hfill ($\mName{act}$ has commutative-like property).

\ei
}
\end{dev-def}

\section{Monoid Theory Applied to Strings}\label{sec:strings}

In this section we will show how the machinery of our monoid theory
formalization can be applied to a theory of strings over an abstract
alphabet.  A string over an alphabet $A$ is a finite sequence of
values from $A$.  The finite sequence $s$ can be represented as a
partial function $s : \bN \tarrow A$ such that, for some $n \in \bN$,
$s(m)$ is defined iff $m < n$.

In Table~\ref{tab:nd-sequences} we introduce compact notation for
finite (and infinite) sequences represented in this manner.  The
notation requires a system of natural numbers as defined in Chapter 11
of~\cite{Farmer25}.  We also introduce some special
notation for strings in Table~\ref{tab:monoids-special}.

\begin{table}[t]\footnotesize
\bc
\begin{tabular}{|lll|}
\hline

  $\cSequencesPC {\alpha} {\beta}$
  \index{8ma@$\cSequencesPC {\alpha} {\beta}$}
& stands for
& $\cFunQTyX {\mathbf{C}^{N}_{\cSetTy {\alpha}}} {\beta}$.\\

  $\cSeqQTy {\beta}$
  \index{8mb@$\cSeqQTy {\beta}$}
& stands for
& $\cSequencesPC {\alpha} {\beta}$.\\

  $\cStreamsPC {\alpha} {\beta}$
  \index{8mc@$\cStreamsPC {\alpha} {\beta}$}
& stands for
& $\cSet {s} {\cSeqQTy {\beta}} {\cTotal {s}}$.\\

  $\cSeqInfQTy {\beta}$
  \index{8md@$\cSeqInfQTy {\beta}$}
& stands for
& $\cStreamsPC {\alpha} {\beta}$.\\

  $\cListsPC {\alpha} {\beta}$
  \index{8me@$\cListsPC {\alpha} {\beta}$}
& stands for
& $\{s : {\cSeqQTy {\beta}} \mid 
  {\cForsomeQTyX {n} {\mathbf{C}^{N}_{\cSetTy {\alpha}}}
  {\cForallQTyX {m} {\mathbf{C}^{N}_{\cSetTy {\alpha}}} {}}}$\\
&
& \hspace*{2ex}${\cIffX {\cIsDefX {\cFunApp {s} {m}}} 
  {\cFunAppBX {\mathbf{C}^{\le}_{\cFunTyBX {\alpha} {\alpha} {\cB}}} {m} 
  {\cFunApp {\mathbf{C}^{P}_{\cFunTyX {\alpha} {\alpha}}} {n}}}}\}$.\\

  $\cSeqFinQTy {\beta}$
  \index{8mf@$\cSeqFinQTy {\beta}$}
& stands for
& $\cListsPC {\alpha} {\beta}$.\\

  $\cConsPC {\alpha} {\beta}$
  \index{8mg@$\cConsPC {\alpha} {\beta}$}
& stands for
& $\cFunAbsX {x} {\beta} {\cFunAbsQTyX {s} {\cSeqQTy {\beta}} 
  {\cFunAbsQTyX {n} {\mathbf{C}^{N}_{\cSetTy {\alpha}}}} {}}$\\
&
& \hspace*{2ex}$\cIfX {\cEqX {n} {\mathbf{C}^{0}_\alpha}} {x} 
   {\cFunAppX {s} {\cFunApp 
   {\mathbf{C}^{P}_{\cFunTyX {\alpha} {\alpha}}} {n}}}$.\\

  $\cCons {\mathbf{A}_\beta} {\mathbf{B}_{\cFunTyX {\alpha} {\beta}}}$
  \index{8mh@$\cCons {\mathbf{A}_\beta} {\mathbf{B}_{\cFunTyX {\alpha} {\beta}}}$}
& stands for
& $\cFunAppBX {\cConsPC {\alpha} {\beta}}
  {\mathbf{A}_\beta} {\mathbf{B}_{\cFunTyX {\alpha} {\beta}}}$.\\

\iffalse
  $\cHdPC {\alpha} {\beta}$
  \index{8mi@$\cHdPC {\alpha} {\beta}$}
& stands for
& $\cFunAbsQTyX {s} {\cSeqQTy {\beta}} {\cDefDesQTyX {x} {\beta}
  {\cForsomeQTyX {t} {\cSeqQTy {\beta}} {\cEqX {s} {\cCons {x} {t}}}}}$.\\

  $\cTlPC {\alpha} {\beta}$
  \index{8mj@$\cTlPC {\alpha} {\beta}$}
& stands for
& $\cFunAbsQTyX {s} {\cSeqQTy {\beta}} {\cDefDesQTyX {t} {\cSeqQTy {\beta}}
  {\cForsomeQTyX {x} {\beta}} {\cEqX {s} {\cCons {x} {t}}}}$.\\
\fi

  $\cNilPC {\alpha} {\beta}$
  \index{8mk@$\cNilPC {\alpha} {\beta}$}
& stands for
& $\cEmpFunPC {\alpha} {\beta}$.\\

  $\cEmpListPC {\alpha} {\beta}$
  \index{8ml@$\cEmpListPC {\alpha} {\beta}$}
& stands for
& $\cNilPC {\alpha} {\beta}$.\\

  $\cListL {\mathbf{A}_\beta}$
  \index{8mm@$\cListL {\mathbf{A}_\beta}$}
& stands for
& $\cCons {\mathbf{A}_\beta} {\cEmpListPC {\alpha} {\beta}}$.\\

  $\cListL {\mathbf{A}^{1}_{\beta}, \ldots, \mathbf{A}^{n}_{\beta}}$
  \index{8mn@$\cListL {\mathbf{A}^{1}_{\beta}, \ldots, \mathbf{A}^{n}_{\beta}}$}
& stands for
& $\cCons {\mathbf{A}^{1}_{\beta}} 
  {\cListL {\mathbf{A}^{2}_{\beta}, \ldots, \mathbf{A}^{n}_{\beta}}}$
  {\sglsp}where $n \ge 2$.\\

  $\cLenPC {\alpha} {\beta}$
  \index{8mo@$\cLenPC {\alpha} {\beta}$}
& stands for
& $\cDefDesQTyX {f} {\cFunQTyX {\cSeqFinQTy {\beta}} 
  {\mathbf{C}^{N}_{\cSetTy {\alpha}}}} {}$\\
&
& \hspace*{2ex}${\cAndX
  {\cEqX {\cFunAppX {f} {\cEmpListPC {\alpha} {\beta}}}
  {\mathbf{C}^{0}_\alpha}} {}}$\\
&
& \hspace*{2ex}${\cForallQTyBX {x} {\beta} {s} {\cSeqFinQTy {\beta}} {}}$\\
&
& \hspace*{4ex}${\cEqX 
  {\cFunAppX {f} {\cCons {x} {s}}}
  {\cFunAppBX 
  {\mathbf{C}^{+}_{\cFunTyBX {\alpha} {\alpha} {\alpha}}}
  {\cFunApp {f} {s}}
  {\cFunApp {\mathbf{C}^{S}_{\cFunTyX {\alpha} {\alpha}}}
  {\mathbf{C}^{0}_\alpha}}}}$.\\

  $\cLen {\mathbf{A}_{\cFunTyX {\alpha} {\beta}}}$
  \index{8mp@$\cLen {\mathbf{A}_{\cFunTyX {\alpha} {\beta}}}$}
& stands for
& $\cFunAppX {\cLenPC {\alpha} {\beta}}
  {\mathbf{A}_{\cFunTyX {\alpha} {\beta}}}$.\\

  $\cAppendPC {\alpha} {\beta}$
  \index{8mq@$\cAppendPC {\alpha} {\beta}$}
& stands for
& $\cDefDesQTyX {f} 
     {\cFunTyBX {\cSeqFinQTy {\beta}} {\cSeqFinQTy {\beta}} {\cSeqFinQTy {\beta}}} {}$\\
&
& \hspace*{2ex}
    $\cAndX 
       {\cForallX 
          {t} 
          {\cSeqFinQTy {\beta}} 
          {\cEqX 
             {\cFunAppBX {f} {\cEmpListPC {\alpha} {\beta}} {t}} 
             {t}}}
       {}$\\
&
& \hspace*{2ex}
    $\cForallBX
       {x}
       {\beta}
       {s,t} 
       {\cSeqFinQTy {\beta}} 
       {\cEqX 
          {\cFunAppBX {f} {\cCons {x} {s}} {t}}
          {\cCons {x} {\cFunAppBX {f} {s} {t}}}}$.\\

\iffalse
  $\cNlistsPC {\alpha} {\beta}$
  \index{8mr@$\cNlistsPC {\alpha} {\beta}$}
& stands for
& $\cFunAbsQTyX {n} {\mathbf{C}^{N}_{\cSetTy {\alpha}}} 
  {\cSet {s} {\cSeqFinQTy {\beta}} {\cEqX {\cLen {s}} {n}}}$.\\

  $\cSeqNFinQTy {\beta} {\mathbf{N}_\alpha}$
  \index{8ms@$\cSeqNFinQTy {\beta} {\mathbf{N}_\alpha}$}
& stands for
& $\cFunAppX {\cNlistsPC {\alpha} {\beta}}
  {\mathbf{N}_\alpha}$.\\
\fi

\hline
\end{tabular}
\ec
\caption{Notational Definitions for Sequences}
\label{tab:nd-sequences}
\end{table}

\begin{table}[t]
\bc
\begin{tabular}{|lll|}
\hline

  $\cCatApp {\mathbf{X}_{\cFunTyX {R} {A}}} {\mathbf{Y}_{\cFunTyX {R} {A}}}$
& stands for
& $\cBinX 
     {\mathbf{X}_{\cFunTyX {R} {A}}} 
     {\textsf{cat}}
     {\mathbf{Y}_{\cFunTyX {R} {A}}}$.\\

  $\cSetCatApp 
     {\mathbf{S}_{\cSetTy {\cFunTyX {R} {A}}}} 
     {\mathbf{T}_{\cSetTy {\cFunTyX {R} {A}}}}$
& stands for
& $\cBinX
     {\mathbf{S}_{\cSetTy {\cFunTyX {R} {A}}}}
     {\textsf{set-cat}}
     {\mathbf{T}_{\cSetTy {\cFunTyX {R} {A}}}}$.\\

  $\cIterCat {\mathbf{i}} {\mathbf{M}_R} {\mathbf{N}_R} {\mathbf{A}_{\cFunTyX {R} {A}}}$
& stands for
& $\cFunAppCX 
     {\textsf{iter-cat}_{\cFunTyCX {R} {R} {\cFunTy {R} {\cFunTy {R} {A}}} {\cFunTy {R} {A}}}}
     {} {} {}$\vspace{-1ex}\\
&
&    \hspace*{2ex}
    ${\mathbf{M}_R}\,
     {\mathbf{N}_R}\,
     {\cFunAbs {\mathbf{i}} {R} {\mathbf{A}_{\cFunTyX {R} {A}}}}$.\\

\hline
\end{tabular}
\ec
\caption{Notational Definitions for Monoids: Special Notation}
\label{tab:monoids-special}
\end{table}

The development \mbox{\textsf{COF-dev-2}} of the theory \textsf{COF}
presented in Chapter~13 of~\cite{Farmer25} includes a system of
natural numbers~\cite[Proposition 13.11]{Farmer25}.  Therefore, we can define a
theory of strings as an extension of \textsf{COF} plus a base type $A$
that represents an abstract alphabet:

\begin{theory-ext}
{Strings}
{\textsf{STR}.}
{\textsf{COF}.}
{$A$.}
{}
{}
\end{theory-ext}

Since \textsf{STR} is an extension of \textsf{COF}, we can assume that
\textsf{STR-1} is a development of \textsf{STR} that contains the 7
definitions of \textsf{COF-dev-2} named as \textsf{COF-Def1}, \ldots,
\textsf{COF-Def7} and the 22 theorems of \textsf{COF-dev-2} named as
\textsf{COF-Thm1}, \ldots, \textsf{COF-Thm22}.  We can extend
\textsf{STR-1} as follows to include the basic definitions and
theorems of strings:

\begin{dev-ext}
{Strings 2}
{\textsf{STR-2}.}
{\textsf{STR-1}.}
{
\bi

  \item[] $\mName{Def10}$:
  ${\cEqX
     {\mName{str}_{\cSetTy {\cFunTyX {R} {A}}}}
     {\cSeqFinQTy {A}}}$
  \hfill (string quasitype).

  \item[] $\mName{Def11}$:
  ${\cEqX
      {\epsilon_{\cFunTyX {R} {A}}}
      {\cEmpListPC {R} {A}}}$ \hfill (empty string).

  \item[] $\mName{Def12}$:
  ${\cEqX
      {\mName{cat}_{\cFunTyX {\cProdTy {\cFunTy {R} {A}} {\cFunTy {R} {A}}} {\cFunTy {R} {A}}}} 
      {\cAppendPC {R} {A}}}$\\
  \phantom{x} \hfill (concatenation).

\iffalse      
      \hspace*{2ex}
     ${\cDefDesQTyX 
         {f} 
         {\cFunTyX 
            {\cProdTy {\mName{str}} {\mName{str}}}
            {\mName{str}}} {}}$\\
      \hspace*{4ex}
     ${\cAndX 
         {\cForallX 
            {x} 
            {\mName{str}}
            {\cEqX 
               {\cFunAppX 
                  {f} 
                  {\cOrdPair {\epsilon} {x}}}
               {x}}} {}}$\\
        \hspace*{4ex}
       ${\cForallBX
           {a}
           {A}
           {x,y} 
           {\mName{str}}
           {\cEqX 
              {\cFunAppX {f} {\cOrdPair {\cConsX {a} {x}} {y}}}
              {\cConsX {a} {\cFunAppX {f} {\cOrdPair {x} {y}}}}}}$.
\fi

  \item[] $\mName{Thm33}$:
  ${\cForallX 
      {x} 
      {\mName{str}}
      {\cBinBX 
         {\cCatAppX {\epsilon} {x}}
         {=}
         {\cCatAppX {x} {\epsilon}}
         {=}
         {x}}}$
    \hfill ($\epsilon$ is an identity element). 

  \item[] $\mName{Thm34}$:
  ${\cForallX 
      {x,y,z} 
      {\mName{str}}
      {\cEqX
         {\cCatAppX {x} {\cCatApp {y} {z}}}
         {\cCatAppX {\cCatApp {x} {y}} {z}}}}$
     \hfill ($\mName{cat}$ is associative).

\ei
}
\end{dev-ext}

\bsp
\noindent
\textsf{Def10--Def12} utilize the compact notation introduced in
Table~\ref{tab:nd-sequences} and \textsf{Thm33--Thm34} utilize the
compact notation introduced in Table~\ref{tab:monoids-special}.
\esp

We can define a development translation from \textsf{MON-over-COF} to
\textsf{STR-2} as follows:

\begin{dev-trans-def}
{\small \textnormal{\textsf{MON-over-COF}} to \textnormal{\textsf{STR-2}}}
{\textsf{MON-over-COF-to-STR-2}.}
{\textsf{MON-over-COF}.}
{\textsf{STR-2}.}
{
\be 

  \item $R \mapsto R$.

  \item $M \mapsto {\mName{str}_{\cSetTy {\cFunTyX {R} {A}}}}$.

\ee
}
{
\be 

  \item $0_R \mapsto 0_R$.

  \item[] $\vdots$ \setcounter{enumi}{9}

  \item $\mName{lub}_{\cFunTyBX {R} {\cSetTy {R}} {\cB}} \mapsto
    \mName{lub}_{\cFunTyBX {R} {\cSetTy {R}} {\cB}}$.

  \item ${\cdot_{\cFunTyX {\cProdTy {M} {M}} {M}}} \mapsto
    {\mName{cat}_{\cFunTyX {\cProdTy {\cFunTy {R} {A}} {\cFunTy {R} {A}}} {\cFunTy {R} {A}}}}$.

  \item ${\mathsf{e}_M} \mapsto {\epsilon_{\cFunTyX {R} {A}}}$.

\ee
}
\end{dev-trans-def}

\textsf{MON-over-COF-to-STR-2} has one obligation of the first kind
for the mapped base type $M$, which is clearly valid since
$\mName{str}_{\cSetTy {\cFunTyX {R} {A}}}$ is nonempty.
\textsf{MON-over-COF-to-STR-2} has 12 obligations of the second kind
for the 12 mapped constants.  The first 10 are trivially valid.  The
last 2 are valid by $\mathsf{Def12}$ and $\mathsf{Def11}$,
respectively.  And \textsf{MON-over-COF-to-STR-2} has 20 obligations
of the third kind for the 20 axioms of \textsf{MON-over-COF}. The
first 18 are trivially valid.  The last 2 are valid by
$\mathsf{Thm34}$ and $\mathsf{Thm33}$, respectively.  Therefore,
\textsf{MON-over-COF-to-STR-2} is a development morphism from the
theory \textsf{MON-over-COF} to the development \textsf{STR-2} by the
Morphism Theorem \cite[Theorem 14.16]{Farmer25}.

The development morphism \textsf{MON-over-COF-to-STR-2} allows us to
transport definitions and theorems about monoids to the development
\textsf{STR-2}.  Here are five examples transported as a group:

\begin{group-transport}
{Transport to \textnormal{\textsf{STR-2}}}
{\textsf{monoid-machinery-via-MON-over-COF-1-to-STR-2}.}
{\textsf{MON-over-COF-1}.}
{\textsf{STR-2}.}
{\textsf{MON-over-COF-to-STR-2}.}
{
\bi

  \item[] $\mName{Thm1}$: ${\cMonoid {\cUnivSetPC {M}}
    {\cdot_{\cFunTyX {\cProdTy {M} {M}} {M}}} {\mathsf{e}_M}}$\\
  \phantom{x} \hfill (models of \textsf{MON} define monoids).

  \item[] $\mName{Def3}$: 
  ${\cEqX 
      {\odot_{\cFunTyX {\cProdTy {\cSetTy {M}} {\cSetTy {M}}} {\cSetTy {M}}}} 
      {\cFunAppX
         {\textsf{set-op}_{\cFunTyX {\cFunTy {\cProdTy {M} {M}} {M}} {\cFunTy {\cProdTy {\cSetTy {M}} {\cSetTy {M}}} {\cSetTy {M}}}}} 
         {\cdot}}}$\\ \phantom{x} \hfill (set product).

  \item[] $\mName{Def4}$: ${\cEqX {\mathsf{E}_{\cSetTy {M}}}
    {\cFinSetL {\mathsf{e}_M}}}$ \hfill (set identity element).

  \item[] \textsf{Thm12 (Thm1-via-MON-to-set-monoid)}:\\
  \hspace*{2ex}
  ${\cMonoid
     {\cUnivSetPC {\cSetTy {M}}}
     {\odot_{\cFunTyX {\cProdTy {\cSetTy {M}} {\cSetTy {M}}} {\cSetTy {M}}}} 
     {\mathsf{E}_{\cSetTy {M}}}}$\\
  \phantom{x} \hfill (set  monoids are monoids).

  \item[] $\mName{Def9}$: 
  ${\cEqX {\mName{prod}_{\cFunTyCX {R} {R} {\cFunTy {R} {M}} {M}}} {}}$\\
    \hspace*{2ex}
   ${\cDefDesX {f} 
    {\cFunTyCX {Z_{\cSetTy {R}}} {Z_{\cSetTy {R}}} {\cFunTy {Z_{\cSetTy {R}}} {M}} {M}} {}}$\\
    \hspace*{4ex}${\cForallBX {m,n} {Z_{\cSetTy {R}}} {g} {\cFunTyX {Z_{\cSetTy {R}}} {M}}
    {\cQuasiEqX {\cFunAppCX {f} {m} {n} {g}} {}}}$\\
    \hspace*{6ex}${\cIf {\cBinX {m} {>} {n}} {\mathsf{e}} 
    {\cBinX {\cFunAppC {f} {m} {\cBin {n} {-} {1}} {g}} {\cdot} 
    {\cFunApp {g} {n}}}}$ \hfill (iterated product).

\ei
}
{
\bi

\item[] \textsf{Thm35 (Thm1-via-MON-over-COF-to-STR-2)}:\\
  \hspace*{2ex}
  ${\cMonoid
      {\mName{str}_{\cSetTy {\cFunTyX {R} {A}}}}
      {\mName{cat}_{\cFunTyX {\cProdTy {\cFunTy {R} {A}} {\cFunTy {R} {A}}} {\cFunTy {R} {A}}}}
      {\epsilon_{\cFunTyX {R} {A}}}}$\\
  \phantom{x} \hfill (strings form a monoid).

  \item[] \textsf{Def13 (Def3-via-MON-over-COF-to-STR-2)}:\\
  \hspace*{2ex} 
  ${\cEqX 
      {\textsf{set-cat}_{\cFunTyX {\cProdTy {\cSetTy {\cFunTyX {R} {A}}} {\cSetTy {\cFunTyX {R} {A}}}} {\cSetTy {\cFunTyX {R} {A}}}}} {}}$\\
      \hspace*{2ex}
     ${\cFunAppX
         {\textsf{set-op}_{\cFunTyX {\cFunTy {\cProdTy {\cFunTy {R} {A}} {\cFunTy {R} {A}}} {\cFunTy {R} {A}}} {\cFunTy {\cProdTy {\cSetTy {\cFunTyX {R} {A}}} {\cSetTy {\cFunTyX {R} {A}}}} {\cSetTy {\cFunTyX {R} {A}}}}}} 
         {\mName{cat}}}$\\
  \phantom{x} \hfill (set concatenation). 

  \item[] \textsf{Def14 (Def4-via-MON-over-COF-to-STR-2)}:\\
    \hspace*{2ex}
    ${\cEqX {\mathsf{E}_{\cSetTy {{\cFunTyX {R} {A}}}}}
    {\cFinSetL {\epsilon_{\cFunTyX {R} {A}}}}}$ \hfill (set identity element). 

  \item[] \textsf{Thm36 (Thm12-via-MON-over-COF-1-to-STR-2)}:\\
  \hspace*{2ex}
  ${\cMonoid 
      {\cSetQTy {\mName{str}_{\cSetTy {\cFunTyX {R} {A}}}}} 
      {\textsf{set-cat}_{\cFunTyX {\cProdTy {\cSetTy {\cFunTyX {R} {A}}} {\cSetTy {\cFunTyX {R} {A}}}} {\cSetTy {\cFunTyX {R} {A}}}}} 
      {\mathsf{E}_{\cSetTy {{\cFunTyX {R} {A}}}}}}$\\
  \phantom{x} \hfill (string sets form a monoid).

  \item[] \textsf{Def15 (Def9-via-MON-over-COF-1-to-STR-2)}:\\
  \hspace*{2ex} 
  ${\cEqX {\textsf{iter-cat}_{\cFunTyCX {R} {R} {\cFunTy {R} {\cFunTy {R} {A}}} {\cFunTy {R} {A}}}} {}}$\\
    \hspace*{2ex}
   ${\cDefDesX {f} 
    {\cFunTyCX {Z_{\cSetTy {R}}} {Z_{\cSetTy {R}}} {\cFunTy {Z_{\cSetTy {R}}} {\cFunTy {R} {A}}} {\cFunTy {R} {A}}} {}}$\\
    \hspace*{4ex}${\cForallBX {m,n} {Z_{\cSetTy {R}}} {g} {\cFunTyX {Z_{\cSetTy {R}}} {\cFunTy {R} {A}}}
    {\cQuasiEqX {\cFunAppCX {f} {m} {n} {g}} {}}}$\\
    \hspace*{6ex}${\cIf {\cBinX {m} {>} {n}} {\epsilon} 
    {\cBinX {\cFunAppC {f} {m} {\cBin {n} {-} {1}} {g}} {\mathsf{cat}} 
    {\cFunApp {g} {n}}}}$\\
  \phantom{x} \hfill (iterated concatenation).

\ei
}
{\textsf{STR-3}.}
{\textsf{MON-over-COF-1-to-STR-3}.}
\end{group-transport} 

\noindent
Notation for the application of \[{\textsf{set-cat}_{\cFunTyX
    {\cProdTy {\cSetTy {\cFunTyX {R} {A}}} {\cSetTy {\cFunTyX {R}
          {A}}}} {\cSetTy {\cFunTyX {R} {A}}}}}\]
and \[{\textsf{iter-cat}_{\cFunTyCX {R} {R} {\cFunTy {R} {\cFunTy {R}
        {A}}} {\cFunTy {R} {A}}}}\] are defined in
Table~\ref{tab:monoids-special}.

\section{Related Work}\label{sec:related-work}

As we have seen, a theory (or development) graph provides an effective
architecture for formalizing a body of mathematical knowledge.  It is
especially useful for creating a large library of formal mathematical
knowledge that, by necessity, must be constructed in parallel by
multiple developers.  The library is built in parts by separate
development teams and then the parts are linked together by morphisms.
Mathematical knowledge is organized as a theory graph in several proof
assistants and logical frameworks including Ergo~\cite{NicksonEtAl96},
IMPS~\cite{FarmerEtAl92b,FarmerEtAl98b}\index{IMPS}, Isabelle~\cite{Ballarin14},
LF~\cite{RabeSchuermann09}, MMT~\cite{RabeKohlhase13}, and
PVS~\cite{OwreShankar01}.  Theory graphs are also employed in several
software specification and development systems including
ASL~\cite{SannellaWirsing83}, CASL~\cite{AstesianoEtAl02,
  AutexierEtAl99}, EHDM~\cite{RushbyEtAl91},
Hets~\cite{MossakowskiEtAl07}, IOTA~\cite{NakajimaYuasa83},
KIDS~\cite{Smith91}, OBJ~\cite{GoguenEtAl00}, and
Specware~\cite{SrinivasJullig95}.

Simple type theory in the form of Church's type theory is a popular
logic for formal mathematics.  There are several proof assistants that
implement versions of Church's type theory including
HOL~\cite{GordonMelham93}, HOL Light~\cite{Harrison09},
IMPS~\cite{FarmerEtAl93,FarmerEtAl98b}\index{IMPS},
Isabelle/HOL~\cite{Paulson94}, ProofPower~\cite{ProofPowerWebSite},
PVS~\cite{OwreEtAl96}, and TPS~\cite{AndrewsEtAl96}.  As we mentioned
in Section~\ref{sec:intro}, the IMPS proof assistant is especially
noteworthy here since it implements
LUTINS~\cite{Farmer90,Farmer93b,Farmer94}, a version of Church's type
theory that admits undefined expressions and is closely related to
Alonzo.

In recent years, there has been growing interest in formalizing
mathematics within dependent logics.  Several proof assistants and
programming languages are based on versions of dependent type theory
including Agda~\cite{BoveEtAl09, Norell07},
Automath~\cite{NederpeltEtAl94}, Epigram~\cite{EpigramWebSite},
{\fstar}~\cite{F*WebSite}, Idris~\cite{IdrisWebSite},
Lean~\cite{deMouraEtAl15}, Nuprl~\cite{Constable86}, and
Rocq~\cite{RocqWebSite}.  So which type theory is better for formal
mathematics, simple type theory or dependent type theory?  This
question has become hotly contested.  We hope that the reader will see
our formalization of monoid theory in Alonzo as evidence for the
efficacy of simple type theory as a logical basis for formal
mathematics.  The reader might also be interested in looking at these
recent papers that advocate for simple type theory:~\cite{BordgEtAl22,
  Paulson19, Paulson23}.

Since monoid theory is a relatively simple subject, there have not
been many attempts to formalize it by itself, but there have been
several formalizations of group theory.  Here are some
examples:~\cite{Garillot11, GonthierEtAl07, Kachapova21, Russinoff22,
  YuEtAl03, Zipperer16}.

There are two other important alternatives to the standard approach to
formal mathematics.  The first is Tom Hales' \emph{formal abstracts in
mathematics} project~\cite{FormAbsWebSite,Hales18} in which proof
assistants are used to the create \emph{formal abstracts}, which are
formal presentations of mathematical theorems without formal proofs.
The second is Michael Kohlhase's \emph{flexiformal
mathematics}~\cite{Iancu17,Kohlhase12,KohlhaseEtAl17} initiative in
which mathematics is a mixture of traditional and formal mathematics
and proofs can be either traditional or formal.  The alternative
approach we offer is similar to both of these approaches, but there
are important differences.  The formal abstracts approach seeks to
formalize \emph{collections of theorems} without proofs using proof
assistants, while we seek to formalize \emph{theory graphs} with
either traditional or formal proofs using supporting software that can
be much simpler than a proof assistant.  The objective of the
flexiformal mathematics approach is to give the user the flexibility
to produce mathematics with varying degrees of formality.  In
contrast, our approach is to produce mathematics that is fully formal
except for proofs.

\section{Conclusion}\label{sec:conc}

The developments and development morphisms presented in
Sections~\ref{sec:monoids}--\ref{sec:strings} form the development
graph $G_{\rm mon}$ shown in Figure~\ref{fig:dev-graph}.  The
development graph shows all the development morphisms that we have
explicitly defined (7 inclusions via theory extension modules and 10
noninclusions via theory and development definition modules) plus an
implicit inclusion from \textsf{COM-MON} to \textsf{COM-MON-over-COF}.
A development morphism that is an inclusion is designated by a
$\hookrightarrow$ arrow and a noninclusion is designated by a
$\rightarrow$ arrow.  There are many, many more useful development
morphisms that are not shown in $G_{\rm mon}$, including implicit
inclusions and a vast number of development morphisms into the theory
\textsf{COF}.

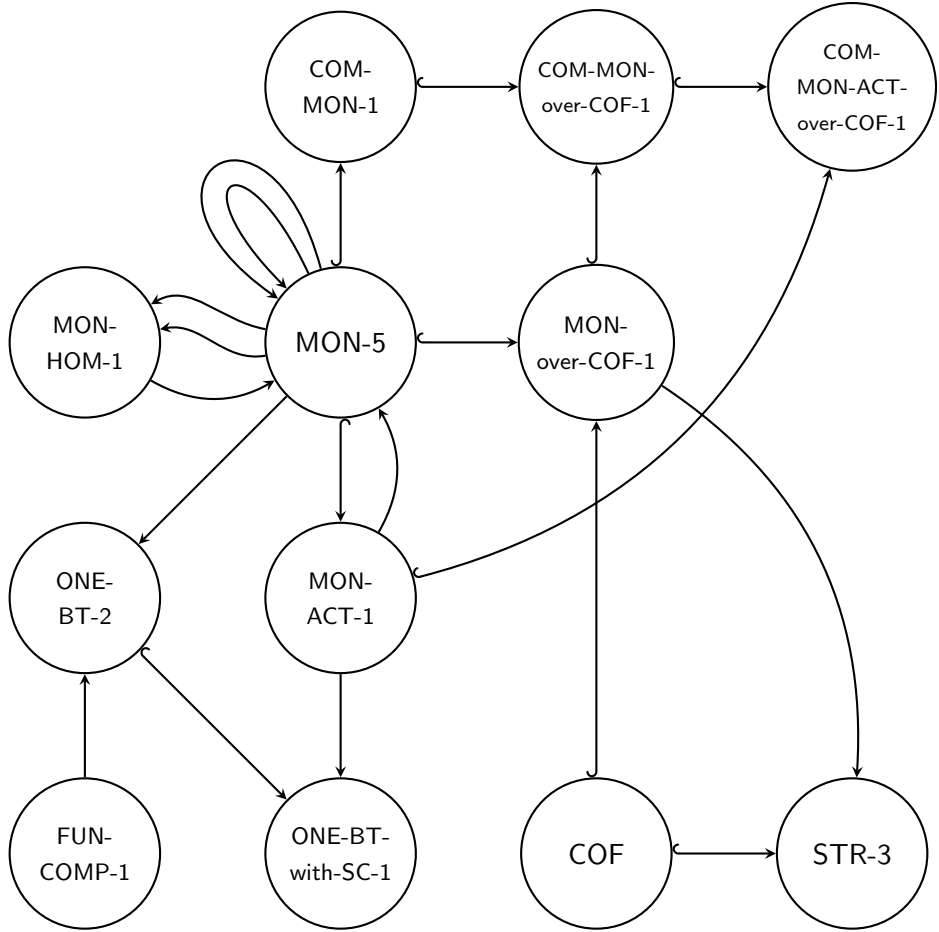
\begin{figure}[t]
\tikzset{every loop/.style={min distance=5mm,in=0,out=60,looseness=6}}
\tikzset{state/.style={circle, draw, minimum size=2cm}}
\tikzset{node distance=3.4cm,on grid,auto}
\center
\begin{tikzpicture}
  \node[state, thick] (mon) {\textsf{MON-5}};
  \node[state, thick, align=center] (com-mon) [above = of mon] 
    {\footnotesize{\textsf{COM-}} \\ \footnotesize{\textsf{MON-1}}};
  \node[state, thick, align=center] (com-hom) [left = of mon] 
    {\footnotesize{\textsf{MON-}} \\ \footnotesize{\textsf{HOM-1}}};
  \node[state, thick, align=center] (com-mon-cof) [right = of com-mon] 
    {\scriptsize{\textsf{COM-MON-}} \\ \scriptsize{\textsf{over-COF-1}}};
  \node[state, thick, align=center] (mon-cof) [right = of mon]  
    {\footnotesize{\textsf{MON-}} \\ \footnotesize{\textsf{over-COF-1}}};
  \node[state, thick, align=center] (mon-act) [below = of mon] 
    {\footnotesize{\textsf{MON-}} \\ \footnotesize{\textsf{ACT-1}}};
  \node[state, thick, align=center] (one) [left = of mon-act] 
    {\footnotesize{\textsf{ONE-}} \\ \footnotesize{\textsf{BT-2}}}; 
  \node[state, thick, align=center] (fun-comp) [below = of one]  
    {\footnotesize{\textsf{FUN-}} \\ \footnotesize{\textsf{COMP-1}}};
  \node[state, thick, align=center] (one-sc) [below = of mon-act]
    {\footnotesize{\textsf{ONE-BT-}} \\ \footnotesize{\textsf{with-SC-1}}};
  \node[state, thick] (cof) [right = of one-sc] {\textsf{COF}};
  \node[state, thick, align=center] (com-mon-act-cof) [right = of com-mon-cof] 
     {\scriptsize{\textsf{COM-}} \\ 
      \scriptsize{\textsf{MON-ACT-}} \\
      \scriptsize{\textsf{over-COF-1}}};
  \node[state, thick] (str) [right = of cof] {\textsf{STR-3}};
  \path [inclusion]
    (mon)             edge (com-mon)
                      edge (mon-cof)
                      edge (mon-act)
    (com-mon)         edge (com-mon-cof)
    (com-mon-cof)     edge (com-mon-act-cof)
    (mon-act)         edge [bend right] (com-mon-act-cof)
    (mon-cof)         edge (com-mon-cof)
    (one)             edge (one-sc)
    (cof)             edge (str)
                      edge (mon-cof);
  \path [morphism]
    (mon)             edge [out=115, in=135, looseness=15] (mon)
                      edge [out=105, in=145, looseness=10] (mon)
                      edge [out=170, in=30] (com-hom)
                      edge [out=190, in=10] (com-hom)
                      edge (one)
    (com-hom)         edge [bend right] (mon)
    (mon-act)         edge (one-sc)
                      edge [bend right] (mon)
    (fun-comp)        edge (one)
    (mon-cof)         edge [bend left] (str);
\end{tikzpicture}
\caption{The Monoid Theory Development Graph}
\label{fig:dev-graph}
\end{figure}

The construction of $G_{\rm mon}$ illustrates how a body of
mathematical knowledge can be formalized in Alonzo as a development
graph in accordance with the little theories method and the
alternative approach.  $G_{\rm mon}$ could be extended to include
other mathematical concepts related to monoids such as categories.  It
could be incorporated in a development graph that formalizes a more
extensive body of mathematical knowledge.  And it could also be used
as a foundation for building a formalization of group theory.  This
would be done by lifting each development $D$ of a theory $T$ that
extends \textsf{MON} to a development $D'$ of a theory $T'$ that
extends a theory \textsf{GRP} of groups obtained by adding an inverse
operation to \textsf{MON}.  The lifting of $D$ to $D'$ would include
constructing inclusions from \textsf{MON} to \textsf{GRP} and from $T$
to~$T'$ via theory extensions.

The formalization of monoid theory we have presented demonstrates
three things.  First, it demonstrates the power of the little theories
method.  The formalization is largely free of redundancy since each
mathematical topic is articulated in just one development $D$, the
development for the little theory that is optimal for the topic in
level of abstraction and choice of vocabulary.  If we create a
translation $\Phi$ from $D$ to another development $D'$ and prove that
$\Phi$ is a morphism, then we can freely transport the definitions and
theorems of $D$ to $D'$ via $\Phi$.  That is, an abstract concept or
fact that has been validated in $D$ can be translated to a concrete
instance of the concept or fact that is automatically validated in
$D'$ provided the translation is a morphism.  (This is illustrated by
our use of the development morphism \textsf{MON-over-COF-to-STR-2} to
transport definitions and theorems about monoids to a development
about strings.)  As the result, the same concept or fact can appear in
many places in the theory graph but under different assumptions and
involving different vocabulary.  (For example, the notion of a
submonoid represented by the constant ${\mathsf{submonoid}_{\cFunTyX
    {\cSetTy {M}} {\cB}}}$ defined in \textsf{MON-1} appears in
\textsf{ONE-BT-2} as the notion of a transformation monoid represented
by the constant ${\textsf{trans-monoid}_{\cFunTyX {\cSetTy {\cFunTyX
        {S} {S}}} {\cB}}}$.)  In short, we have shown how the little
theories method enables mathematical knowledge to be formalized to
maximize clarity and minimize redundancy.

Second, the formalization demonstrates that the alternative approach
to formal mathematics (with traditional and formal proofs) has two
advantages over the standard approach (with only formal proofs):
(1)~communication is more effective since the user has greater freedom
of expression and (2)~formalization is easier since the approach
offers greater accessibility.  The standard approach is done with the
help of a proof assistant and all proofs are formal and mechanically
checked.  Proof assistants are consequently very complex and
notoriously difficult to learn how to use.  Traditional proofs are
easier to read and write than formal proofs and are better suited for
communicating the ideas behind proofs.  Moreover, since the
alternative approach does not require a facility for developing and
checking formal proofs, it can be done with software support that is
much simpler and easier to use than a proof assistant.  (In this
paper, our software support was just a set of LaTeX macros and
environments.)

Third, the formalization demonstrates that Alonzo is well suited for
expressing and reasoning about mathematical ideas.  The simple type
theory machinery of Alonzo --- function and product types, function
application and abstraction, definite description, and ordered pairs
--- enables mathematical expressions to be formulated in a direct and
natural manner.  It also enables almost every single mathematical
structure or set of similar mathematical structures to be specified by
an Alonzo development.  (For example, the development
\textsf{ONE-BT-2} specifies the set of mathematical structures
consisting of a set $S$ and the set $S \tarrow S$ of transformations
on $S$.)  The admission of undefined expressions in Alonzo enables
statements involving partial and total functions and definite
descriptions to be expressed directly, naturally, and succinctly.
(For example, if $M = (m,\cdot,e)$ is a monoid, the operation that
makes a submonoid $m' \subseteq m$ of $M$ a monoid itself is exactly
what is expected: the partial function that results from restricting
$\cdot$ to $m' \times m'$.)  And the notational definitions and
conventions employed in Alonzo enables mathematical expressions to be
presented with largely the same notation that is used mathematical
practice.  (For example, \textsf{Thm33}: ${\cForallX {x} {\mName{str}}
  {\cBinBX {\cCatAppX {\epsilon} {x}} {=} {\cCatAppX {x} {\epsilon}}
    {=} {x}}}$, that states $\epsilon$ is an identity element for
concatenation, is written just as one would expect it to be written in
mathematical practice.)

\newpage

We believe that this paper achieves our overarching goal: To
demonstrate that mathematical knowledge can be very effectively
formalized in a version of simple type theory like Alonzo using the
little theories method and the alternative approach to formal
mathematics.  We also believe that it illustrates the benefits of
employing the little theories method, the alternative approach, and
Alonzo in formal mathematics.

\appendix

\section{Validation of Definitions and Theorems}\label{app:validation}

Let $D = (T, \Xi)$ be a development where $T$ is the bottom theory of
the development and $\Xi = \mList{P_1,\ldots,P_n}$ is the list of
definition and theorem packages of the development.  For each $i$ with
$1 \le i \le n$, $P_i$ has the form $\mTuple{p,{\cCon {\mathbf{c}}
    {\alpha}},\mathbf{A}_\alpha,\pi}$ if $P_i$ is a definition package
and has the form $P_i =\mTuple{p,\mathbf{A}_\cB,\pi}$ if $P_i$ is a
theorem package.  Define $T_0 = T$ and, for all $i$ with $0 \le i \le
n - 1$, define $T_{i+1} = T[P_{i+1}]$ if $P_{i+1}$ is a definition
package and $T_{i+1} = T_i$ if $P_{i+1}$ is a theorem package.  In the
former case, $\pi$ is a proof that $\cIsDefX {\mathbf{A}_\alpha}$ is
valid in $T_i$, and in the latter case, $\pi$ is a proof that
$\mathbf{A}_\cB$ is valid in $T_i$.  These proofs may be either
traditional or formal.  See Chapter 12 of~\cite{Farmer25} for further
details.

The validation proofs for the definitions and theorems of a
development are not included in the modules we have used to construct
developments and to transport definitions and theorems.  Instead, we
give in this appendix, for each of the definitions and theorems in the
developments defined in Sections~\ref{sec:monoids}--\ref{sec:strings},
a traditional proof that validates the definition or theorem.  The
proofs are almost entirely straightforward.  The proofs extensively
reference the axioms, rules of inference, and metatheorems of
$\mathfrak{A}$, the formal proof system for Alonzo presented
in~\cite{Farmer25}.  These are legitimate to use since $\mathfrak{A}$
is~sound by the Soundness Theorem \cite[Theorem B.11]{Farmer25}.

\subsection{Development of \textnormal{\textsf{MON}}} 

\be

  \item $\mName{Thm1}$: ${\cMonoid {\cUnivSetPC {M}}
    {\cdot_{\cFunTyX {\cProdTy {M} {M}} {M}}} {\mathsf{e}_M}}$\\
  \phantom{x} \hfill (models of \textsf{MON} define monoids).

\begin{proof}[ of the theorem.] 
Let $T = (L, \Gamma)$ be $\textsf{MON}$.  We must show \[(\star)
{\sglsp} T \vDash {\cMonoid {\cUnivSetPC M} {\cdot_{\cFunTyX {\cProdTy
        {M} {M}} {M}}} {\mathsf{e}_M}}.\]
\begin{align*}
\Gamma &\vDash
    {\cIsDefX {\cUnivSetPC M}}\tag{1}\\
\Gamma &\vDash
    \cUnivSetPC{M} \not= \cEmpSetPC{M}\tag{2}\\
\Gamma &\vDash
    {\cIsDefInQTyX
       {\cdot_{\cFunTyX{\cProdTy {M} {M}} {M}}}
       {\cFunQTyX
          {\cProdQTy {\cUnivSetPC {M}} {\cUnivSetPC {M}}}
          {\cUnivSetPC {M}}}}\tag{3}\\
\Gamma &\vDash
    {\cIsDefInQTyX
       {\mathsf{e}_M}
       {\cUnivSetPC {M}}}\tag{4}\\
\Gamma &\vDash
    \cForallX{x, y, z} {\cUnivSetPC M}
    {\cEqX 
       {\cBinX {x} {\cdot} {\cBin {y} {\cdot} {z}}}
       {\cBinX {\cBin {x} {\cdot} {y}} {\cdot} {z}}}\tag{5}\\
\Gamma &\vDash
    {\cForallX
       {x} 
       {\cUnivSetPC M} 
       {\cEqX
          {\cBinX {\textsf{e}} {\cdot} {x}}
          {x}}}\tag{6}\\
\Gamma &\vDash
    {\cMonoid 
       {\cUnivSetPC M} 
       {\cdot_{\cFunTyX {\cProdTy {M} {M}} {M}}} 
       {\mathsf{e}_M}}\tag{7}
\end{align*}
(1) and (2)~follow from parts 1 and 2, respectively, of Lemma~\ref{lem:univ-sets}; 
(3)~follows from~\cite[Axiom~A5.2]{Farmer25} and parts 8--10 of Lemma~\ref{lem:univ-sets};
(4)~follows from~\cite[Axiom~A5.2]{Farmer25} and part 8 of Lemma~\ref{lem:univ-sets}; 
(5) and (6) follow from Axioms 1 and~2, respectively, of $T$ and part 5 of
Lemma~\ref{lem:univ-sets}; 
and (7)~follows from (1)--(6) and the definition of $\textsf{MONOID}$ in 
Table~\ref{tab:monoids-abbr}.  Therefore, $(\star)$ holds.
\end{proof}

  \item $\mName{Thm2}$: ${\cTotal {\cdot_{\cFunTyX {\cProdTy {M} {M}}
        {M}}}}$ \hfill ($\cdot$ is total).

\begin{proof}[ of the theorem.]
Let $\textbf{A}_\cB$ be 
\[\cForallX {x} {\cProdTyX {M} {M}}
{\cIsDefX {\cFunApp {\cdot_{\cFunTyX {\cProdTy {M} {M}} {M}}} {x}}}\]
and $T = (L,\Gamma)$ be \textsf{MON}.  $\mathsf{TOTAL}$ is the
abbreviation introduced by the notational definition given in
Table~\ref{tab:nd-functions}, and so ${\cTotal {\cdot_{\cFunTyX
      {\cProdTy {M} {M}} {M}}}}$ stands for $\textbf{A}_\cB$.  Thus we
must show $(\star)$~$T \vDash \textbf{A}_\cB$.  
\begin{align*}
\Gamma &\vDash 
    {\cIsDefX {\cVar {x} {\cProdTyX {M} {M}}}}\tag{1}\\
\Gamma &\vDash
    {\cEqX 
       {\cVar {x} {\cProdTyX {M} {M}}}
       {\cOrdPair 
          {\cFunAppX {\mathsf{fst}} {x}}
          {\cFunAppX {\mathsf{snd}} {x}}}}\tag{2}\\
\Gamma &\vDash
    {\cAndX 
       {\cIsDefX {\cFunApp {\mathsf{fst}} {x}}}
       {\cIsDefX {\cFunApp {\mathsf{snd}} {x}}}}\tag{3}\\
\Gamma &\vDash
    {\cEqX 
       {\cBinX 
          {\cFunApp {\mathsf{fst}} {x}} 
          {\cdot} 
          {\cBin 
             {\cFunApp {\mathsf{fst}} {x}} 
             {\cdot} 
             {\cFunApp {\mathsf{snd}} {x}}}}
       {\cBinX 
          {\cBin 
             {\cFunApp {\mathsf{fst}} {x}}
             {\cdot} 
             {\cFunApp {\mathsf{fst}} {x}}}
          {\cdot} 
          {\cFunApp {\mathsf{snd}} {x}}}}\tag{4}\\
\Gamma &\vDash
    {\cIsDefX
       {\cBin 
          {\cFunApp {\mathsf{fst}} {x}} 
          {\cdot} 
          {\cFunApp {\mathsf{snd}} {x}}}}\tag{5}\\
\Gamma &\vDash
    {\cIsDefX
       {\cFunApp 
          {\cdot_{\cFunTyX {\cProdTy {M} {M}} {M}}} 
          {\cOrdPair 
             {\cFunAppX {\mathsf{fst}} {x}} 
             {\cFunAppX {\mathsf{snd}} {x}}}}}\tag{6}\\
\Gamma &\vDash
    \textbf{A}_\cB\tag{7}
\end{align*}
(1)~follows from variables always being defined by~\cite[Axiom A5.1]{Farmer25};
(2)~follows from (1) and~\cite[Axiom A7.4]{Farmer25} by
Universal Instantiation~\cite[Theorem A.14]{Farmer25}; 
(3)~follows from (2) by~\cite[Axioms A5.5, A7.2, and A7.3]{Farmer25};
(4)~follows from (3) and Axiom 1 of $T$ by 
Universal Instantiation~\cite[Theorem A.14]{Farmer25}; 
(5)~follows from (4) by~\cite[Axioms A5.4 and A5.10]{Farmer25};
(6)~follows from (5) by notational definition; and
(7)~follows from (6) by Universal Generalization~\cite[Theorem A.30]{Farmer25} 
using (2) and the fact that ${\cVar{x} {\cProdTy {M} {M}}}$ is not free in 
$\Gamma$ since $\Gamma$ is a set of sentences.
\end{proof}

  \item $\mName{Thm3}$: ${\cForallX {x} {M} {\cImpliesX {\cForall
        {y} {M} {\cBinBX {\cBinX {x} {\cdot} {y}} {=} {\cBinX {y}
            {\cdot} {x}} {=} {y}}} {\cEqX {x}
        {\mathsf{e}}}}}$\\ \phantom{x} \hfill (uniqueness of identity
    element).

\begin{proof}[ of the theorem.]
Let $\textbf{A}_\cB$ be \[{\cForallX {y} {M} {\cBinBX {\cBinX {\cVar
        {x} {M}} {\cdot} {y}} {=} {\cBinX {y} {\cdot} {\cVar {x} {M}}}
    {=} {y}}}\] and $T = (L,\Gamma)$ be \textsf{MON}.  We must show
$(\star)$ $T \vDash {\cForallX {x} {M} {\cImpliesX {\textbf{A}_\cB}
    {\cEqX {x} {\mathsf{e}}}}}$.
\begin{align*}
\Gamma \cup \mSet{\textbf{A}_\cB} 
  &\vDash 
  {\cIsDefX {\mathsf{e}}}\tag{1}\\
\Gamma \cup \mSet{\textbf{A}_\cB} 
  &\vDash 
  {\cIsDefX {\cVar {x} {M}}}\tag{2}\\
\Gamma \cup \mSet{\textbf{A}_\cB} 
  &\vDash 
  {\cBinBX 
     {\cBinX {\cVar {x} {M}} {\cdot} {\mathsf{e}}} 
     {=} 
     {\cBinX {\mathsf{e}} {\cdot} {{\cVar {x} {M}}}} 
     {=} 
     {\mathsf{e}}}\tag{3}\\
\Gamma \cup \mSet{\textbf{A}_\cB} 
  &\vDash  
  {\cBinBX 
     {\cBinX {\mathsf{e}} {\cdot} {\cVar {x} {M}}} 
     {=} 
     {\cBinX {\cVar {x} {M}} {\cdot} {\mathsf{e}}} 
     {=} 
     {\cVar {x} {M}}}\tag{4}\\ 
\Gamma \cup \mSet{\textbf{A}_\cB} 
  &\vDash  
  {\cEqX {\cVar {x} {M}} {\mathsf{e}}}\tag{5}\\
\Gamma
  &\vDash
  {\cImpliesX {\textbf{A}_\cB} {\cEqX {\cVar {x} {M}} {\mathsf{e}}}}\tag{6}\\
\Gamma
  &\vDash
  {\cForallX 
     {x}
     {M} 
     {\cImpliesX {\textbf{A}_\cB} {\cEqX {x} {\mathsf{e}}}}}\tag{7}
\end{align*}
(1) follows from constants always being defined by~\cite[Axiom
  A5.2]{Farmer25}; 
(2)~follows from variables always being defined by~\cite[Axiom A5.1]{Farmer25}; 
(3)~follows (1) and $\textbf{A}_\cB$ by
Universal Instantiation~\cite[Theorem A.14]{Farmer25}; 
(4) follows (2) and Axiom 2 of $T$ by Universal Instantiation; 
(5) follows from (3) and (4) by the Equality Rules~\cite[Lemma A.13]{Farmer25}; 
(6)~follows from (5) by the Deduction Theorem~\cite[Lemma A.50]{Farmer25}; 
and (7) follows from~(6) by Universal Generalization~\cite[Theorem A.30]{Farmer25}
using the fact that $\cVar {x} {M}$ is not free in~$\Gamma$ since
$\Gamma$ is a set of sentences. Therefore, $(\star)$ holds.
\end{proof}

  \item $\mName{Def1}$: ${\cEqX {\mathsf{submonoid}_{\cFunTyX
        {\cSetTy {M}} {\cB}}} {}}$\\ \hspace*{2ex} ${\cFunAbsX {s}
    {\cSetTy {M}} {\cAndX {\cAndX {\cNotEqX {s} {\cEmpSetPC {M}}}
        {\cIsDefInQTy {\cRestrictX {\cdot} {\cProdTyX {s} {s}}}
          {\cFunTyX {\cProdTy {s} {s}} {s}}}} {\cInX {\mathsf{e}}
        {s}}}}$ \hfill (submonoid).

\begin{proof}[ that RHS is defined.] 
Let $\textbf{A}_{\cFunTyX {\cSetTy {M}} {\cB}}$ be the RHS of
\textsf{Def1}.~We must show that $\textsf{MON} \vDash \cIsDefX
       {\textbf{A}_{\cFunTyX {\cSetTy {M}} {\cB}}}$.  This follows
       immediately from function abstractions always being defined
       by~\cite[Axiom A5.11]{Farmer25}.
\end{proof}

  \item $\mName{Thm4}$: ${\cForallX {s} {\cSetTy {M}} {\cImpliesX
      {\cFunAppX {\mathsf{submonoid}} {s}} {\cMonoid {s}
        {\cdot\wrestricted_{\cProdTyX {s} {s}}}
        {\mathsf{e}}}}}$\\ \phantom{x} \hfill (submonoids are
    monoids).

\begin{proof}[ of the theorem.]
Let $\textbf{A}_\cB$ be
\[\mathsf{submonoid}(s)\]
and $T = (L,\Gamma)$ be \textsf{MON} extended by $\mName{Def1}$.  We
must show
\[(\star) {\sglsp} T \vDash \cForallX {s} {\cSetTy{M}} {\cImpliesX{\textbf{A}_\cB}
    {\cMonoid {s} {\cRestrictX{\cdot} {\cProdTy {s} {s}}} {\mathsf{e}}}}.\]
\begin{align*}
    \Gamma \cup \mSet{\textbf{A}_\cB}
        &\vDash
        \cIsDefX{s_{\cSetTy{M}}}\tag{1}\\
    \Gamma \cup \mSet{\textbf{A}_\cB}
        &\vDash
        s \not= \cEmpSetPC{M}\tag{2}\\
    \Gamma \cup \mSet{\textbf{A}_\cB}
        &\vDash
        \cIsDefInQTyX{\cRestrictX{\cdot} {\cProdTy {s} {s}}} {\cFunQTyX{\cProdQTy{s} {s}} {s}}\tag{3}\\
    \Gamma \cup \mSet{\textbf{A}_\cB}
        &\vDash
        \mathsf{e} \in s\tag{4}\\
    \Gamma \cup \mSet{\textbf{A}_\cB}
        &\vDash
        \cIsDefInQTyX{\mathsf{e}} {s}\tag{5}\\
    \Gamma \cup \mSet{\textbf{A}_\cB}
        &\vDash
        \cForallX {x,y,z} {s} {\cFunAppX{\cRestrictX{\cdot} {\cProdTy {s} {s}}} {\cOrdPair{x} {\cFunAppX{\cRestrictX{\cdot} {\cProdTy {s} {s}}} {\cOrdPair{y} {z}}}}}\\
            &\hspace*{2ex}=
            \cFunAppX{\cRestrictX{\cdot} {\cProdTy {s} {s}}} {\cOrdPair{\cFunAppX{\cRestrictX{\cdot} {\cProdTy {s} {s}}} {\cOrdPair{x} {y}}} {z}}\tag{6}\\
    \Gamma \cup \mSet{\textbf{A}_\cB}
        &\vDash
        \cForallX{x} {s} {\cFunAppX{\cRestrictX{\cdot} {\cProdTy {s} {s}}} {\cOrdPair{\mathsf{e}} {x}}
        = \cFunAppX{\cRestrictX{\cdot} {\cProdTy {s} {s}}} {\cOrdPair{x} {\mathsf{e}}}
        = x}\tag{7}\\
    \Gamma \cup \mSet{\textbf{A}_\cB}
        &\vDash
        \cMonoid{s} {\cRestrictX{\cdot} {\cProdTy {s} {s}}} {\mathsf{e}}\tag{8}\\
    \Gamma
        &\vDash
        \cImpliesX{\textbf{A}_\cB} {\cMonoid{s} {\cRestrictX{\cdot} {\cProdTy {s} {s}}} {\mathsf{e}}}\tag{9}\\
    \Gamma
        &\vDash
        \cForallX{s} {\cSetTy{M}} {\cImpliesX{\textbf{A}_\cB} {\cMonoid{s} {\cRestrictX{\cdot} {\cProdTy {s} {s}}} {\mathsf{e}}}}\tag{10}
\end{align*}
(1)~follows from variables always being defined by~\cite[Axiom A5.1]{Farmer25}; 
(2), (3), and (4) follow directly from $\mathsf{Def1}$; 
(5)~follows from \cite[Axiom A5.2]{Farmer25} and (4); 
(6)~and (7) follow from $\textsf{Thm1}$,  
$\cRestrictX{\cdot} {\cProdTy {s} {s}} \sqsubseteq 
\cdot_{\cFunTyX{\cProdTy {M} {M}} {M}}$, 
and the fact that $\cRestrictX{\cdot} {\cProdTy {s} {s}}$ is total on 
$\cProdTyX {s} {s}$ by $\mathsf{Thm2}$; 
(8)~follows from (1)--(3) and (5)--(7) by the definition of $\textsf{MONOID}$ in 
Table~\ref{tab:monoids-abbr}; 
(9)~follows from (8) by the Deduction Theorem \cite[Theorem A.50]{Farmer25}; 
and (10)~follows from (9) by Universal Generalization \cite[Theorem A.30]{Farmer25} 
using the fact that $\cVar{s} {\cSetTy{M}}$ is not free in $\Gamma$ 
since $\Gamma$ is a set of sentences. Therefore, $(\star)$ holds.
\end{proof}

  \item $\mName{Thm5}$: ${\cFunAppX {\mathsf{submonoid}} {\cFinSetL
      {\mathsf{e}}}}$ \hfill (minimum submonoid).

\begin{proof}[ of the theorem.]
Let $T = (L, \Gamma)$ be $\mathsf{MON}$ extended by $\mName{Def1}$. We
must show $(\star)$ $T \vDash {\cFunAppX {\mathsf{submonoid}}
  {\cFinSetL{\mathsf{e}}}}$.
\begin{align*}
    \Gamma &\vDash
      \cInX{\mathsf{e}}{\cFinSetL{\mathsf{e}}} \tag{1}\\
    \Gamma &\vDash
        \cFinSetL{\mathsf{e}} \not= \cEmpSetPC{M}\tag{2}\\
    \Gamma &\vDash
        \mathsf{e} \cdot \mathsf{e} = \mathsf{e}\tag{3}\\
    \Gamma &\vDash
        \cIsDefInQTyX{\cRestrictX{\cdot} {\cProdTyX{\cSetTy{\mathsf{e}}} {\cSetTy{\mathsf{e}}}}} {\cFunQTyX{\cProdQTy{\cFinSetL{\mathsf{e}}}{\cFinSetL{\mathsf{e}}}}{\cFinSetL{\mathsf{e}}}}\tag{4}\\
    \Gamma &\vDash
        {\cFunAppX {\mathsf{submonoid}} {\cFinSetL{\mathsf{e}}}}\tag{5}
\end{align*}
(1)~is trivial; 
(2)~follows from (1) because $\cFinSetL{\mathsf{e}}$ has at least one member; 
(3) follows from Axiom 2 of $T$ by Universal Instantiation 
\cite[Theorem A.14]{Farmer25}; 
(4) follows directly from (1), (3), and the fact that the only member of 
$\cFinSetL{\mathsf{e}}$ is $\mathsf{e}$; and
(5) follows from (1), (2), (4), and $\mathsf{Def1}$. 
Therefore, $(\star)$ holds.
\end{proof}

  \item $\mName{Thm6}$: ${\cFunAppX {\mathsf{submonoid}}
    {\cUnivSetPC {M}}}$ \hfill (maximum submonoid).

\begin{proof}[ of the theorem.]
Let $T = (L, \Gamma)$ be $\mathsf{MON}$ extended by $\mName{Def1}$.
We must show $(\star)$ $T \vDash {\cFunAppX {\mathsf{submonoid}}
  {\cUnivSetPC {M}}}$.
\begin{align*}
    \Gamma &\vDash
        \cMonoid{\cUnivSetPC{M}} {\cdot_{\cFunTyX{\cProdTy {M} {M}} {M}}} {\mathsf{e}}\tag{1}\\
    \Gamma &\vDash
        \cUnivSetPC{M} \not= \cEmpSetPC{M} \land \mathsf{e} \in \cUnivSetPC{M}\tag{2}\\
    \Gamma &\vDash
      \cIsDefInQTyX{\cRestrictX{\cdot} {\cProdTyX{\cUnivSetPC M} {\cUnivSetPC M}}} {\cFunQTyX{\cProdQTy{\cUnivSetPC M}{\cUnivSetPC M}}{\cUnivSetPC M}}\tag{3}\\
    \Gamma &\vDash 
       {\cFunAppX {\mathsf{submonoid}} {\cUnivSetPC {M}}}\tag{4}\\
\end{align*}
(1) is $\mathsf{Thm1}$; (2) follows immediately from (1); (3) follows
from (1) by part 12 of Lemma~\ref{lem:univ-sets}; and (4) follows from
(1), (2), (3), and $\mathsf{Def1}$.  Therefore, $(\star)$ holds.
\end{proof}

  \item $\mName{Def2}$: ${\cEqX {\cdot^{\rm op}_{\cFunTyX {\cProdTy
          {M} {M}} {M}}} {\cFunAbsX {p} {\cProdTyX {M} {M}} {\cBinX
        {\cFunApp {\mathsf{snd}} {p}} {\cdot} {\cFunApp {\mathsf{fst}}
          {p}}}}}$ \hfill (opposite of $\cdot$).

\begin{proof}[ that RHS is defined.] 
Similar to the proof that the RHS of $\mathsf{Def1}$ is defined.
\end{proof}

  \item $\mName{Thm7}$: ${\cForallX {x,y,z} {M} {\cEqX {\cBinX {x}
        {\cdot^{\rm op}} {\cBin {y} {\cdot^{\rm op}} {z}}} {\cBinX
        {\cBin {x} {\cdot^{\rm op}} {y}} {\cdot^{\rm op}} {z}}}}$\\
    \phantom{x} \hfill ($\cdot^{\rm op}$ is associative).

\begin{proof}[ of the theorem.]
Let $\textbf{A}_\cB$ be
\[\cEqX {\cBinX {x}
    {\cdot^{\rm op}} {\cBin {y} {\cdot^{\rm op}} {z}}} {\cBinX
    {\cBin {x} {\cdot^{\rm op}} {y}} {\cdot^{\rm op}} {z}}\]
and $T = (L, \Gamma)$ be $\mathsf{MON}$ extended by \textsf{Def2}.  We
must show \[(\star) {\sglsp} T \vDash \cForallX{x, y, z} {M}
{\textbf{A}_\cB}.\]
\begin{align*}
    \Gamma &\vDash
        \cIsDefX{\cVar{x} {M}} \land
        \cIsDefX{\cVar{y} {M}} \land
        \cIsDefX{\cVar{z} {M}}\tag{1}\\
    \Gamma &\vDash
        {\cEqX 
           {\cBinX {\cBin {z} {\cdot} {y}} {\cdot} {x}}
           {\cBinX {z} {\cdot} {\cBin {y} {\cdot} {x}}}}\tag{2}\\
    \Gamma &\vDash
        \textbf{A}_\cB\tag{3}\\
    \Gamma &\vDash
        {\cForallX {x,y,z} {M} {\textbf{A}_\cB}}\tag{4}
\end{align*}
(1)~follows from variables always being defined by \cite[Axiom A5.1]{Farmer25}; 
(2)~follows from (1) and Axiom 1 of $T$ by Universal Instantiation
\cite[Theorem A.14]{Farmer25} and the Equality Rules \cite[Theorem A.13]{Farmer25};
(3)~follows from Lemma~\ref{lem:cdot-opposite} and (2)
by repeated applications of Rule $\mbox{R2}'$ \cite[Lemma A.2]{Farmer25}
using $({\star}{\star})$ the fact that $\cVar{x}{M}$, $\cVar{y}{M}$, and 
$\cVar{z}{M}$ are not free in $\Gamma$ since $\Gamma$ is a set of sentences; and
(4)~follows from (3) by Universal Generalization \cite[Theorem A.30]{Farmer25} 
again using $({\star}{\star})$.  
Therefore $(\star)$ holds.
\end{proof}

  \item $\mName{Thm8}$: ${\cForallX {x} {M} {\cBinBX {\cBinX
        {\mathsf{e}} {\cdot^{\rm op}} {x}} {=} {\cBinX {x} {\cdot^{\rm
            op}} {\mathsf{e}}} {=} {x}}}$\\ 
   \phantom{x} \hfill ($\mathsf{e}$ is an identity element with
   respect to $\cdot^{\rm op}$).

\begin{proof}[ of the theorem.]
Let $\textbf{A}_\cB$ be
\[{\cBinBX 
     {\cBinX {\mathsf{e}} {\cdot^{\rm op}} {x}}
     {=}
     {\cBinX {x} {\cdot^{\rm op}} {\mathsf{e}}}
     {=} 
     {x}}\]
and $T = (L, \Gamma)$ be $\mathsf{MON}$ extended by \textsf{Def2}.  We
must show \[(\star) {\sglsp} T \vDash \cForallX{x} {M}
{\textbf{A}_\cB}.\]
\begin{align*}
    \Gamma &\vDash
        {\cIsDefX {\cVar{x} {M}}}\tag{1}\\
    \Gamma &\vDash
        {\cBinBX 
           {\cBinX {x} {\cdot} {\mathsf{e}}}
           {=}
           {\cBinX {\mathsf{e}} {\cdot} {x}}
           {=} 
           {x}}\tag{2}\\
    \Gamma &\vDash
        \textbf{A}_\cB\tag{3}\\
    \Gamma &\vDash
        {\cForallX {x} {M} {\textbf{A}_\cB}}\tag{4}
\end{align*}
(1)~follows from variables always being defined by \cite[Axiom A5.1]{Farmer25}; 
(2)~follows from (1) and Axiom 2 of $T$ by Universal Instantiation
\cite[Theorem A.14]{Farmer25} and the Equality Rules \cite[Theorem A.13]{Farmer25};
(3)~follows from Lemma~\ref{lem:cdot-opposite} and (2)
by repeated applications of Rule $\mbox{R2}'$ \cite[Lemma A.2]{Farmer25}
using $({\star}{\star})$ the fact that $\cVar{x}{M}$ is not free in $\Gamma$ 
since $\Gamma$ is a set of sentences; and
(4)~follows from (3) by Universal Generalization \cite[Theorem A.30]{Farmer25} 
again using $({\star}{\star})$.  
Therefore $(\star)$ holds.
\end{proof}

  \item $\mName{Def3}$: ${\cEqX {\odot_{\cFunTyX {\cProdTy {\cSetTy
            {M}} {\cSetTy {M}}} {\cSetTy {M}}}} {\cFunAppX
      {\cSetProdPC {\textsf{set-op}} {M} {M} {M}} {\cdot}}}$\\ 
  \phantom{x} \hfill (set product).

\begin{proof}[ that RHS is defined.]
Let $\textbf{A}_{\cFunTyX {\cProdTy {\cSetTy {M}} {\cSetTy {M}}}
  {\cSetTy {M}}}$ be the RHS of \textsf{Def3}.~We must show $(\star)$
$\textsf{MON} \vDash \cIsDefX {\textbf{A}_{\cFunTyX {\cProdTy {\cSetTy
        {M}} {\cSetTy {M}}} {\cSetTy {M}}}}$.  Since constants are
always defined by~\cite[Axiom A5.2]{Farmer25}, $\textbf{A}_{\cFunTyX
  {\cProdTy {\cSetTy {M}} {\cSetTy {M}}} {\cSetTy {M}}}$ beta-reduces
to a function abstraction by \cite[Axiom A4]{Farmer25}.  Since every
function abstraction is defined by~\cite[Axiom A5.11]{Farmer25}, we
have $(\star)$ by Quasi-Equality Substitution \cite[Lemma
  A.2]{Farmer25}.
\end{proof}

  \item $\mName{Def4}$: ${\cEqX {\mathsf{E}_{\cSetTy {M}}}
    {\cFinSetL {\mathsf{e}_M}}}$ \hfill (set identity element).

\begin{proof}[ that RHS is defined.]
We must show $(\star)$ $\textsf{MON} \vDash
\cIsDefX{\cFinSetL{\mathsf{e}_M}}$. Now $\cFinSetL{\mathsf{e}_M}$
stands for
\[
\cFunAppX{(\cFunAbsX{x_1} {M} {\cFunAbsX{x} {M} {x = x_1}})} {(\mathsf{e}_M)}.
\]
Since constants are always defined by~\cite[Axiom A5.2]{Farmer25},
$\cFinSetL{\mathsf{e}_M}$ beta-reduces to
\[
\cFunAbsX{x} {M} {x = \textsf{e}_M}
\]
by \cite[Axiom A4]{Farmer25}.  Since every function abstraction is
defined by~\cite[Axiom A5.11]{Farmer25}, we have $(\star)$ by
Quasi-Equality Substitution \cite[Lemma A.2]{Farmer25}.
\end{proof}

  \item $\mName{Thm9}$: ${\cForallX {x,y,z} {\cSetTy {M}} {\cEqX
      {\cBinX {x} {\odot} {\cBin {y} {\odot} {z}}} {\cBinX {\cBin {x}
          {\odot} {y}} {\odot} {z}}}}$
  \hfill ($\odot$ is  associative).

\begin{proof}[ of the theorem.]
Let $T = (L, \Gamma)$ be $\mathsf{MON}$ extended by $\mathsf{Def3}$.
We must show
\[(\star) {\sglsp} T \vDash \cForallX{x,y,z} {\cSetTy {M}} {\cEqX
      {\cBinX {x} {\odot} {\cBin {y} {\odot} {z}}} {\cBinX {\cBin {x}
          {\odot} {y}} {\odot} {z}}}.\]
\begin{align*}
\Gamma &\vDash
    {\cIsDefX {\cVar{x} {\cSetTy M}}} \land 
    {\cIsDefX {\cVar{y} {\cSetTy M}}} \land 
    {\cIsDefX {\cVar{z} {\cSetTy M}}}\tag{1}\\
\Gamma &\vDash
   {\cEqX
      {\cBinX {x} {\odot} {\cBin {y} {\odot} {z}}} {}}\\
      &\phantom{{}\vDash{}}
      \hspace*{2ex}
      {\cSet 
         {d} 
         {M} 
         {\cForsomeCX
            {a}
            {x}
            {b}
            {y}
            {c}
            {z}
            {\cEqX 
               {d} 
               {\cBinX {a} {\cdot} {\cBin {b} {\cdot} {c}}}}}}\tag{2}\\
\Gamma &\vDash
   {\cEqX
      {\cBinX {\cBin {x} {\odot} {y}} {\odot} {z}} {}}\\
      &\phantom{{}\vDash{}}
      \hspace*{2ex}
      {\cSet 
         {d} 
         {M} 
         {\cForsomeCX
            {a}
            {x}
            {b}
            {y}
            {c}
            {z}
            {\cEqX 
               {d} 
               {\cBinX {\cBin {a} {\cdot} {b}} {\cdot} {c}}}}}\tag{3}\\
\Gamma &\vDash
   {\cEqX
     {\cBinX {x} {\odot} {\cBin {y} {\odot} {z}}}
     {\cBinX {\cBin {x} {\odot} {y}} {\odot} {z}}}\tag{4}\\
\Gamma &\vDash
   {\cForallX 
      {x,y,z} 
      {\cSetTy M} 
      {\cEqX
         {\cBinX {x} {\odot} {\cBin {y} {\odot} {z}}} 
         {\cBinX {\cBin {x} {\odot} {y}} {\odot} {z}}}}\tag{5}
\end{align*}
(1)~follows from variables always being defined by~\cite[Axiom A5.1]{Farmer25};
(2) and (3)~follow from (1) and $\mathsf{Def3}$;
(4)~follows from (2) and (3) by Axiom 1 of~$T$; and
(5)~follows from (4) by Universal Generalization \cite[Theorem A.30]{Farmer25}
using the fact that $x$, $y$, and $z$ are not free in $\Gamma$ since 
$\Gamma$ is a set of sentences. Therefore, $(\star)$ holds.
\end{proof}

  \item $\mName{Thm10}$: ${\cForallX {x} {\cSetTy {M}} {\cBinBX
      {\cBinX {\mathsf{E}} {\odot} {x}} {=} {\cBinX {x}
        {\odot} {\mathsf{E}}} {=} {x}}}$\\ \phantom{x}
    \hfill ($\mathsf{E}$ is an identity element with
    respect to $\odot$).

\begin{proof}[ of the theorem.]
Let $T = (L, \Gamma)$ be \textsf{MON} extended by $\mathsf{Def3}$ and
$\mathsf{Def4}$. We must show
\[(\star) {\sglsp} T \vDash {\cForallX {x} {\cSetTy {M}} {\cBinBX
      {\cBinX {\mathsf{E}} {\odot} {x}} {=} {\cBinX {x}
        {\odot} {\mathsf{E}}} {=} {x}}}.\]
\begin{align*}
\Gamma &\vDash
    {\cIsDefX {\cVar{x} {\cSetTy M}}}\tag{1}\\
\Gamma &\vDash
    {\cIsDefX {\mathsf{E}}}\tag{2}\\
\Gamma &\vDash
   {\cEqX
      {\cBinX {\mathsf{E}} {\odot} {x}}
      {\cSet 
         {b} 
         {M} 
         {\cForsomeX
            {a}
            {x}
            {\cEqX 
               {b} 
               {\cBinX {\mathsf{e}} {\cdot} {a}}}}}}\tag{3}\\
\Gamma &\vDash
   {\cEqX
      {\cBinX {x} {\odot} {\mathsf{E}}}
      {\cSet 
         {b} 
         {M} 
         {\cForsomeX
            {a}
            {x}
            {\cEqX 
               {b} 
               {\cBinX {a} {\cdot} {\mathsf{e}}}}}}}\tag{4}\\
\Gamma &\vDash
   {\cBinBX
      {\cBinX {\mathsf{E}} {\odot} {x}} 
      {=} 
      {\cBinX {x} {\odot} {\mathsf{E}}} 
      {=} 
      {x}}\tag{5}\\
\Gamma &\vDash
   {\cForallX
      {x}
      {\cSetTy {M}}
      {\cBinBX
         {\cBinX {\mathsf{E}} {\odot} {x}} 
         {=} 
         {\cBinX {x} {\odot} {\mathsf{E}}} 
         {=} 
         {x}}}\tag{6}
\end{align*}
(1)~follows from variables always being defined by~\cite[Axiom A5.1]{Farmer25};
(2)~follows from constants always defined by~\cite[Axiom A5.2]{Farmer25};
(3) and (4)~follow from (1), (2), $\mathsf{Def3}$, and $\mathsf{Def4}$;
(5)~follows from (3) and (4) by Axiom 2 of~$T$; and
(6)~follows from (5) by Universal Generalization \cite[Theorem A.30]{Farmer25}
using the fact that $x$ is not free in $\Gamma$ since 
$\Gamma$ is a set of sentences. Therefore, $(\star)$ holds.
\end{proof}

  \item \textsf{Thm11 (Thm1-via-MON-to-opposite-monoid)}:\\
  \hspace*{2ex} ${\cMonoid {\cUnivSetPC {M}} {\cdot^{\rm op}_{\cFunTyX
        {\cProdTy {M} {M}} {M}}} {\mathsf{e}_M}}$
  \hfill (opposite monoids are monoids).
  
\begin{proof}[ of the theorem.] 
Let $T$ be the top theory of \textsf{MON-1}.  We must show $T \vDash
\textsf{Thm11}$.  We have previously proved $(\star)$ $\textsf{MON}
\vDash \textsf{Thm1}$.  $\Phi = \textsf{MON-to-opposite-monoid}$ is a
development morphism from \textsf{MON} to \textsf{MON-1}, and so
$\widetilde{\Phi} = (\mu,\nu)$ is a theory morphism from \textsf{MON}
to $T$.  Thus $(\star)$ implies $T \vDash \nu(\textsf{Thm1})$.
Therefore, $T \vDash \textsf{Thm11}$ since $\textsf{Thm11} =
\nu(\textsf{Thm1})$.
\end{proof}

  \item \textsf{Thm12 (Thm1-via-MON-to-set-monoid)}:\\
  \hspace*{2ex}
  ${\cMonoid
     {\cUnivSetPC {\cSetTy {M}}}
     {\odot_{\cFunTyX {\cProdTy {\cSetTy {M}} {\cSetTy {M}}} {\cSetTy {M}}}} 
     {\mathsf{E}_{\cSetTy {M}}}}$\\
  \phantom{x} \hfill (set monoids are monoids).

\begin{proof}[ of the theorem.] 
Similar to the proof of \textsf{Thm11}.
\end{proof}

\ee

\subsection{Development of \textnormal{\textsf{COM-MON}}}

\be

  \item $\mName{Thm13}$: ${\cComMonoid {\cUnivSetPC {M}}
    {\cdot_{\cFunTyX {\cProdTy {M} {M}} {M}}} {\mathsf{e}_M}}$\\
  \phantom{x} \hfill (models of \textsf{COM-MON} define commutative monoids).

\begin{proof}[ of the theorem.]
Let $T = (L, \Gamma)$ be $\textsf{COM-MON}$. We must show
\[(\star) {\sglsp} T \vDash \cComMonoid{\cUnivSetPC {M}} {\cdot_{\cFunTyX {\cProdTy {M} {M}} {M}}} {\mathsf{e}_M}.\]
\begin{align*}
\Gamma &\vDash
    {\cMonoid 
       {\cUnivSetPC {M}} 
       {\cdot_{\cFunTyX {\cProdTy {M} {M}} {M}}} 
       {\mathsf{e}_M}}\tag{1}\\
\Gamma &\vDash
    {\cForallX {x,y} {\cUnivSetPC M} {x \cdot y = y \cdot x}}\tag{2}
\end{align*}
(1)~follows from $\textsf{MON} \le T$ and the fact that 
$\mName{Thm1}$ is a theorem of $\textsf{MON}$; and 
(2)~follows from Axiom 3 of $T$ and part 5 of Lemma~\ref{lem:univ-sets}.
Therefore, $(\star)$~follows from (1), (2), and the notational definition of 
$\textsf{COM-MONOID}$ given in Table~\ref{tab:monoids-abbr}. 
\end{proof}

  \item $\mName{Def5}$: ${\cEqX {\mathsf{\le}_{\cFunTyBX {M} {M}
        {\cB}}} {\cFunAbsX {x,y} {M} {\cForsomeX {z} {M} {\cEqX
          {\cBinX {x} {\cdot} {z}} {y}}}}}$ \hfill (weak order).

\begin{proof}[ that RHS is defined.]
Similar to the proof that the RHS of $\textsf{Def1}$ is defined.
\end{proof}

  \item $\mName{Thm14}$: ${\cForallX {x} {M} {\cBinX {x} {\le} {x}}}$\hfill
    (reflexivity).

\begin{proof}[ of the theorem.]
Let $T = (L, \Gamma)$ be $\textsf{COM-MON}$ extended by \textsf{Def5}.
We must show
\[(\star) {\sglsp} T \vDash {\cForallX {x} {M} {\cBinX {x} {\le} {x}}}.\]
\begin{align*}
\Gamma &\vDash
    {\cIsDefX {\cVar {x} {M}}}\tag{1}\\
\Gamma &\vDash
    {\cQuasiEqX 
       {\cBin {x} {\le} {x}}
       {\cForsome {z} {M} {\cEqX {\cBinX {x} {\cdot} {z}} {x}}}}\tag{2}\\
\Gamma &\vDash
    {\cEqX {\cBinX {x} {\cdot} {\mathsf{e}}} {x}}\tag{3}\\
\Gamma &\vDash
    {\cForsomeX {z} {M} {\cEqX {\cBinX {x} {\cdot} {z}} {x}}}\tag{4}\\
\Gamma &\vDash
    {\cBinX {x} {\le} {x}}\tag{5}\\
\Gamma &\vDash
    {\cForallX {x} {M} {\cBinX {x} {\le} {x}}}\tag{6}
\end{align*}
(1)~follows from variables always being defined by~\cite[Axiom A5.1]{Farmer25};
(2)~follows from \textsf{Def5} and Extensionality~\cite[Axiom A3]{Farmer25} using
the Substitution Rule~\cite[Theorem A.31]{Farmer25} and 
Beta-Reduction~\cite[Axiom A4]{Farmer25};
(3)~follows from (1) and Axiom 2 of $T$ by Universal 
Instantiation~\cite[Theorem A.14]{Farmer25}; 
(4)~follows from (3) by Existential Generalization \cite[Theorem A.51]{Farmer25}; 
(5)~follows from (2) and (4) by Rule $\mbox{R2}'$ \cite[Lemma A.2]{Farmer25}; and
(6)~follows from (5) by Universal Generalization \cite[Theorem A.30]{Farmer25} 
using the fact that $x$ is not free in $\Gamma$ since $\Gamma$ is a set of sentences.
Therefore, $(\star)$ holds.
\end{proof}

  \item $\mName{Thm15}$: ${\cForallX {x,y,z} {M} {\cImpliesX {\cAnd
        {\cBinX {x} {\le} {y}} {\cBinX {y} {\le} {z}}} {\cBinX {x}
        {\le} {z}}}}$ \hfill (transitivity).

\begin{proof}[ of the theorem.]
Let $\textbf{A}_\cB$ be ${\cAnd {\cBinX {x} {\le} {y}} {\cBinX {y}
    {\le} {z}}}$, $\textbf{B}_\cB$ be ${\cEqX {\cBinX {x} {\cdot} {u}}
  {y}}$, and $\textbf{C}_\cB$ be ${\cEqX {\cBinX {y} {\cdot} {v}}
  {z}}$ (where these variables all have type $M$).  Also let $T = (L,
\Gamma)$ be $\textsf{COM-MON}$ extended by \textsf{Def5}.  We must
show
\[(\star) {\sglsp} T \vDash 
{\cForallX {x,y,z} {M} {\cImpliesX {\textbf{A}_\cB} {\cBinX {x} {\le} {z}}}}.\]
\begin{align*}
\Gamma \cup \mSet{\textbf{B}_\cB, \textbf{C}_\cB} &\vDash
   {\cIsDefX {\cVar {x} {M}}} \land
   {\cIsDefX {\cVar {y} {M}}} \land
   {\cIsDefX {\cVar {z} {M}}} \land
   {\cIsDefX {\cVar {u} {M}}} \land {}\\
&\phantom{{}\vDash{}}
   \hspace*{2ex}
   {\cIsDefX {\cVar {v} {M}}}\tag{1}\\
\Gamma \cup \mSet{\textbf{B}_\cB, \textbf{C}_\cB} &\vDash
   {\cEqX
      {\cBinX
         {\cBin {x} {\cdot} {u}} 
         {\cdot}
         {v}}
      {z}}\tag{2}\\
\Gamma \cup \mSet{\textbf{B}_\cB, \textbf{C}_\cB} &\vDash
   {\cEqX
      {\cBinX
         {\cBin {x} {\cdot} {u}} 
         {\cdot}
         {v}}
      {\cBinX
         {x}
         {\cdot}
         {\cBin {u} {\cdot} {v}}}}\tag{3}\\
\Gamma \cup \mSet{\textbf{B}_\cB, \textbf{C}_\cB} &\vDash
   {\cEqX
      {\cBinX
         {x}
         {\cdot}
         {\cBin {u} {\cdot} {v}}}
      {z}}\tag{4}\\
\Gamma \cup \mSet{\textbf{B}_\cB, \textbf{C}_\cB} &\vDash
   {\cForsomeX 
      {w} 
      {M} 
      {\cEqX
         {\cBinX {x} {\cdot} {w}}
         {z}}}\tag{5}\\
\Gamma \cup \mSet{\textbf{B}_\cB} &\vDash
   {\cImpliesX
      {\cEq {\cBinX {y} {\cdot} {v}} {z}}
      {\cForsome {w} {M} {\cEqX {\cBinX {x} {\cdot} {w}} {z}}}}\tag{6}\\
\Gamma \cup \mSet{\textbf{B}_\cB} &\vDash
   {\cImpliesX
      {\cForsome {v} {M} {\cEqX {\cBinX {y} {\cdot} {v}} {z}}}
      {\cForsome {w} {M} {\cEqX {\cBinX {x} {\cdot} {w}} {z}}}}\tag{7}\\
\Gamma &\vDash
   {\cImpliesX
      {\cEq {\cBinX {x} {\cdot} {u}} {y}} {}}\\
       &\phantom{{}\vDash{}}
   \hspace*{2ex}
   {\cImplies
      {\cForsome {v} {M} {\cEqX {\cBinX {y} {\cdot} {v}} {z}}}
      {\cForsome {w} {M} {\cEqX {\cBinX {x} {\cdot} {w}} {z}}}}\tag{8}\\
\Gamma &\vDash
   {\cImpliesX
      {\cForsome {u} {M} {\cEqX {\cBinX {x} {\cdot} {u}} {y}}} {}}\\
       &\phantom{{}\vDash{}}
   \hspace*{2ex}
   {\cImplies
      {\cForsome {v} {M} {\cEqX {\cBinX {y} {\cdot} {v}} {z}}}
      {\cForsome {w} {M} {\cEqX {\cBinX {x} {\cdot} {w}} {z}}}}\tag{9}\\
\Gamma &\vDash
   {\cImpliesX
      {\cBinX {x} {\le} {y}}
      {\cImplies
         {\cBinX {y} {\le} {z}}
         {\cBinX {x} {\le} {z}}}}\tag{10}\\
\Gamma &\vDash
   {\cImpliesX
      {\textbf{A}_\cB}
      {\cBinX {x} {\le} {z}}}\tag{11}\\
\Gamma &\vDash
    {\cForallX {x,y,z} {M} {\cImpliesX {\textbf{A}_\cB} {\cBinX {x}
        {\le} {z}}}}\tag{12}
\end{align*}
(1)~follows from variables always being defined by~\cite[Axiom A5.1]{Farmer25};
(2)~follows from $\textbf{B}_\cB$ and $\textbf{C}_\cB$ by the 
Equality Rules \cite[Lemma A.13]{Farmer25};
(3)~follows from Axiom 1 of $T$ by Universal 
Instantiation~\cite[Theorem A.14]{Farmer25};
(4)~follows from (2) and (3) by the Equality Rules \cite[Lemma A.13]{Farmer25};
(5)~follows from (1), (4), and $\mName{Thm2}$ by Existential Generalization 
\cite[Theorem A.51]{Farmer25};
(6) and (8)~follow from (5) and (7), respectively, by the Deduction Theorem 
\cite[Theorem A.50]{Farmer25};
(7)~and (9)~follow from (6) and (8), respectively, by 
Existential Instantiation~\cite[Theorem A.52]{Farmer25};
(10)~follows from (1), (9), and $\mName{Def5}$ by Beta-Reduction 
\cite[Axiom A4]{Farmer25} and 
Alpha-Conversion \cite[Theorem A.18]{Farmer25};
(11)~follows from (10) by the Tautology Rule \cite[Corollary A.46]{Farmer25}; and
(12)~follows from (11) by Universal Generalization \cite[Theorem A.30]{Farmer25} 
using the fact that $x$, $y$, and $z$ are not free in $\Gamma$ 
since $\Gamma$ is a set of sentences.
Therefore, $(\star)$ holds.
\end{proof}

\ee

\subsection{Development of \textnormal{\textsf{FUN-COMP}}} 

\be

  \item $\mName{Thm16}$: 
  ${\cForallCX 
     {f} {\cFunTyX {\cBaseTy A} {\cBaseTy B}} 
     {g} {\cFunTyX {\cBaseTy B} {\cBaseTy C}} 
     {h} {\cFunTyX {\cBaseTy C} {\cBaseTy D}} 
     {\cEqX 
        {\cFunCompX {f} {\cFunComp {g} {h}}}
        {\cFunCompX {\cFunComp {f} {g}} {h}}}}$\\
  \phantom{x} \hfill ($\circ$ is associative).

\begin{proof}[ of the theorem.]
Let $\textbf{A}_\cB$ be the theorem and $T = (L, \Gamma)$ be
$\textsf{FUN-COMP}$. We must show $(\star) {\sglsp} T \vDash
\textbf{A}_\cB$.
\begin{align*}
\Gamma &\vDash
   {\cIsDefX {\cVar{f} {\cFunTyX{A} {B}}}} \land
   {\cIsDefX {\cVar{g} {\cFunTyX{B} {C}}}} \land
   {\cIsDefX {\cVar{h} {\cFunTyX{C} {D}}}} \land
   {\cIsDefX {\cVar{x} {A}}}\tag{1}\\
\Gamma &\vDash
    {\cQuasiEqX
       {\cFunAppX {((f \circ g) \circ h)} {x}} 
       {\cFunAppX {h} {\cFunApp {g} {\cFunApp {f} {x}}}}}\tag{2}\\
\Gamma &\vDash
    {\cQuasiEqX
       {\cFunAppX {(f \circ (g \circ h))} {x}} 
       {\cFunAppX {h} {\cFunApp {g} {\cFunApp {f} {x}}}}}\tag{3}\\
\Gamma &\vDash
    {\cQuasiEqX
       {\cFunAppX {((f \circ g) \circ h)} {x}}
       {\cFunAppX {(f \circ (g \circ h))} {x}}}\tag{4}\\ 
\Gamma &\vDash
    {\cForallX {x}{A}
       {\cFunAppX {(f \circ (g \circ h))} {x}} \simeq
       {\cFunAppX {((f \circ g) \circ h)} {x}}}\tag{5}\\
\Gamma &\vDash
    f \circ (g \circ h) = (f \circ g) \circ h\tag{6}\\
\Gamma &\vDash \textbf{A}_\cB\tag{7}
\end{align*}
(1)~follows from variables always being defined by~\cite[Axiom A5.1]{Farmer25}; 
(2) and (3) both follow from (1), the definition of $\circ$ in 
Table~\ref{tab:monoids-pc}, function abstractions are always defined by 
\cite[Axiom A5.11]{Farmer25}, ordered pairs of defined components are always 
defined by \cite[Axiom A7.1]{Farmer25}, Beta-Reduction 
\cite[Axiom A4]{Farmer25}, and 
Quasi-Equality Substitution \cite[Lemma A.2]{Farmer25};
(4) follows from (2) and (3) by the Quasi-Equality Rules \cite[Lemma A.4]{Farmer25};
(5)~follows from (4) by Universal Generalization \cite[Theorem A.30]{Farmer25} 
using the fact that $x$ is not free in $\Gamma$ since $\Gamma$ is a set of sentences;
(6)~follows from (5) by Extensionality \cite[Axiom A3]{Farmer25}; and 
(7) follows from (6) by Universal Generalization using the fact that 
$f$, $g$, and $h$ are not free in $\Gamma$ since $\Gamma$ is a set of sentences. 
Therefore, $(\star)$ holds.
\end{proof}

  \item $\mName{Thm17}$:
  ${\cForallX 
     {f} {\cFunTyX {\cBaseTy A} {\cBaseTy B}} 
     {\cBinBX 
        {\cFunCompX {\cIdFunPC {A}} {f}}
        {=}
        {\cFunCompX {f} {\cIdFunPC {B}}}
        {=}
        {f}}}$\\
  \phantom{x} \hfill (identity functions are left and right identity elements).

\begin{proof}[ of the theorem.]
Let $\textbf{A}_\cB$ be the theorem and $T = (L, \Gamma)$ be
$\textsf{FUN-COMP}$. We must show $(\star) {\sglsp} T \vDash
\textbf{A}_\cB$.
\begin{align*}
\Gamma &\vDash
   {\cAndX
      {\cIsDefX {\cVar {f} {\cFunTyX{A} {B}}}}
      {\cIsDefX {\cVar {x} {A}}}}\tag{1}\\
\Gamma &\vDash
   {\cQuasiEqX
      {\cFunAppX {\cBin {\cIdFunPC {A}} {\circ} {f}} {x}}
      {\cFunAppX {f} {x}}}\tag{2}\\
\Gamma &\vDash
   {\cQuasiEqX
      {\cFunAppX {\cBin {f} {\circ} {\cIdFunPC {B}}} {x}}
      {\cFunAppX {f} {x}}}\tag{3}\\
\Gamma &\vDash
   {\cForallX
      {x}
      {A}
      {\cQuasiEqX
         {\cFunAppX {\cBin {\cIdFunPC {A}} {\circ} {f}} {x}}
         {\cFunAppX {f} {x}}}}\tag{4}\\
\Gamma &\vDash
   {\cForallX
      {x}
      {A}
      {\cQuasiEqX
         {\cFunAppX {\cBin {f} {\circ} {\cIdFunPC {B}}} {x}}
         {\cFunAppX {f} {x}}}}\tag{5}\\
\Gamma &\vDash
   {\cEqX {\cBinX {\cIdFunPC {A}} {\circ} {f}} {f}}\tag{6}\\ 
\Gamma &\vDash
   {\cEqX {\cBinX {f} {\circ} {\cIdFunPC {B}}} {f}}\tag{7}\\       
\Gamma &\vDash
   {\cBinBX 
      {\cFunCompX {\cIdFunPC {A}} {f}}
      {=}
      {\cFunCompX {f} {\cIdFunPC {B}}}
      {=}
      {f}}\tag{8}\\
\Gamma &\vDash \textbf{A}_\cB\tag{9}
\end{align*}
(1)~follows from variables always being defined by~\cite[Axiom A5.1]{Farmer25}; 
(2) and (3) both follow from (1), the definitions of $\mathsf{id}$ and $\circ$ in 
Table~\ref{tab:monoids-pc}, function abstractions are always defined by 
\cite[Axiom A5.11]{Farmer25}, 
ordered pairs of defined components are always defined by \cite[Axiom A7.1]{Farmer25}, 
Beta-Reduction \cite[Axiom A4]{Farmer25}, and 
Quasi-Equality Substitution \cite[Lemma A.2]{Farmer25};
(4) and (5) both follow from (2) and (3), respectively, by Universal 
Generalization \cite[Theorem A.30]{Farmer25} using the fact that $x$ is not free 
in $\Gamma$ since $\Gamma$ is a set of sentences;
(6) and (7) follow from (4)~and (5), respectively, by Extensionality 
\cite[Axiom A3]{Farmer25}; 
(8)~follows from (6) and (7) by the Equality Rules~\cite[Lemma A.13]{Farmer25};
and (9) follows from (8) by Universal Generalization using the fact that 
$f$ is not free in $\Gamma$ since $\Gamma$ is a set of sentences. 
Therefore, $(\star)$ holds.
\end{proof}

\ee

\subsection{Development of \textnormal{\textsf{ONE-BT}}} 

\be

  \item \textsf{Thm18 (Thm16-via-FUN-COMP-to-ONE-BT)}:\\
  \hspace*{2ex}
  ${\cForallX 
     {f,g,h} 
     {\cFunTyX {\cBaseTy S} {\cBaseTy S}} 
     {\cEqX 
        {\cFunCompX {f} {\cFunComp {g} {h}}}
        {\cFunCompX {\cFunComp {f} {g}} {h}}}}$
  \hfill ($\circ$ is associative).

\begin{proof}[ of the theorem.] 
Similar to the proof of \textsf{Thm11}.
\end{proof}

  \item \textsf{Thm19 (Thm17-via-FUN-COMP-to-ONE-BT)}:\\
  \hspace*{2ex}
  ${\cForallX 
     {f} {\cFunTyX {\cBaseTy S} {\cBaseTy S}} 
     {\cBinBX 
        {\cFunCompX {\cIdFunPC {S}} {f}}
        {=}
        {\cFunCompX {f} {\cIdFunPC {S}}}
        {=}
        {f}}}$\\
  \phantom{x} \hfill (${\cIdFunPC {S}}$ is an identity element with respect to $\circ$).

\begin{proof}[ of the theorem.] 
Similar to the proof of \textsf{Thm11}.
\end{proof}

  \item \textsf{Def6 (Def1-via-MON-to-ONE-BT)}:\\
  \hspace*{2ex} 
  ${\cEqX 
      {\textsf{trans-monoid}_{\cFunTyX {\cSetTy {\cFunTyX {S} {S}}} {\cB}}} {}}$\\
      \hspace*{2ex}
     ${\cFunAbsX 
         {s} 
         {\cSetTy {\cFunTyX {S} {S}}} {}}$\\
         \hspace*{4ex}
        ${\cNotEqX {s} {\cEmpSetPC {\cFunTyX {S} {S}}}} \And {}$\\
         \hspace*{4ex}
        ${\cTotalOn 
            {{\cFunCompPairPC {S} {S} {S}}\wrestricted_{\cProdTyX {s} {s}}} 
            {\cProdTyX {s} {s}} 
            {s}} \And {}$\\
         \hspace*{4ex}
        ${\cInX {\cIdFunPC {S}} {s}}$ 
  \hfill (transformation monoid).

\begin{proof}[ that RHS is defined.] 
Let $\textbf{A}^{1}_{\cFunTyX {\cSetTy {M}} {\cB}}$ be the RHS of
\textsf{Def1}, $\textbf{A}^{2}_{\cFunTyX {\cSetTy {\cFunTyX {S} {S}}}
  {\cB}}$ be the RHS of \textsf{Def6}, $T_1$ be \textsf{MON}, and
$T_2$ be \textsf{ONE-BT}, the top theory of \textsf{ONE-BT-1}.  We
must show $T_2 \vDash \cIsDefX {\textbf{A}^{2}_{\cFunTyX {\cSetTy
      {\cFunTyX {S} {S}}} {\cB}}}$.  We have previously proved
$(\star)$ $T_1 \vDash \cIsDefX {\textbf{A}^{1}_{\cFunTyX {\cSetTy {M}}
    {\cB}}}$.  $\textsf{MON-to-ONE-BT} = (\mu,\nu)$ is a theory
morphism from $T_1$ to $T_2$.  Thus $(\star)$ implies $T_2 \vDash
\nu(\cIsDefX{\textbf{A}^{1}_{\cFunTyX {\cSetTy {M}} {\cB}}})$.
Therefore, $T_2 \vDash \cIsDefX {\textbf{A}^{2}_{\cFunTyX {\cSetTy
      {\cFunTyX {S} {S}}} {\cB}}}$ since ${\textbf{A}^{2}_{\cFunTyX
    {\cSetTy {\cFunTyX {S} {S}}} {\cB}}} =
\nu(\textbf{A}^{1}_{\cFunTyX {\cSetTy {M}} {\cB}})$.
\end{proof}

  \item \textsf{Thm20 (Thm4-via-MON-to-ONE-BT)}:\\
  \hspace*{2ex}
  ${\cForallX 
      {s} 
      {\cSetTy {\cFunTyX {S} {S}}} {}}$\\
      \hspace*{4ex}
     ${\cImpliesX
         {\cFunAppX {\textsf{trans-monoid}} {s}} 
         {\cMonoid 
            {s}
            {{\cFunCompPairPC {S} {S} {S}}\wrestricted_{\cProdTyX {s} {s}}}
            {\cIdFunPC {S}}}}$\\ 
  \phantom{x} \hfill (transformation monoids are monoids).

\begin{proof}[ of the theorem.] 
Similar to the proof of \textsf{Thm11}.
\end{proof}

\ee

\subsection{Development of \textnormal{\textsf{MON-ACT}}} 

\be

  \item $\mName{Thm21}$: ${\cMonAction {\cUnivSetPC {M}}
    {\cUnivSetPC {S}} {\cdot_{\cFunTyX {\cProdTy {M} {M}} {M}}}
    {\mathsf{e}_M} {\mName{act}_{\cFunTyX {\cProdTy {M} {S}}
        {S}}}}$\\ 
  \phantom{x} \hfill (models of \textsf{MON-ACT} define monoid actions).

\begin{proof}[ of the theorem.] 
Let $T = (L, \Gamma)$ be $\textsf{MON-ACT}$. We must show
\[(\star) {\sglsp} T \vDash {\cMonAction {\cUnivSetPC {M}}
    {\cUnivSetPC {S}} {\cdot_{\cFunTyX {\cProdTy {M} {M}} {M}}}
    {\mathsf{e}_M} {\mName{act}_{\cFunTyX {\cProdTy {M} {S}}
        {S}}}}.\]
\begin{align*}
\Gamma &\vDash
    {\cMonoid
      {\cUnivSetPC {M}} 
      {\cdot_{\cFunTyX{\cProdTy {M} {M}} {M}}} 
      {\textsf{e}_M}}\tag{1}\\
\Gamma &\vDash
    {\cIsDefX {\cUnivSetPC S}}\tag{2}\\
\Gamma &\vDash
    \cUnivSetPC{S} \not= \cEmpSetPC{S}\tag{3}\\
\Gamma &\vDash
    {\cIsDefInQTyX
       {\mName{act}_{\cFunTyX {\cProdTy {M} {S}} {S}}}
       {\cFunQTyX
          {\cProdQTy {\cUnivSetPC {M}} {\cUnivSetPC {S}}}
          {\cUnivSetPC {S}}}}\tag{4}\\
\Gamma &\vDash
    \cForallBX{x, y} {\cUnivSetPC M}
    {s} {\cUnivSetPC S}
    {x \: \textsf{act} \: (y \: \textsf{act} \: s)
    =
    (x \cdot y) \: \textsf{act} \: s}\tag{5}\\
\Gamma &\vDash
    \cForallX{s} {\cUnivSetPC S} {\textsf{e} \: \textsf{act} \: s = s}\tag{6}\\
\Gamma &\vDash
    {\cMonAction {\cUnivSetPC {M}}
    {\cUnivSetPC {S}} {\cdot_{\cFunTyX {\cProdTy {M} {M}} {M}}}
    {\mathsf{e}_M} {\mName{act}_{\cFunTyX {\cProdTy {M} {S}}
        {S}}}}\tag{7}
\end{align*}
(1)~follows from $\textsf{MON} \vDash \textsf{Thm1}$ and $\textsf{MON} \le T$; 
(2) and (3)~follow from parts 1 and 2, respectively, of Lemma~\ref{lem:univ-sets}; 
(4) follows from~\cite[Axiom~5.2]{Farmer25} and parts 8--10 of Lemma~\ref{lem:univ-sets}; 
(5) and (6) follow from Axioms 3 and~4, respectively, of $T$ and part 5 of
Lemma~\ref{lem:univ-sets}; 
and (7)~follows from (1)--(6) and the definition of $\textsf{MON-ACTION}$ in 
Table~\ref{tab:monoids-abbr}.  Therefore, $(\star)$ holds.
\end{proof}

  \item[] $\mName{Thm22}$: ${\cTotal {\mName{act}_{\cFunTyX {\cProdTy
          {M} {S}} {S}}}}$ \hfill ($\mName{act}$ is total).

\begin{proof}[ of the theorem.]  
Let $T = (L, \Gamma)$ be \textsf{MON-ACT}.  $T \vDash {\cTotal
  {\mName{act}_{\cFunTyX {\cProdTy {M} {S}} {S}}}}$ follows from Axiom
3 of $T$ in the same way that $T \vDash {\cTotal {\cdot_{\cFunTyX
      {\cProdTy {M} {M}} {M}}}}$ follows from Axiom 1 of \textsf{MON}
as shown in the proof of $\mName{Thm2}$.
\end{proof}

  \item $\mName{Def7}$: ${\cEqX {\mName{orbit}_{\cFunTyX {S} {\cSetTy
          {S}}}} {\cFunAbsX {s} {S} {\cSet {t} {S} {\cForsomeX {x} {M}
          {\cEqX {\cBinX {x} {\mName{act}} {s}} {t}}}}}}$\hfill
    (orbit).

\begin{proof}[ that RHS is defined.] 
Similar to the proof that the RHS of \textsf{Def1} is defined.
\end{proof}

  \item $\mName{Def8}$: ${\cEqX {\mName{stabilizer}_{\cFunTyX {S}
        {\cSetTy {M}}}} {\cFunAbsX {s} {S} {\cSet {x} {M} {\cEqX
          {\cBinX {x} {\mName{act}} {s}} {s}}}}}$\hfill (stabilizer).

\begin{proof}[ that RHS is defined.] 
Similar to the proof that the RHS of \textsf{Def1} is defined.
\end{proof}

  \item $\mName{Thm23}$: ${\cForallX {s} {S} {\cFunAppX
      {\mName{submonoid}} {\cFunApp {\mName {stabilizer}} {s}}}}$
    \hfill (stabilizers are submonoids).

\begin{proof}[ of the theorem.]
Let $T = (L, \Gamma)$ be \textsf{MON-ACT} extended by \textsf{Def7}
and \textsf{Def8}. We must show
\[(\star) {\sglsp} T \vDash {\cForallX {s} {S} {\cFunAppX
      {\mName{submonoid}} {\cFunApp {\mName {stabilizer}} {s}}}}.\]
\begin{align*}
\Gamma &\vDash
    {\cIsDefX {\cVar{s} {S}}}\tag{1}\\
\Gamma &\vDash
    {\cIsDefX {\textsf{e}_M}}\tag{2}\\
\Gamma &\vDash
    {\cEqX 
       {\cFunApp {\mName{stabilizer}} {s}} 
        {\cSet {x} {M} {\cEqX {\cBinX {x} {\mName{act}} {s}} {s}}}}\tag{3}\\
\Gamma &\vDash
    {\textsf{e} \in {\cFunApp {\mName {stabilizer}} {s}}}\tag{4}\\
\Gamma &\vDash
    {\cFunApp {\mName {stabilizer}} {s}} \not= \cEmpSetPC{M}\tag{5}\\
\Gamma &\vDash
    {\cIsDefInQTyX 
       {\cRestrictX 
          {\cdot}  
          {\cProdTyX 
             {\cFunApp {\mName {stabilizer}} {s}} 
             {\cFunApp {\mName {stabilizer}} {s}}}} {}}\\
       &\phantom{{}\vDash{}}
       {\cFunTyX 
          {\cProdTy 
             {\cFunApp {\mName {stabilizer}} {s}}
             {\cFunApp {\mName {stabilizer}} {s}}}
          {\cFunApp {\mName {stabilizer}} {s}}}\tag{6}\\
\Gamma &\vDash
    \cIsDefX{\cFunApp {\mName {stabilizer}} {s}}\tag{7}\\
\Gamma &\vDash
    {\cFunAppX {\mName{submonoid}} {\cFunApp {\mName {stabilizer}} {s}}} = {}\\
       &\phantom{{}\vDash{}}
       {\cNotEqX {\cFunApp {\mName {stabilizer}} {s}} {\cEmpSetPC{M}}} \land {}\\
       &\phantom{{}\vDash{}}
    {\cIsDefInQTyX 
       {\cRestrictX 
          {\cdot}  
          {\cProdTyX 
             {\cFunApp {\mName {stabilizer}} {s}} 
             {\cFunApp {\mName {stabilizer}} {s}}}} {}}\\
       &\phantom{{}\vDash{}}
       {\cFunTyX 
          {\cProdTy 
             {\cFunApp {\mName {stabilizer}} {s}}
             {\cFunApp {\mName {stabilizer}} {s}}}
          {\cFunApp {\mName {stabilizer}} {s}}} \land {}\\
       &\phantom{{}\vDash{}}
       {\textsf{e} \in {\cFunApp {\mName {stabilizer}} {s}}}\tag{8}\\ 
\Gamma &\vDash
    \cFunAppX
      {\mName{submonoid}} {\cFunApp {\mName {stabilizer}} {s}}\tag{9}\\
\Gamma &\vDash
    \cForallX{s} {S} {\cFunAppX
      {\mName{submonoid}} {\cFunApp {\mName {stabilizer}} {s}}}\tag{10}
\end{align*}
(1)~follows from variables always being defined by~\cite[Axiom A5.1]{Farmer25}; 
(2)~follows from constants always being defined by~\cite[Axiom A5.2]{Farmer25}; 
(3)~follows from \textsf{Def8} by the Equality Rules~\cite[Lemma A.13]{Farmer25} and
Beta-Reduction \cite[Axiom A4]{Farmer25} applied to (1) and the RHS of the result; 
(4)~follows from (3) and Axiom 4 of $T$; 
(5)~follows immediately from (4); 
(6)~follows from \textsf{Thm2}, (3), and Axiom 3 of $T$;
(7)~follows from (3) and \cite[Axiom A5.4]{Farmer25};
(8)~follows from \textsf{Def1} by the Equality Rules and Beta-Reduction
applied to (7) and the RHS of the result;
(9)~follows from (4), (5), (6), and (8) by the Tautology 
Rule~\cite[Corollary A.46]{Farmer25};
(10)~follows from (9) by Universal Generalization \cite[Theorem A.30]{Farmer25} 
using the fact that $s$ is free in $\Gamma$ because $\Gamma$ is a set of sentences. 
Therefore, $(\star)$ holds.
\end{proof}

  \item \textsf{Thm24 (Thm21-via-MON-ACT-to-MON)}:\\
  \hspace*{2ex} ${\cMonAction {\cUnivSetPC {M}} {\cUnivSetPC {M}}
    {\cdot_{\cFunTyX {\cProdTy {M} {M}} {M}}} {\mathsf{e}_M}
    {\cdot_{\cFunTyX {\cProdTy {M} {M}} {M}}}}$\\ 
  \phantom{x} \hfill (first example is a monoid action).

\begin{proof}[ of the theorem.] 
Similar to the proof of \textsf{Thm11}.
\end{proof}

\ee

\subsection{Development of \textnormal{\textsf{ONE-BT-with-SC}}} 

\be

  \item \textsf{Thm25 (Thm21-via-MON-ACT-to-ONE-BT-with-SC)}:\\
  \hspace*{2ex} 
  $\textsf{MON-ACTION}(
  {\mathsf{F}_{\cSetTy {\cFunTyX {S} {S}}}},\\
  \hspace*{18.2ex}
  {\cUnivSetPC {S}},\\ 
  \hspace*{18.2ex}
  {\circ_{\cFunTyX {\cProdTy {\cFunTy {S} {S}} {\cFunTy {S} {S}}}
        {\cFunTy {S} {S}}}}{\wrestricted_{\cProdQTyX {\mathsf{F}}
        {\mathsf{F}}}},\\
  \hspace*{18.2ex}
  {\cIdFunPC {S}},\\
  \hspace*{18.2ex}
  {\cFunAppPairPC {S} {S}}{\wrestricted_{\cProdQTyX {\mathsf{F}} {S}}})$\\
  \phantom{x} \hfill (second example is a monoid action).

\begin{proof}[ of the theorem.] 
Similar to the proof of \textsf{Thm11}.
\end{proof}

\ee

\subsection{Development of \textnormal{\textsf{MON-HOM}}} 

\be

  \item $\mName{Thm26}$:\\
  \hspace*{2ex} 
  $\textsf{MON-HOM}(
  {\cUnivSetPC {M_1}},\\
  \hspace*{15.2ex}
  {\cUnivSetPC {M_2}},\\
  \hspace*{15.2ex}
  {\cdot_{\cFunTyX {\cProdTy {M_1} {M_1}} {M_1}}},\\
  \hspace*{15.2ex}
  {\mathsf{e}_{M_1}},\\
  \hspace*{15.2ex}
  {\cdot_{\cFunTyX {\cProdTy {M_2} {M_2}} {M_2}}},\\
  \hspace*{15.2ex}
  {\mathsf{e}_{M_2}},\\
  \hspace*{15.2ex}
  {\mathsf{h}_{\cFunTyX {M_1} {M_2}}})$\\ 
  \phantom{x} \hfill (models of \textsf{MON-HOM} define monoid homomorphisms).

\begin{proof}[ of the theorem.]
Let $T = (L, \Gamma)$ be $\textsf{MON-HOM}$ and $\textbf{A}_\cB$ be
\begin{align*}
& \textsf{MON-HOM}(
  {\cUnivSetPC {M_1}},
  {\cUnivSetPC {M_2}},
  {\cdot_{\cFunTyX {\cProdTy {M_1} {M_1}} {M_1}}},
  {\mathsf{e}_{M_1}},
  {\cdot_{\cFunTyX {\cProdTy {M_2} {M_2}} {M_2}}},\\
& \hspace{2ex}
  {\mathsf{e}_{M_2}},
  {\mathsf{h}_{\cFunTyX {M_1} {M_2}}}).
\end{align*}
We must show $(\star) {\sglsp} T \vDash \textbf{A}_\cB$.
\begin{align*}
\Gamma &\vDash
   {\cMonoid{\cUnivSetPC{M_1}} {\cdot_{\cFunTyX{\cProdTy {M_1} {M_1}} {M_1}}} {\textsf{e}_{M_1}}}\tag{1}\\
\Gamma &\vDash
   {\cMonoid{\cUnivSetPC{M_2}} {\cdot_{\cFunTyX{\cProdTy {M_2} {M_2}} {M_2}}} {\textsf{e}_{M_2}}}\tag{2}\\
\Gamma &\vDash
   {\cIsDefInQTyX{\textsf{h}_{\cFunTyX{M_1} {M_2}}}{\cFunQTyX{\cUnivSetPC{M_1}} {\cUnivSetPC{M_2}}}}\tag{3}\\
\Gamma &\vDash
   {\cForallX{x, y} {\cUnivSetPC{M_1}} {\textsf{h} \: (x \cdot y) = (\textsf{h} \: x) \cdot (\textsf{h} \: y)}}\tag{4}\\
\Gamma &\vDash {\textbf{A}_\cB}\tag{5}
\end{align*}
(1)~and~(2) follow similarly to the proof of $\textsf{Thm1}$; 
(3)~follows from~\cite[Axiom~5.2]{Farmer25} and parts 8 and 9 of 
Lemma~\ref{lem:univ-sets};
(4)~follows from Axiom 5 of $T$ and part 5 of Lemma~\ref{lem:univ-sets}; 
(5)~follows from (1)--(4),  Axiom 6 of $T$, and the definition of 
$\textsf{MON-HOM}$ in Table~\ref{tab:monoids-abbr}.  Therefore, $(\star)$ holds.~
\end{proof}

  \item[] $\mName{Thm27}$: ${\cTotal {\mathsf{h}_{\cFunTyX {M_1}
        {M_2}}}}$ \hfill (${\mathsf{h}_{\cFunTyX {M_1} {M_2}}}$ is
    total).

\begin{proof}[ of the theorem.]
Let $T = (L, \Gamma)$ be \textsf{MON-HOM}.  $T \vDash {\cTotal
  {\mathsf{h}_{\cFunTyX {M_1} {M_2}}}}$ follows from Axiom 5 of $T$ in
the same way that $T \vDash {\cTotal {\cdot_{\cFunTyX {\cProdTy {M}
        {M}} {M}}}}$ follows from Axiom 1 of \textsf{MON} as shown in
the proof of $\mName{Thm2}$.
\end{proof}

  \item \textsf{Thm28 (Thm26-via-MON-HOM-to-MON-4)}\\
  \hspace*{2ex} 
  $\textsf{MON-HOM}(
  {\cUnivSetPC {M}},\\
  \hspace*{15.2ex}
  {\cUnivSetPC {\cSetTy {M}}},\\
  \hspace*{15.2ex}
  {\cdot_{\cFunTyX {\cProdTy {M} {M}} {M}}},\\
  \hspace*{15.2ex}
  {\mathsf{e}_{M}},\\
  \hspace*{15.2ex}
  {\cdot_{\cFunTyX {\cProdTy {\cSetTy {M}} {\cSetTy {M}}} {\cSetTy {M}}}},\\
  \hspace*{15.2ex}
  {\mathsf{E}_{\cSetTy {M}}},\\
  \hspace*{15.2ex}
  {\mathsf{h}_{\cFunTyX {M} {\cSetTy {M}}}})$
  \hfill (example is a monoid homomorphism).

\begin{proof}[ of the theorem.] 
Similar to the proof of \textsf{Thm11}.
\end{proof}

\ee

\subsection{Development of \textnormal{\textsf{MON-over-COF}}}

\be

  \item $\mName{Def9}$: ${\cEqX {\mName{prod}_{\cFunTyCX {R} {R} {\cFunTy {R} {M}} {M}}} {}}$\\
    \hspace*{2ex}
   ${\cDefDesX {f} 
    {\cFunTyCX {Z_{\cSetTy {R}}} {Z_{\cSetTy {R}}} {\cFunTy {Z_{\cSetTy {R}}} {M}} {M}} {}}$\\
    \hspace*{4ex}${\cForallBX {m,n} {Z_{\cSetTy {R}}} {g} {\cFunTyX {Z_{\cSetTy {R}}} {M}}
    {\cQuasiEqX {\cFunAppCX {f} {m} {n} {g}} {}}}$\\
    \hspace*{6ex}${\cIf {\cBinX {m} {>} {n}} {\mathsf{e}} 
    {\cBinX {\cFunAppC {f} {m} {\cBin {n} {-} {1}} {g}} {\cdot} 
    {\cFunApp {g} {n}}}}$ \hfill (iterated product).

\begin{proof}[ that RHS is defined.]
Let
\begin{align*}
\textbf{A}_\cB =
{\cForallBX {m,n} {Z_{\cSetTy {R}}} {g} {\cFunTyX {Z_{\cSetTy {R}}} {M}}
    {\cQuasiEqX {\cFunAppCX {f} {m} {n} {g}} {}}}\\
    \hspace*{6ex}{\cIf {\cBinX {m} {>} {n}} {\mathsf{e}} 
    {\cBinX {\cFunAppC {f} {m} {\cBin {n} {-} {1}} {g}} {\cdot} 
    {\cFunApp {g} {n}}}}.
\end{align*}
Suppose that two functions $f_1$ and $f_2$ satisfy
$\textbf{A}_\cB$. It is easy to see that $f_1$ and $f_2$ must be the
same function based on the recursive structure of $f$ in
$\textbf{A}_\cB$.  Thus, $\textbf{A}_\cB$ specifies a unique function,
and so the RHS of $\mName{Def9}$ is defined by \cite[Axiom
  A6.1]{Farmer25}.
\end{proof}

  \item $\mName{Thm29}$: ${\cForallBX {m} {Z_{\cSetTy {R}}} {g}
    {\cFunTyX {Z_{\cSetTy {R}}} {M}} {\cQuasiEqX {\cProd {i} {m} {m}
        {\cFunAppX {g} {i}}}{\cFunAppX {g} {m}}}}$\\
    \phantom{x} \hfill (trivial product).

\begin{proof}[ of the theorem.]
Let $\textbf{A}_\cB$ be the theorem and $T = (L, \Gamma)$ be
$\textsf{COM-MON-over-COF}$ extended with $\mathsf{Def9}$. We must
show $(\star) {\sglsp} T \vDash \textbf{A}_\cB$.

Let $\Delta$ be the set $\mSet{{\cInX {m} {Z_{\cSetTy R}}}, {\cInX {g}
    {\cFunQTyX {Z_{\cSetTy R}} {M}}}}$.
\begin{align*}
\Gamma \cup \Delta &\vDash
   {\cIsDefX{\cVar{m} {R}}} \land
   {\cIsDefX{\cVar{g} {\cFunTyX{R} {M}}}}\tag{1}\\
\Gamma \cup \Delta &\vDash
   {\cQuasiEqX
      {\cProd {i} {m} {m} {\cFunAppX {g} {i}}}
      {\cBinX
         {\cProd {i} {m} {m-1} {\cFunAppX {g} {i}}}
         {\cdot}
         {\cFunAppX {g} {m}}}}\tag{2}\\
\Gamma \cup \Delta &\vDash
   {\cQuasiEqX
      {\cBinX
         {\cProd {i} {m} {m-1} {\cFunAppX {g} {i}}}
         {\cdot}
         {\cFunAppX {g} {m}}}
      {\cBinX
         {\mName{e}}
         {\cdot}
         {\cFunAppX {g} {m}}}}\tag{3}\\
\Gamma \cup \Delta &\vDash
   {\cQuasiEqX
      {\cBinX
         {\mName{e}}
         {\cdot}
         {\cFunAppX {g} {m}}}
      {\cFunAppX {g} {m}}}\tag{4}\\
\Gamma \cup \Delta &\vDash
   {\cQuasiEqX
      {\cProd {i} {m} {m} {\cFunAppX {g} {i}}}
      {\cFunAppX {g} {m}}}\tag{5}\\
\Gamma &\vDash
   {\textbf{A}_\cB}\tag{6}
\end{align*}
(1)~follows from variables always being defined by~\cite[Axiom A5.1]{Farmer25};
(2) and (3)~follow from (1) and $\mName{Def9}$;
(4)~follows from Axiom 20 of~$T$;
(5)~follows from (2), (3), and (4) by the Quasi-Equality Rules 
\cite[Lemma A.4]{Farmer25}; and
(6)~follows from (5) by the Deduction Theorem \cite[Theorem
  A.50]{Farmer25} and by Universal Generalization \cite[Theorem
  A.30]{Farmer25} using the fact that $m$ and $g$ are not free in
$\Gamma$ since $\Gamma$ is a set of sentences. Therefore, $(\star)$
holds.
\end{proof}

  \item $\mName{Thm30}$: 
  ${\cForallBX 
      {m,k,n} 
      {Z_{\cSetTy {R}}} 
      {g}
      {\cFunTyX {Z_{\cSetTy {R}}} {M}} {}}$\\
      \hspace*{2ex}
     ${\cImpliesX
         {\cBinBX {m} {<} {k} {<} {n}}
         {\cQuasiEqX 
           {\cBinX
              {\cProd {i} {m} {k} {\cFunAppX {g} {i}}}
              {\cdot}
              {\cProd {i} {k+1} {n} {\cFunAppX {g} {i}}}}
           {\cProdX {i} {m} {n} {\cFunAppX {g} {i}}}}}$\\
     \phantom{x} \hfill (extended iterated product).

\begin{proof}[ of the theorem.]
Let $\mathbf{A}_\cB$ be the theorem and $T = (L, \Gamma)$ be
$\textsf{MON-over-COF}$ extended by $\mName{Def9}$. We must show $(\mbox{a})
{\sglsp} T \vDash \mathbf{A}_\cB$.

Let $\Delta$ be the set 
\[\mSet{{\cInX {m} {Z_{\cSetTy {R}}}}, {\cInX
    {k} {Z_{\cSetTy {R}}}}, {\cInX {n} {Z_{\cSetTy {R}}}}, {\cBinBX
  {m} {<} {k} {<} {n}}}.\] We will prove
\[(\mbox{b}) {\sglsp} \Gamma \cup \Delta \vDash 
     {\cQuasiEqX 
        {\cBinX
           {\cProd {i} {m} {k} {\cFunAppX {g} {i}}}
           {\cdot}
           {\cProd {i} {k+1} {n} {\cFunAppX {g} {i}}}}
        {\cProd {i} {m} {n} {\cFunAppX {g} {i}}}}\]
from all $n > k$ by induction on the $n$.

\noindent
\emph{Base case: $n = k+1$.} Then:
\begin{align*}
\Gamma \cup \Delta &\vDash
   {\cIsDefX{\cVar{m} {R}}} \land
   {\cIsDefX{\cVar{k} {R}}} \land
   {\cIsDefX{\cVar{n} {R}}} \land
   {\cIsDefX{\cVar{g} {\cFunTyX{R} {M}}}}\tag{1}\\
\Gamma \cup \Delta &\vDash
   {\cQuasiEqX 
      {\cBinX
         {\cProd {i} {m} {k} {\cFunAppX {g} {i}}}
         {\cdot}
         {\cProd {i} {k+1} {n} {\cFunAppX {g} {i}}}}
      {\cBinX
         {\cProd {i} {m} {k} {\cFunAppX {g} {i}}}
         {\cdot}
         {\cFunAppX {g} {n}}}}\tag{2}\\
\Gamma \cup \Delta &\vDash
   {\cQuasiEqX 
      {\cBinX
         {\cProd {i} {m} {k} {\cFunAppX {g} {i}}}
         {\cdot}
         {\cFunAppX {g} {n}}}
      {\cProd {i} {m} {n} {\cFunAppX {g} {i}}}}\tag{3}
\end{align*}
(1)~follows from variables always being defined by~\cite[Axiom A5.1]{Farmer25};
(2)~follows from $n = k+1$ and $\mName{Thm29}$; and
(3)~follows from $n = k+1$, (1), and $\mName{Def9}$.
Thus (b) holds by the Quasi-Equality Rules \cite[Lemma A.4]{Farmer25} 
when $n = k+1$.

\noindent
\emph{Induction step: $n > k+1$ and assume} 
\[\Gamma \cup \Delta \vDash 
     {\cQuasiEqX 
        {\cBinX
           {\cProd {i} {m} {k} {\cFunAppX {g} {i}}}
           {\cdot}
           {\cProd {i} {k+1} {n-1} {\cFunAppX {g} {i}}}}
        {\cProd {i} {m} {n-1} {\cFunAppX {g} {i}}}}.\]
Then:
\begin{align*}
\Gamma \cup \Delta &\vDash
   {\cIsDefX{\cVar{m} {R}}} \land
   {\cIsDefX{\cVar{k} {R}}} \land
   {\cIsDefX{\cVar{n} {R}}} \land
   {\cIsDefX{\cVar{g} {\cFunTyX{R} {M}}}}\tag{1}\\
\Gamma \cup \Delta &\vDash
   {\cQuasiEqX 
      {\cBinX
         {\cProd {i} {m} {k} {\cFunAppX {g} {i}}}
         {\cdot}
         {\cProd {i} {k+1} {n} {\cFunAppX {g} {i}}}}
      {\cBinX   
         {\cProd {i} {m} {k} {\cFunAppX {g} {i}}}
         {\cdot}
         {\cBin
            {\cProd {i} {k+1} {n-1} {\cFunAppX {g} {i}}}
            {\cdot}
            {\cFunAppX {g} {n}}}}}\tag{2}\\
\Gamma \cup \Delta &\vDash
   {\cQuasiEqX 
      {\cBinX   
         {\cProd {i} {m} {k} {\cFunAppX {g} {i}}}
         {\cdot}
         {\cBin
            {\cProd {i} {k+1} {n-1} {\cFunAppX {g} {i}}}
            {\cdot}
            {\cFunAppX {g} {n}}}}
      {\cBinX
         {\cProd {i} {m} {n-1} {\cFunAppX {g} {i}}}
         {\cdot}
         {\cFunAppX {g} {n}}}}\tag{3}\\
\Gamma \cup \Delta &\vDash
   {\cQuasiEqX 
      {\cBinX
         {\cProd {i} {m} {n-1} {\cFunAppX {g} {i}}}
         {\cdot}
         {\cFunAppX {g} {n}}}
      {\cProd {i} {m} {n} {\cFunAppX {g} {i}}}}\tag{4}
\end{align*}
(1)~follows from variables always being defined by~\cite[Axiom A5.1]{Farmer25};
(2)~and (4)~follows from (1) and $\mName{Def9}$; and
(3)~follows from Axiom~19 of~$T$ and the induction hypothesis.
Thus (b) holds by the Quasi-Equality Rules \cite[Lemma A.4]{Farmer25} 
when $n > k+1$.

Therefore, (b) holds for all $n > k$, and (a) follows from this by the
Deduction Theorem \cite[Theorem A.50]{Farmer25} and by Universal
Generalization \cite[Theorem A.30]{Farmer25} using the fact that $m$,
$k$, $n$, and $g$ are not free in $\Gamma$ since $\Gamma$ is a set of
sentences.
\end{proof}

\ee

\subsection{Development of \textnormal{\textsf{COM-MON-over-COF}}} 

\be

  \item $\mName{Thm31}$: 
  ${\cForallBX 
      {m,n} 
      {Z_{\cSetTy {R}}} 
      {g,h}
      {\cFunTyX {Z_{\cSetTy {R}}} {M}} {}}$\\
      \hspace*{2ex}
     ${\cQuasiEqX 
        {\cBinX
           {\cProd {i} {m} {n} {\cFunAppX {g} {i}}}
           {\cdot}
           {\cProd {i} {m} {n} {\cFunAppX {h} {i}}}}
        {\cProdX 
           {i} 
           {m} 
           {n} 
           {\cBinX
              {\cFunApp {g} {i}}
              {\cdot}
              {\cFunApp {h} {i}}}}}$\\
     \phantom{x} \hfill (product of iterated products).

\begin{proof}[ of the theorem.]
Let $\mathbf{A}_\cB$ be the theorem and $T = (L, \Gamma)$ be
$\textsf{COM-MON-over-COF}$ extended by $\mName{Def9}$. We must show $(\mbox{a})
{\sglsp} T \vDash \mathbf{A}_\cB$.  

Let $\Delta$ be the set $\mSet{{\cInX {n} {Z_{\cSetTy {R}}}}, {\cInX
    {g} {\cFunQTyX {Z_{\cSetTy R}} {M}}}}$.  We will prove
\[(\mbox{b}) {\sglsp} \Gamma \cup \Delta \vDash 
     {\cQuasiEqX 
        {\cBinX
           {\cProd {i} {m} {n} {\cFunAppX {g} {i}}}
           {\cdot}
           {\cProd {i} {m} {n} {\cFunAppX {h} {i}}}}
        {\cProdX 
           {i} 
           {m} 
           {n} 
           {\cBinX
              {\cFunApp {g} {i}}
              {\cdot}
              {\cFunApp {h} {i}}}}}\]
for all $n$ by induction on the $n$.

\noindent
\emph{Base case: $n < m$.} Then:
\begin{align*}
\Gamma \cup \Delta &\vDash
   {\cIsDefX{\cVar{n} {R}}} \land
   {\cIsDefX{\cVar{g} {\cFunTyX{R} {M}}}}\tag{1}\\
\Gamma \cup \Delta &\vDash
   {\cQuasiEqX 
      {\cBinX
         {\cProd {i} {m} {n} {\cFunAppX {g} {i}}}
         {\cdot}
         {\cProd {i} {m} {n} {\cFunAppX {h} {i}}}}
      {\cBinX {\mName{e}} {\cdot} {\mName{e}}}}\tag{2}\\
\Gamma \cup \Delta &\vDash
   {\cQuasiEqX 
      {\cProdX 
         {i} 
         {m} 
         {n} 
         {\cBinX
            {\cFunApp {g} {i}}
            {\cdot}
            {\cFunApp {h} {i}}}}
      {\mName{e}}}\tag{3}
\end{align*}
(1)~follows from variables always being defined by~\cite[Axiom A5.1]{Farmer25}; and
(2)~and (3)~follow from $n < m$, (1), and $\mName{Def9}$.
Thus (b) holds by Axiom 20 of $T$ and the Quasi-Equality Rules
\cite[Lemma A.4]{Farmer25} when $n < m$.

\noindent
\emph{Induction step: $n \ge m$ and assume} 
\[\Gamma \cup \Delta \vDash 
     {\cQuasiEqX 
        {\cBinX
           {\cProd {i} {m} {n-1} {\cFunAppX {g} {i}}}
           {\cdot}
           {\cProd {i} {m} {n-1} {\cFunAppX {h} {i}}}}
        {\cProdX 
           {i} 
           {m} 
           {n-1} 
           {\cBinX
              {\cFunApp {g} {i}}
              {\cdot}
              {\cFunApp {h} {i}}}}}.\]
Then:
\begin{align*}
\Gamma \cup \Delta &\vDash
   {\cIsDefX{\cVar{n} {R}}} \land
   {\cIsDefX{\cVar{g} {\cFunTyX{R} {M}}}}\tag{1}\\
\Gamma \cup \Delta &\vDash
   {\cQuasiEqX 
      {\cBinX
         {\cProd {i} {m} {n} {\cFunAppX {g} {i}}}
         {\cdot}
         {\cProd {i} {m} {n} {\cFunAppX {h} {i}}}}
      {\cBinX
         {\cBinX
            {\cProd {i} {m} {n-1} {\cFunAppX {g} {i}}}
            {\cdot}
            {\cFunAppX {g} {n}}}
         {\cdot}
         {\cBinX
            {\cProd {i} {m} {n-1} {\cFunAppX {h} {i}}}
            {\cdot}
            {\cFunAppX {h} {n}}}}}\tag{2}\\
\Gamma \cup \Delta &\vDash
   {\cQuasiEqX 
      {\cBinX
         {\cBinX
            {\cProd {i} {m} {n-1} {\cFunAppX {g} {i}}}
            {\cdot}
            {\cFunAppX {g} {n}}}
         {\cdot}
         {\cBinX
            {\cProd {i} {m} {n-1} {\cFunAppX {h} {i}}}
            {\cdot}
            {\cFunAppX {h} {n}}}} {}}\\
&\phantom{{}\vDash{}}
    {\cBinX
       {\cBinX
          {\cProd {i} {m} {n-1} {\cFunAppX {g} {i}}}
          {\cdot}
          {\cProd {i} {m} {n-1} {\cFunAppX {h} {i}}}}
       {\cdot}
       {\cBinX
          {\cFunAppX {g} {n}}
          {\cdot}
          {\cFunAppX {h} {n}}}}\tag{3}\\
\Gamma \cup \Delta &\vDash
   {\cQuasiEqX
      {\cBinX
         {\cBinX
            {\cProd {i} {m} {n-1} {\cFunAppX {g} {i}}}
            {\cdot}
           {\cProd {i} {m} {n-1} {\cFunAppX {h} {i}}}}
         {\cdot}
         {\cBinX
            {\cFunAppX {g} {n}}
            {\cdot}
            {\cFunAppX {h} {n}}}} {}}\\
&\phantom{{}\vDash{}}
      {\cBinX
         {\cProd 
            {i} 
            {m} 
            {n-1} 
            {\cBinX
               {\cFunApp {g} {i}}
               {\cdot}
               {\cFunApp {h} {i}}}}
         {\cdot}
         {\cBin
            {\cFunAppX {g} {n}}
            {\cdot}
            {\cFunAppX {h} {n}}}}\tag{4}\\
\Gamma \cup \Delta &\vDash
   {\cQuasiEqX
       {\cBinX
          {\cProd 
             {i} 
             {m} 
             {n-1} 
             {\cBinX
                {\cFunApp {g} {i}}
                {\cdot}
                {\cFunApp {h} {i}}}}
          {\cdot}
          {\cBin
             {\cFunAppX {g} {n}}
             {\cdot}
             {\cFunAppX {h} {n}}}}
        {\cProdX 
           {i} 
           {m} 
           {n} 
           {\cBinX
              {\cFunApp {g} {i}}
              {\cdot}
              {\cFunApp {h} {i}}}}}\tag{5}       
\end{align*}
(1)~follows from variables always being defined by~\cite[Axiom A5.1]{Farmer25};
(2)~and (5)~follow from (1) and $\mName{Def9}$; 
(3)~follows from Axiom 21 of~$T$; and
(4)~follows from the induction hypothesis.
Thus (b) holds by the Quasi-Equality Rules \cite[Lemma A.4]{Farmer25} 
when $n \ge m$.

Therefore, (b) holds for all $n$, and (a) follows from this by the
Deduction Theorem \cite[Theorem A.50]{Farmer25} and by Universal
Generalization \cite[Theorem A.30]{Farmer25} using the fact that $n$
and $g$ are not free in $\Gamma$ since $\Gamma$ is a set of sentences.
\end{proof}

\ee

\subsection{Development of \textnormal{\textsf{COM-MON-ACT-over-COF}}} 

\be

  \item $\mName{Thm32}$: 
  ${\cForallBX      
     {x,y} 
     {M} 
     {s} 
     {S} 
     {\cEqX    
        {\cBinX {x} {\mName{act}} {\cBin {y} {\mName{act}} {s}}} 
        {\cBinX {y} {\mName{act}} {\cBin {x} {\mName{act}} {s}}}}}$\\
   \phantom{x} \hfill ($\mName{act}$ has commutative-like property).

\begin{proof}[ of the theorem.] 
Let $\mathbf{A}_\cB$ be the theorem and $T = (L, \Gamma)$ be
$\textsf{COM-MON-ACT-over-}$ $\textsf{COF}$. We must show $(\star)
{\sglsp} T \vDash \mathbf{A}_\cB$.
\begin{align*}
\Gamma &\vDash
    \cIsDefX{\cVar{x} {M}} \land
    \cIsDefX{\cVar{y} {M}} \land
    \cIsDefX{\cVar{s} {S}}\tag{1}\\
\Gamma &\vDash
    x \: \cFunAppX{\textsf{act}} {(y \: \cFunAppX{\textsf{act}} {s})} =
    (x \cdot y) \: \cFunAppX{\textsf{act}} {s}\tag{2}\\
\Gamma &\vDash
    y \: \cFunAppX{\textsf{act}} {(x \: \cFunAppX{\textsf{act}} {s})} =
    (y \cdot x) \: \cFunAppX{\textsf{act}} {s}\tag{3}\\    
\Gamma &\vDash
    \cBinX{x} {\cdot} {y} = \cBinX{y} {\cdot} {x}\tag{4}\\
\Gamma &\vDash
    y \: \cFunAppX{\textsf{act}} {(x \: \cFunAppX{\textsf{act}} {s})} =
    (x \cdot y) \: \cFunAppX{\textsf{act}} {s}\tag{5}\\
\Gamma &\vDash
    x \: \cFunAppX{\textsf{act}} {(y \: \cFunAppX{\textsf{act}} {s})} =
    y \: \cFunAppX{\textsf{act}} {(x \: \cFunAppX{\textsf{act}} {s})}\tag{6}\\
\Gamma &\vDash
    {\mathbf{A}_\cB}\tag{7}
\end{align*}
(1)~follows from variables always being defined by~\cite[Axiom A5.1]{Farmer25}; 
(2)~and (3) follow from (1) and Axiom 22 of $T$ by
Universal Instantiation \cite[Theorem A.14]{Farmer25}; 
(4)~follows from (1) and Axiom 21 of $T$ by Universal Instantiation; 
(5)~follows from (4) and (3) by Quasi-Equality Substitution \cite[Lemma A.2]{Farmer25}; 
(6)~follows from (2) and (5) by the Equality Rules~\cite[Lemma A.13]{Farmer25}; 
(7)~follows from (6) by Universal Generalization \cite[Theorem A.30]{Farmer25} 
using the fact that $x$, $y$, and $s$ are not free in $\Gamma$ since $\Gamma$ is a set of
sentences. Therefore, $(\star)$ holds.~
\end{proof}

\ee

\subsection{Development of \textnormal{\textsf{STR}}} 

\be

  \item $\mName{Def10}$:
  ${\cEqX
     {\mName{str}_{\cSetTy {\cFunTyX {R} {A}}}}
     {\cSeqFinQTy {A}}}$
  \hfill (string quasitype).

\begin{proof}[ that RHS is defined.] 
Let $T$ be the top theory of \textsf{STR-1}.  We must show $(\star)$
$T \vDash \cIsDefX{\cSeqFinQTy {A}}$. Now $\cSeqFinQTy{A}$ stands for
\[\cSet{s}
{\cSeqQTy{A}} {\cForsomeX{n} {\textbf{C}_{\cSetTy{R}}^N} {\cForallX{m}
    {\textbf{C}_{\cSetTy{R}}^N} {\cIsDefX{(s \: m)} \Leftrightarrow
      \textbf{C}_{\cFunTyX{A} {\cFunTyX{A} {\cB}}}m \: n}}}\] based on
the notational definitions in Table~\ref{tab:nd-sequences}.  Thus
$(\star)$ holds because function abstractions are always defined
by~\cite[Axiom A5.11]{Farmer25}.
\end{proof}

  \item $\mName{Def11}$:
  ${\cEqX
      {\epsilon_{\cFunTyX {R} {A}}}
      {\cEmpListPC {R} {A}}}$ \hfill (empty string).

\begin{proof}[ that RHS is defined.] 
Let $T$ be the top theory of \textsf{STR-1}.  We must show $(\star)$
$T \vDash \cIsDefX{{\cEmpListPC {R} {A}}}$. Now $\cEmpListPC {R} {A}$
stands for
\[\cFunAbsX{x} {R} {\cBotPC{A}}\]
based on the notational definitions in Tables~\ref{tab:nd-definedness}
and~\ref{tab:nd-sequences}.  Thus $(\star)$ holds because function
abstractions are always defined by~\cite[Axiom A5.11]{Farmer25}.~
\end{proof}

  \item[] $\mName{Def12}$:
  ${\cEqX
      {\mName{cat}_{\cFunTyX {\cProdTy {\cFunTy {R} {A}} {\cFunTy {R} {A}}} {\cFunTy {R} {A}}}} 
      {\cAppendPC {R} {A}}}$\\
  \phantom{x} \hfill (concatenation).

\begin{proof}[ that RHS is defined.] 
Let $T$ be the top theory of \textsf{STR-1}.  We must show
\[(\star) {\sglsp} T \vDash \cIsDefX{{\cAppendPC {R} {A}}}.\]
The pseudoconstant ${\cAppendPC {\alpha} {\beta}}$ is defined in
Table~\ref{tab:nd-sequences}.  For all $\alpha$ and $\beta$,
${\cAppendPC {\alpha} {\beta}}$ denotes the concatenation function for
finite sequences over the denotation of $\beta$.  Therefore, $(\star)$
holds.
\end{proof}

  \item $\mName{Thm33}$:
  ${\cForallX 
      {x} 
      {\mName{str}}
      {\cBinBX 
         {\cCatAppX {\epsilon} {x}}
         {=}
         {\cCatAppX {x} {\epsilon}}
         {=}
         {x}}}$
    \hfill ($\epsilon$ is an identity element). 

\begin{proof}[ of the theorem.]
Let $T = (L, \Gamma)$ be the top theory of $\textsf{STR-1}$ extended
by \textsf{Def10--Def12}.  We must show:
\bi

  \item[] (a)
  $T \vDash 
   {\cForallX 
      {x} 
      {\mName{str}}
      {\cEqX
         {\cCatAppX {\epsilon} {x}}
         {x}}}.$

  \item[] (b)
  $T \vDash 
   {\cForallX 
      {x} 
      {\mName{str}}
      {\cEqX
         {\cCatAppX {x} {\epsilon}}
         {x}}}.$

\ei

Let $\Delta$ be the set $\mSet{\cInX {x} {\mName{str}}}$. Then:
\begin{align*}
\Gamma \cup \Delta &\vDash
  {\cEqX 
     {\cCatAppX {\epsilon} {x}}
     {x}}\tag{1}\\
\Gamma &\vDash
  {\cForallX 
     {x} 
     {\mName{str}}
     {\cEqX {\cCatAppX {\epsilon} {x}} {x}}}\tag{2}
\end{align*}
(1)~follows from $\cInX {x} {\mName{str}}$ and $\mathsf{Def12}$; and
(2)~follows from (1) by the Deduction Theorem \cite[Theorem
  A.50]{Farmer25} and then by Universal Generalization \cite[Theorem
  A.30]{Farmer25} using the fact that $\cVar{x}{\cFunTyX{R}{A}}$ is
not free in $\Gamma$ since $\Gamma$ is a set of sentences.  
Therefore, (a) holds.

We will prove (c) $\Gamma \cup \Delta \vDash {\cEqX {\cCatAppX {x} {\epsilon}}
  {x}}$ by induction on the length of $x$.

\noindent
\emph{Base case: $x$ is $\epsilon$.} Then $\Gamma \cup \Delta \vDash
     {\cEqX {\cCatAppX {\epsilon} {\epsilon}} {\epsilon}}$ is an
     instance of (1) above.

\noindent
\emph{Induction step: $x$ is $\cCons {a} {y}$ and assume $\Gamma \cup \Delta
\vDash {\cEqX {\cCatAppX {y} {\epsilon}} {y}}$.}  Then:
\begin{align*}
\Gamma \cup \Delta &\vDash
  {\cEqX 
     {\cCatAppX {\cCons {a} {y}} {\epsilon}}
     {\cCons {a} {\cCatAppX {y} {\epsilon}}}}\tag{1}\\
\Gamma \cup \Delta &\vDash
  {\cEqX 
     {\cCons {a} {\cCatAppX {y} {\epsilon}}}
     {\cCons {a} {y}}}\tag{2}
\end{align*}
(1) follows from $\cInX {x} {\mName{str}}$ and $\mathsf{Def12}$; and
(2) follows from the induction hypothesis and (1) by Quasi-Equality
Substitution \cite[Lemma A.2]{Farmer25}.  Thus $\Gamma \cup \Delta
\vDash {\cEqX {\cCatAppX {\cCons {a} {y}} {\epsilon}} {\cCons {a}
    {y}}}$ holds by the Equality Rules~\cite[Lemma A.13]{Farmer25}.

Therefore (c) holds, and (b) follows from (c) by the Deduction Theorem
\cite[Theorem A.50]{Farmer25} and then by Universal Generalization
\cite[Theorem A.30]{Farmer25} using the fact that
$\cVar{x}{\cFunTyX{R}{A}}$ is not free in $\Gamma$ since $\Gamma$ is a
set of sentences.  
\end{proof}

  \item $\mName{Thm34}$:
  ${\cForallX 
      {x,y,z} 
      {\mName{str}}
      {\cEqX
         {\cCatAppX {x} {\cCatApp {y} {z}}}
         {\cCatAppX {\cCatApp {x} {y}} {z}}}}$
     \hfill ($\mName{cat}$ is associative).

\begin{proof}[ of the theorem.]
Let $T = (L, \Gamma)$ be the top theory of $\textsf{STR-1}$ extended
by \textsf{Def10--Def12}.  We must show
\[\mbox{(a)} {\sglsp} 
T \vDash
  {\cForallX 
     {x,y,z} 
     {\mName{str}}
     {\cEqX
        {\cCatAppX {x} {\cCatApp {y} {z}}}
        {\cCatAppX {\cCatApp {x} {y}} {z}}}}.
\]
Let $\Delta$ be the set $\mSet{\cInX {x} {\mName{str}}, \cInX
{y} {\mName{str}}, \cInX {z} {\mName{str}}}$.  We will prove 
\[\mbox{(b)} {\sglsp} \Gamma \cup \Delta \vDash 
{\cEqX
   {\cCatAppX {x} {\cCatApp {y} {z}}}
   {\cCatAppX {\cCatApp {x} {y}} {z}}}
\]
by induction on the length of $x$.

\noindent
\emph{Base case: $x$ is $\epsilon$.} Then:
\begin{align*}
\Gamma \cup \Delta &\vDash 
   {\cEqX
      {\cCatAppX {\epsilon} {\cCatApp {y} {z}}}
      {\cCatApp {y} {z}}}\tag{1}\\
\Gamma \cup \Delta &\vDash 
   {\cEqX
      {\cCatApp {y} {z}}
      {\cCatAppX {\cCatApp {\epsilon} {y}} {z}}}\tag{2}
\end{align*}
(1) and (2) follow from $\mathsf{Thm33}$.  Thus $\Gamma \cup \Delta
\vDash {\cEqX {\cCatAppX {\epsilon} {\cCatApp {y} {z}}} {\cCatAppX
    {\cCatApp {\epsilon} {y}} {z}}}$ holds by the Equality
Rules~\cite[Lemma A.13]{Farmer25}.

\noindent
\emph{Induction step: $x$ is $\cCons {a} {w}$ and assume $\Gamma \cup \Delta
\vDash {\cEqX {\cCatAppX {w} {\cCatApp {y} {z}}} {\cCatAppX
    {\cCatApp {w} {y}} {z}}}$.}  Then:
\begin{align*}
\Gamma \cup \Delta &\vDash 
   {\cEqX
      {\cCatAppX {\cCons {a} {w}} {\cCatApp {y} {z}}}
      {\cConsX {a} {\cCatAppX {w} {\cCatApp {y} {z}}}}}\tag{1}\\
\Gamma \cup \Delta &\vDash 
   {\cEqX
      {\cConsX {a} {\cCatAppX {w} {\cCatApp {y} {z}}}}
      {\cConsX {a} {\cCatAppX {\cCatApp {w} {y}} {z}}}}\tag{2}\\
\Gamma \cup \Delta &\vDash 
   {\cEqX
      {\cConsX {a} {\cCatAppX {\cCatApp {w} {y}} {z}}}
      {\cCatAppX {\cCons {a} {\cCatAppX {w} {y}}} {z}}}\tag{3}\\
\Gamma \cup \Delta &\vDash 
   {\cEqX
      {\cCatAppX {\cCons {a} {\cCatAppX {w} {y}}} {z}}
      {\cCatAppX {\cCatApp {\cCons {a} {w}} {y}} {z}}}\tag{4}      
\end{align*}
(1), (3), and (4) follow from $\cInX {x} {\mName{str}}$, $\cInX {y}
    {\mName{str}}$, and $\cInX {z} {\mName{str}}$ and
    $\mathsf{Def12}$; and (2)~follows from the induction hypothesis.
    Thus \[\Gamma \cup \Delta \vDash {\cEqX {\cCatAppX {\cCons {a} {w}} {\cCatApp
          {y} {z}}} {\cCatAppX {\cCatApp {\cCons {a} {w}} {y}} {z}}}\]
    holds by the Equality Rules~\cite[Lemma A.13]{Farmer25}.

Therefore (b) holds, and (a) follows from (b) by the Deduction Theorem
\cite[Theorem A.50]{Farmer25} and then by Universal Generalization
\cite[Theorem A.30]{Farmer25} using the fact that
$\cVar{x}{\cFunTyX{R}{A}}$, $\cVar{y}{\cFunTyX{R}{A}}$, and
$\cVar{z}{\cFunTyX{R}{A}}$ are not free in $\Gamma$ since $\Gamma$ is
a set of sentences.
\end{proof}

  \item \textsf{Thm35 (Thm1-via-MON-over-COF-to-STR-2)}:\\
  \hspace*{2ex}
  ${\cMonoid
      {\mName{str}_{\cSetTy {\cFunTyX {R} {A}}}}
      {\mName{cat}_{\cFunTyX {\cProdTy {\cFunTy {R} {A}} {\cFunTy {R} {A}}} {\cFunTy {R} {A}}}}
      {\epsilon_{\cFunTyX {R} {A}}}}$\\
  \phantom{x} \hfill (strings form a monoid).

\begin{proof}[ of the theorem.] 
Similar to the proof of \textsf{Thm11}.
\end{proof}

  \item \textsf{Def13 (Def3-via-MON-over-COF-to-STR-2)}:\\
  \hspace*{2ex} 
  ${\cEqX 
      {\textsf{set-cat}_{\cFunTyX {\cProdTy {\cSetTy {\cFunTyX {R} {A}}} {\cSetTy {\cFunTyX {R} {A}}}} {\cSetTy {\cFunTyX {R} {A}}}}} {}}$\\
      \hspace*{2ex}
     ${\cFunAppX
         {\textsf{set-op}_{\cFunTyX {\cFunTy {\cProdTy {\cFunTy {R} {A}} {\cFunTy {R} {A}}} {\cFunTy {R} {A}}} {\cFunTy {\cProdTy {\cSetTy {\cFunTyX {R} {A}}} {\cSetTy {\cFunTyX {R} {A}}}} {\cSetTy {\cFunTyX {R} {A}}}}}} 
         {\mName{cat}}}$\\
  \phantom{x} \hfill (set concatenation).

\begin{proof}[ that RHS is defined.] 
Similar to the proof that the RHS of \textsf{Def6} is defined.
\end{proof} 

  \item \textsf{Def14 (Def4-via-MON-over-COF-to-STR-2)}:\\
    \hspace*{2ex}
    ${\cEqX {\mathsf{E}_{\cSetTy {{\cFunTyX {R} {A}}}}}
    {\cFinSetL {\epsilon_{\cFunTyX {R} {A}}}}}$ \hfill (set identity element). 

\begin{proof}[ that RHS is defined.] 
Similar to the proof that the RHS of \textsf{Def6} is defined.
\end{proof}

  \item \textsf{Thm36 (Thm12-via-MON-over-COF-1-to-STR-2)}:\\
  \hspace*{2ex}
  ${\cMonoid 
      {\cSetQTy {\mName{str}_{\cSetTy {\cFunTyX {R} {A}}}}} 
      {\textsf{set-cat}_{\cFunTyX {\cProdTy {\cSetTy {\cFunTyX {R} {A}}} {\cSetTy {\cFunTyX {R} {A}}}} {\cSetTy {\cFunTyX {R} {A}}}}} 
      {\mathsf{E}_{\cSetTy {{\cFunTyX {R} {A}}}}}}$\\
  \phantom{x} \hfill (string sets form a monoid).

\begin{proof}[ of the theorem.] 
Similar to the proof of \textsf{Thm11}.
\end{proof}

  \item \textsf{Def15 (Def9-via-MON-over-COF-1-to-STR-2)}:\\
  \hspace*{2ex} 
  ${\cEqX {\textsf{iter-cat}_{\cFunTyCX {R} {R} {\cFunTy {R} {\cFunTy {R} {A}}} {\cFunTy {R} {A}}}} {}}$\\
    \hspace*{2ex}
   ${\cDefDesX {f} 
    {\cFunTyCX {Z_{\cSetTy {R}}} {Z_{\cSetTy {R}}} {\cFunTy {Z_{\cSetTy {R}}} {\cFunTy {R} {A}}} {\cFunTy {R} {A}}} {}}$\\
    \hspace*{4ex}${\cForallBX {m,n} {Z_{\cSetTy {R}}} {g} {\cFunTyX {Z_{\cSetTy {R}}} {\cFunTy {R} {A}}}
    {\cQuasiEqX {\cFunAppCX {f} {m} {n} {g}} {}}}$\\
    \hspace*{6ex}${\cIf {\cBinX {m} {>} {n}} {\epsilon} 
    {\cBinX {\cFunAppC {f} {m} {\cBin {n} {-} {1}} {g}} {\mathsf{cat}} 
    {\cFunApp {g} {n}}}}$\\
  \phantom{x} \hfill (iterated concatenation).

\begin{proof}[ that RHS is defined.] 
Similar to the proof that the RHS of \textsf{Def6} is defined.
\end{proof}
\ee

\section{Miscellaneous Theorems}\label{app:misc-thms}

\begin{lem}[Universal Sets]\label{lem:univ-sets}
The following formulas are valid:

\be

  \item ${\cIsDefX {\cUnivSetPC {\alpha}}}$.

  \item ${\cNotEqX {\cUnivSetPC {\alpha}} {\cEmpSetPC {\alpha}}}$.

  \item ${\cForallX {x} {\alpha} {\cInX {x} {\cUnivSetPC {\alpha}}}}$.

  \item ${\cEqX
            {\cFunAbs {\mathbf{x}} {\alpha} {\mathbf{B}_\beta}}
            {\cFunAbsQTy {\mathbf{x}} {\cUnivSetPC {\alpha}} {\mathbf{B}_\beta}}}$.

  \item ${\cIffX
            {\cForall {\mathbf{x}} {\alpha} {\mathbf{B}_\cB}}
            {\cForallQTy {\mathbf{x}} {\cUnivSetPC {\alpha}} {\mathbf{B}_\cB}}}$.

  \item ${\cIffX
            {\cForsome {\mathbf{x}} {\alpha} {\mathbf{B}_\cB}}
            {\cForsomeQTy {\mathbf{x}} {\cUnivSetPC {\alpha}} {\mathbf{B}_\cB}}}$.

  \item ${\cQuasiEqX
            {\cDefDes {\mathbf{x}} {\alpha} {\mathbf{B}_\cB}}
            {\cDefDesQTy {\mathbf{x}} {\cUnivSetPC {\alpha}} {\mathbf{B}_\cB}}}$.

  \item ${\cIffX 
            {\cIsDefX {\mathbf{A}_\alpha}} 
            {\cIsDefInQTy {\mathbf{A}_\alpha} {\cUnivSetPC {\alpha}}}}$

  \item ${\cEqX
            {\cUnivSetPC {\cFunTyX {\alpha} {\beta}}}
            {\cFunQTy {\cUnivSetPC {\alpha}} {\cUnivSetPC {\beta}}}}$.

  \item ${\cEqX
            {\cUnivSetPC {\cProdTyX {\alpha} {\beta}}}
            {\cProdQTy {\cUnivSetPC {\alpha}} {\cUnivSetPC {\beta}}}}$.

  \item ${\cEqX
            {\cUnivSetPC {\cSetTy {\alpha}}}
            {\cSetQTy {\cUnivSetPC {\alpha}}}}$.

  \item ${\cEqX
            {\textbf{A}_{\cFunTyX {\cProdTy {\alpha} {\beta}} {\gamma}}}
            {\cRestrictX
               {\textbf{A}_{\cFunTyX {\cProdTy {\alpha} {\beta}} {\gamma}}}
               {\cProdTyX {\cUnivSetPC {\alpha}} {\cUnivSetPC {\beta}}}}}$.

\ee
\end{lem}

\begin{proof}
The proof is left to the reader as an exercise.
\end{proof}

\begin{lem}\label{lem:cdot-opposite}
Let $T$ be $\mathsf{MON}$ extended by the definition $\mathsf{Def2}$.
The formula
\[{\cQuasiEqX 
    {\cBinX {\textbf{A}_M} {\cdot^{\rm op}} {\textbf{B}_M}}
    {\cBinX {\textbf{B}_M} {\cdot} {\textbf{A}_M}}}\]
is valid in $T$.
\end{lem}

\begin{proof}
Let $\textbf{X}_\cB$ be 
\[{\cQuasiEqX 
    {\cBinX {\textbf{A}_M} {\cdot^{\rm op}} {\textbf{B}_M}}
    {\cBinX {\textbf{B}_M} {\cdot} {\textbf{A}_M}}},\]
$N$ be a model of $T$, and $\phi \in \mName{assign}(N)$.  Suppose that
$V^{N}_{\phi}(\textbf{A}_M)$ or $V^{N}_{\phi}(\textbf{B}_M)$ is
undefined.  Then clearly $V^{N}_{\phi}(\textbf{X}_\cB) = \TRUE$.  Now
suppose that $V^{N}_{\phi}(\textbf{A}_M)$ and
$V^{N}_{\phi}(\textbf{B}_M)$ are defined.  Then
$V^{N}_{\phi}(\textbf{X}_\cB) = \TRUE$ by $\mathsf{Def2}$.
\end{proof}

\addcontentsline{toc}{section}{References}
\bibliography{simple-type-theory-bib,monoids-bib}
\bibliographystyle{plain}

\end{document}